\documentclass[12pt,onecolumn,a4paper]{article}
\usepackage{amsmath,amssymb}
\usepackage[OT2,T1]{fontenc} 
\usepackage{tabularx}
\usepackage{tikz-cd}
\usepackage{verbatim}
\usepackage{subcaption}
\usetikzlibrary{external}
\usepackage{algorithm}
\usepackage{algpseudocode}
\usepackage{amsmath}
\usepackage{authblk}

\usepackage[braket, qm]{qcircuit}
\usepackage{enumerate}

\usepackage{ulem}
\usepackage{mathtools}
\mathtoolsset{showonlyrefs}
\usepackage{amsmath}
\usepackage{graphicx}
\usepackage{enumerate}
\usepackage{amssymb}

\PassOptionsToPackage{table}{xcolor} 
\usepackage{xcolor}
\usepackage{colortbl} 
\usepackage{array}
\usepackage{booktabs,caption,fixltx2e}

\usepackage{graphicx}
\usepackage{mdframed}
\usepackage{cite}

\usepackage{amsthm,amsmath}
\usepackage[left=0.75in,right=0.75in,top=0.75in,bottom=0.85in]{geometry}

\usepackage{thm-restate} 
\usepackage{multirow}
\usepackage{makecell}

\newcommand{\appref}[1]{\hyperref[#1]{{Appendix~\ref*{#1}}}}
\newcommand{\be}{\begin{eqnarray} \begin{aligned}}
\newcommand{\ee}{\end{aligned} \end{eqnarray} }
\newcommand{\benn}{\begin{eqnarray*} \begin{aligned}}
\newcommand{\eenn}{\end{aligned} \end{eqnarray*}}
\newcommand*{\textfrac}[2]{{{#1}/{#2}}}
\newcommand*{\bbN}{\mathbb{N}}
\newcommand*{\bbR}{\mathbb{R}}
\newcommand*{\bbC}{\mathbb{C}}
 
\newcommand*{\cB}{\mathcal{B}}

\newcommand*{\cG}{\mathcal{G}}
\newcommand*{\cH}{\mathcal{H}}

\newcommand*{\cK}{\mathcal{K}}
\newcommand*{\cL}{\mathcal{L}}

\newcommand*{\cN}{\mathcal{N}}

\newcommand*{\cU}{\mathcal{U}}

\newcommand*{\cW}{\mathcal{W}}

\newcommand{\bc}{\begin{center}}
\newcommand{\ec}{\end{center}}

\newtheorem{theorem}{Theorem}[section]
\newtheorem{lemma}[theorem]{Lemma}

\newtheorem{definition}[theorem]{Definition}

\newtheorem{corollary}[theorem]{Corollary}

\usepackage{amsfonts}

\def\01{\{0,1\}}
\newcommand{\ceil}[1]{\lceil{#1}\rceil}

\newcommand{\proj}[1]{|#1\rangle\langle#1|}

\newcommand*{\gkp}{\mathsf{GKP}}
\newcommand{\gkpcode}[3]{\cal{GKP}_{#1}^{#2}[#3]}

\newcommand{\gkpcoderect}[3]{\Sha\cal{GKP}_{#1}^{#2}[#3]}

\newcommand{\LSB}[1]{\textrm{LSB}_{#1}}
\newcommand*{\lsb}[1]{\mathsf{LSB}_{#1}}

\newcommand{\base}[2]{[{#2}]_{#1}}

\newcommand*{\cZ}{\mathsf{C}Z}
\newcommand*{\cX}{\mathsf{C}X}

\newcommand*{\Pgate}{\mathsf{P}}
\newcommand*{\Fgate}{\mathsf{F}}

\newcommand{\vertiii}[1]{{\left\vert\kern-0.25ex\left\vert\kern-0.25ex\left\vert #1 
    \right\vert\kern-0.25ex\right\vert\kern-0.25ex\right\vert}}

\newcommand*{\encmap}{\mathsf{Enc}} 
\newcommand*{\decmap}{\mathsf{Dec}}
\newcommand*{\encmapgkp}{\mathsf{gkpEnc}} 
\newcommand*{\decmapgkp}{\mathsf{gkpDec}}
\newcommand{\encodergkp}[3]{\mathsf{gkpEnc}_{#1}^{#2}[#3]}

\newcommand*{\cHin}{\cH_{in}}
\newcommand*{\cHout}{\cH_{out}}

\newcommand*{\cLin}{{\cL_{in}}}
\newcommand*{\cLout}{{\cL_{out}}}

\newcommand{\encoded}[1]{{\mkern1mu \bf #1 \mkern1mu}}
\newcommand{\encodedC}[1]{\boldsymbol{\mathcal{#1}}}
\renewcommand{\cal}[1]{\mathcal{#1}}
\newcommand{\bb}[1]{\mathbb{#1}}

\newcommand{\blue}[1]{{\color{black}#1}}

\newcommand*{\embedmap}{\mathsf{Embed}}
\renewcommand*{\LSB}{\mathsf{LSB}}

\newcommand*{\gateerror}{\mathsf{err}}

\newcommand*{\CZ}{\mathsf{C}Z}

\newcommand*{\qCX}[1]{\mathsf{C}X_{#1}}
\newcommand*{\bittransfer}[2]{\mathsf{Transf}_{#1}^{#2}}

\newcommand*{\vac}{\mathsf{vac}}
\newcommand*{\Uelem}{\cal{U}_{\mathsf{elem}}}
\renewcommand{\ctrl}{\mathsf{ctrl}}
\DeclareSymbolFont{cyrletters}{OT2}{wncyr}{m}{n}
\DeclareMathSymbol{\Sha}{\mathalpha}{cyrletters}{"58}



\usepackage{hyperref}
\hypersetup{
   colorlinks,
   citecolor=blue,
   filecolor=blue,
   linkcolor=blue,
   urlcolor=blue
}

\begin{document}

\title{
Qubit-oscillator-based gate implementations \\
for approximate Gottesman-Kitaev-Preskill codes
}
 \author{Lukas Brenner}
 \author{Beatriz Dias}
 \author{Robert Koenig}
 \affil{Department of Mathematics, School of Computation, Information and Technology,\\ Technical University of Munich, 85748 Garching, Germany}
 \affil{Munich Center for Quantum Science and Technology, 80799 Munich, Germany}
 \date{\today}

\maketitle

\begin{abstract}
We consider hybrid qubit-oscillator systems together with Gaussian, multi-qubit as well as  qubit-controlled Gaussian unitaries. We propose implementations of logical gates for approximate Gottesman-Kitaev-Preskill codes in this model using two oscillators and three qubits. We show that these gate implementations become exact in the limit of large squeezing: The logical gate error is upper bounded by a linear function of the squeezing parameter, and depends polynomially on the number of encoded qubits.  For certain Cliffords, our constructions overcome a shortcoming of well-known Gaussian implementations which have a constant logical gate error even in the absence of noise. 
 \end{abstract}

\tableofcontents

\section{Introduction} \label{sec:intro}

Linear optics offers an efficient and experimentally accessible framework for continuous-variable quantum information, yet its Gaussian structure imposes fundamental limits: it cannot realize error correction, entanglement distillation or universal fault-tolerant computation. Gottesman-Kitaev-Preskill (GKP) codes appear to circumvent these obstacles, but only in their idealized unphysical form. 
For realistic approximate GKP codes, linear optics implementations of key logical gates incur constant error, making them unsuitable for fault-tolerant use \cite{cliffordslinearoptics2025}. This no-go result motivates our proposal: The use of non-Gaussian resources in hybrid qubit-oscillator systems to achieve accurate implementations of logical operators for physically realistic GKP codes.

\paragraph{Linear optics and its limitations.} 
Gaussian linear optics provides a versatile toolbox for continuous-variable (CV) quantum information processing and is a fertile playground for both theoretical and experimental studies. Mathematically, the corresponding states and operations have efficient descriptions, enabling the analytical study of a diverse range of quantum phenomena and processes ranging from entanglement and Bell inequality violations~\cite{PhysRevLett.82.2009,PhysRevA.67.012105,Etesse_2014} to measurement-based quantum information processing~\cite{PhysRevLett.97.110501,PhysRevA.79.062318}. Physically, the corresponding operations are typically easy to realize in the lab. In particular, the subgroup of Gaussian unitaries is generated by a number of elementary building blocks, each of which has a physical realization in terms of a linear optics element such as a mirror.

Let us briefly specify these operations in more detail. A system made of~$n$ harmonic oscillators (or modes) is associated with the Hilbert space~$L^2(\mathbb{R})^{\otimes n}$.
The collection of canonical mode operators will be denoted as~$R=(Q_1,P_1,\ldots,Q_n,P_n)$. 
 They satisfy the canonical commutation relations
\begin{align}
\left[R_k, R_{\ell}\right]=i J_{k, \ell} I_{\cH_n}\quad \text { where } \quad J=\left(\begin{array}{cc}
0 & 1 \\
-1 & 0
\end{array}\right)^{\oplus n}\  .
\end{align}
Of particular interest are the Gaussian (Weyl) displacement operators
\begin{align} \label{eq:displacementrelationsetup}
W(\xi)=e^{-i\xi\cdot JR}\qquad\textrm{ where }\qquad \xi\in \mathbb{R}^{2n}\ .
\end{align}
We call such a displacement of bounded strength 
if~$\|\xi\|$ is bounded by a constant (independent of e.g., the number~$n$ of modes).
A displacement operator~$W(\xi)$ is a one-mode operator if and only if~$\xi\cdot JR$ is a linear combination of the mode operators~$Q_j$ and~$P_j$ only, for some~$j\in\{1,\ldots,n\}$. We also consider Gaussian unitaries of the form
\begin{align}
U(A)=e^{iR\cdot AR}\qquad\textrm{ where }\qquad A=A^T\in\mathsf{Mat}_{2n\times 2n}(\mathbb{R})\ ,
\end{align}
i.e., unitaries generated by Hamiltonians which are quadratic in the mode operators. 
Again, we say that~$U(A)$  is bounded-strength if the operator norm~$\|A\|$ is bounded by a constant, and~$U(A)$ is called a one- respectively two-mode operation if the quadratic form~$R\cdot AR$ only involves operators~$Q_j,P_j$ for a single mode~$j\in \{1,\ldots,n\}$ respectively~$Q_j,P_j,Q_k,P_k$ for two modes~$j,k\in \{1,\ldots,n\}$.  

The set of linear optics operations consists of 
\begin{enumerate}[(i)]
\item\label{it:firstlinearoptics}
preparation of the single-mode vacuum state~$\ket{\vac} \in L^2(\bbR)$
defined as~$\vac(x)= \pi^{-1/4} \cdot e^{-x^2/2}$ on any mode, 
\item \label{it:secondlinearoptics}
application of any one- or two-mode Gaussian unitary of bounded strength, and
\item\label{it:lastlinearoptics}
 single-mode heterodyne or homodyne detection.
 \end{enumerate}
 Any dynamics consisting of a sequence of polynomially-many basic linear optics operations of the type~\eqref{it:firstlinearoptics}--\eqref{it:lastlinearoptics} has an efficient description in the covariance matrix formalism (see e.g.,~\cite{RevModPhys.84.621}), and is typically easily realizable in practice.  Unfortunately, however, this set of operations is fundamentally limited. In particular, it was shown early on that Gaussian operations cannot protect encoded information even against relatively simple Gaussian noise~\cite{PhysRevLett.102.120501,PhysRevA.99.032344}. In other words, suitable quantum error-correcting codes for CV systems need to involve non-Gaussian states. In addition, it was demonstrated that it is not possible to perform entanglement distillation from Gaussian states and entanglement swapping by means of linear optics operations~\cite{PhysRevLett.89.137903,PhysRevLett.89.137904,PhysRevA.66.032316}.
 
 \paragraph{Ideal Gottesman-Kitaev-Preskill codes and linear optics.} First introduced in~\cite{gkp}, GKP codes are a CV analog of stabilizer codes for qubits, and are in many ways the de facto standard for quantum error correction using CV systems.  The ideal GKP code encodes a qudit into formal linear combinations of position-eigenstates of an oscillator, i.e., Dirac-$\delta$-distributions. In more detail, let~$(Q,P)$ denote the canonical position- and momentum operators on~$L^2(\mathbb{R})$. Let~$d\geq 2$ be an integer. The codewords~$\{\gkp(j)_d\}_{j\in\{0,\dots, d-1\}}$ of the ideal GKP code are then defined in position-space as  
\begin{align} \label{eq: def ideal gkp}
\gkp(j)_d(x)\propto \sum_{s\in\mathbb{Z}} \delta\left(x-\sqrt{\frac{2\pi}{d}}(j+sd)\right)\quad\textrm{ for }\quad x \in \bbR \quad \text{ where } \quad j\in\{0,\dots, d-1\}\ .
\end{align}
 The states (distributions)~$\{\gkp(j)_d\}_{j\in\{0,\dots, d-1\}}$ form a basis of the (ideal) GKP code~$\gkpcode{}{}{d}$. It is the space stabilized by the two stabilizer generators
\begin{align}
S^{(1)}_d=e^{i\sqrt{2\pi d}Q}\qquad\textrm{ and }\qquad S^{(2)}_d=e^{-i\sqrt{2\pi d}P}\ ,
\end{align}
which are both displacement operators. This fact means for example that syndrome extraction circuits can be realized by linear optics circuits (see~\cite{gkp}). Furthermore, as discussed in~\cite{gkp}, there is a set of logical Clifford group generators  which can be implemented by Gaussian unitaries.
Thus, although the GKP ``states''  defined by Eq.~\eqref{eq: def ideal gkp} themselves are non-Gaussian, it seems that many operations of interest associated with fault-tolerant computation with GKP codes can be realized using linear optics.

\paragraph{Approximate symmetrically squeezed Gottesman-Kitaev-Preskill codes.}
Any physically realistic  implementation of GKP codes must use approximate, i.e., finitely squeezed, states instead of the formal superpositions given in Eq.~\eqref{eq: def ideal gkp}. In our analysis, 
we use what we call~$\varepsilon$-truncated approximate GKP states with squeezing parameters~$\kappa,\Delta>0$. The corresponding family of states (normalized functions) 
$\{\gkp^\varepsilon_{\kappa,\Delta}(j)_d\}_{j\in\{0,\ldots,d-1\}}\subset L^2(\mathbb{R})$ 
is defined as 
\begin{align}
    \gkp^\varepsilon_{\kappa,\Delta}(j)_d(x) \propto 
    \begin{cases}
    \sum_{s \in \mathbb{Z}} e^{-\kappa^2 s^2/2} \cdot e^{-(x - \sqrt{2\pi/d}j)^2/(4\pi d \Delta^2)} &\textrm{ if } \mathsf{dist}\left(\frac{x - \sqrt{2\pi/d}j}{\sqrt{2\pi d}}, \mathbb{Z}\right) \le \varepsilon\\
    0 &\textrm{ otherwise}
    \end{cases}\label{eq:truncatedvarsvdkap}
\end{align} 
for~$x\in\bbR$ and~$j\in \{0,\ldots,d-1\}$. Here the truncation parameter~$\varepsilon>0$ ensures that each of these functions only has support on intervals of width~$2\sqrt{2\pi d}\varepsilon$ centered around integer multiples of~$\sqrt{2\pi d}$ (shifted by~$\sqrt{2\pi/ d}j$).
In particular, for any truncation parameter~$\varepsilon>0$ satisfying
\begin{align}
\varepsilon\le 1/(2d)\ ,\label{eq:varepsilonchoice}
\end{align}
the states~$\{\gkp^\varepsilon_{\kappa,\Delta}(j)_d\}_{j\in\{0,\ldots,d-1\}}$ are pairwise orthogonal. We typically choose~$\varepsilon=\varepsilon_d=1/(2d)$ and call this the optimal truncation parameter. As discussed in~\cite{cliffordslinearoptics2025}, this choice guarantees orthogonality while also ensuring that these states are close to the untruncated approximate GKP states commonly considered for large squeezing (see~\cite{gkp}), a two-parameter family~$\{\gkp_{\kappa,\Delta}(j)_d\}_{j\in\{0,\ldots,d-1\}}$ of states only depending on the squeezing parameters~$(\kappa,\Delta)$. In other words, the truncation introduced in Eq.~\eqref{eq:truncatedvarsvdkap} should be considered as a technically convenient proxy for the actual approximate GKP states used in any protocol. 

In addition to working with truncated approximate GKP states with the optimal truncation parameter~$\varepsilon_d$, we choose the following linear dependence between the parameters~$\kappa$ (determining the envelope) and~$\Delta$ (determining the individual local maxima): We set
\begin{align}
\Delta = \kappa / (2 \pi d)\ .\label{eq:symmetricsqueezinggeneraldef}
\end{align}
The special choice made in Eq.~\eqref{eq:symmetricsqueezinggeneraldef} is motivated by the fact that this provides a particularly simple linear optics implementation of the Fourier transform (see~\cite[Theorem 6.1]{cliffordslinearoptics2025}). We call the corresponding code (which now only depends on~$\kappa$) the 
symmetrically squeezed GKP code with squeezing parameter~$\kappa$, and denote the code space as 
\begin{align}
    \label{eq:GKPidealcode}
\gkpcode{\kappa}{\star}{d}&:=\mathsf{span}\left\{\gkp^{1/(2d)}_{\kappa, \kappa / (2 \pi d)}(j)_d\right\}_{j\in\{0,\ldots,d-1\}}\ .
\end{align}
We note that for a fixed code space dimension~$d\geq 2$, this defines a one-parameter family~$\{\gkpcode{\kappa}{\star}{d}\}_{\kappa>0}$ of symmetrically squeezed approximate GKP codes.

\paragraph{Approximate GKP codes and limitations of linear optics.}
In recent work~\cite{cliffordslinearoptics2025}, we have shown that using linear optics to perform logical gates in approximate GKP codes does not work as may be expected.  Specifically,
we have considered the logical phase gate~$\Pgate$, a diagonal Clifford acting on computational basis states as~
\begin{align}\label{eq:defPgate}
    \Pgate\ket{j}=e^{i \pi j(j+c_d)/d}\ket{j} \qquad\textrm{where  }\qquad c_d = d \mod 2 \, , \qquad \textrm{for}\qquad j\in \{0,\ldots,d-1\}\ .
\end{align}
We have shown that the linear optics implementation (proposed in~\cite{gkp} for the ideal GKP code)
\begin{align}
    W_\Pgate&=e^{i(Q^2 + c_d \sqrt{2\pi/d}Q)/2}\label{eq:wpgatevd}
\end{align}
of this gate has a constant logical gate error lower bounded by (see~\cite[Result 2]{cliffordslinearoptics2025})
\begin{align}
    \gateerror_{\gkpcode{\kappa}{\star}{d}}(W_\Pgate, \Pgate) > \textfrac{3}{100}\ \label{eq:lowerboundlogicalgateerrorapp}
\end{align}
for all~$\kappa <1/250$ when considering the symmetrically squeezed approximate GKP code~$\gkpcode{\kappa}{\star}{d}$. Eq.~\eqref{eq:lowerboundlogicalgateerrorapp} means that the implementation~\eqref{eq:wpgatevd} is unsuitable for use with the approximate GKP code~$\gkpcode{\kappa}{\star}{d}$ as it introduces a constant logical gate error even in the absence of noise. 
      
In more detail, the (composable)  logical gate error~$\gateerror_{\cL}(W_U,U)$ quantifies the error of an  implementation~$W_U:\cH\rightarrow\cH$  (acting on a physical Hilbert space~$\cH$)
of a logical gate~$U$ on a code space~$\cL\subset \cH$.
It is defined as 
\begin{align}
    \gateerror_{\cL}(W_U,U)=\left\|(\cW-\encodedC{U})\circ \Pi_{\cL}\right\|_\diamond\, , 
\end{align}
where~$\Pi_{\cL}(\rho)=\pi_\cL\rho \pi_{\cL}^\dagger$ is the CP map applying the orthogonal projection~$\pi_\cL:\cH\rightarrow\cL$,~$\cW(\rho)=W_U\rho W_U^\dagger$ is the CPTP map associated with the unitary~$W_U$, and~$\encodedC{U}:\cB(\cL)\rightarrow\cB(\cL)$ is the CPTP map
associated with the logical unitary~$U$ on the code space  (see Eq.~\eqref{eq:Ubold}).
We refer to Appendix~\ref{app:logicalgateerrorproperties} for more details on this definition and a summary of properties of this quantity. 

\paragraph{Beyond linear optics.}  Motivated by the no-go result for linear optics  expressed by Eq.~\eqref{eq:lowerboundlogicalgateerrorapp}, we ask if the use of alternative (non-Gaussian) resources can lead to  logical gate implementations which are exact in the limit~$\kappa\rightarrow 0$ of large squeezing. Specifically, we ask this question in the qubit-oscillator model, where in addition to the~$n$~oscillators and the linear optics operations~\eqref{it:firstlinearoptics}--\eqref{it:lastlinearoptics}, we add~$m$~qubits and the following capabilities:
\begin{enumerate}[(i)]\setcounter{enumi}{3}
\item\label{it:qubitopsvd}
one- and two-qubit gates, single-qubit state preparation of the computational basis state~$\ket{0}$, computational qubit basis measurements and
\item \label{it: qubit controlled ops}
qubit-controlled single-mode phase space displacements, i.e., unitaries of the form
\begin{align}
\ctrl_j W(\xi) = \proj{0}_j\otimes I_{L^2(\bbR)^{\otimes n}}+\proj{1}_j\otimes W(\xi)\ ,
\end{align}
where~$\xi\in \mathbb{R}^{2n}$ and~$W(\xi)=e^{-i\xi\cdot JR}$ is a displacement operator, and where~$\proj{0}_j$ and~$\proj{1}_j$ are the projections onto the computational basis states of the~$j$-th qubit. We assume that~$\xi$ is such that~$W(\xi)$ is a single-mode bounded strength phase-space displacement. 
\end{enumerate}
The additional operations~\eqref{it:qubitopsvd} and~\eqref{it: qubit controlled ops} can be used to introduce non-Gaussianity. In particular, they can be leveraged to 
design protocols preparing GKP states, see e.g.,~\cite{brenner2024complexity}. We refer to~\cite{liu2024hybridoscillatorqubitquantumprocessors} for a recent review of this hybrid qubit-oscillator model, as well as a discussion of associated experimental realizations.

We refer to each of the operations~\eqref{it:firstlinearoptics}--\eqref{it: qubit controlled ops} as an elementary operation in the hybrid qubit-oscillator model, and use the number of such operations as a measure of the (circuit) complexity of  an implementation.
Moreover, we denote the set of elementary unitary operations in~\eqref{it:secondlinearoptics}--\eqref{it: qubit controlled ops} acting on the system~$L^2(\mathbb{R})^{\otimes n} \otimes (\mathbb{C}^2)^{\otimes m}$ by~$\Uelem^{n,m}$. Occasionally, it will be useful to make the (constant) bounds on the squeezing and  displacements operators explicit. For this purpose, we define the set 
\begin{align}
\Uelem^{n,m}(\alpha,\zeta)&:=\{\ctrl_a e^{-itP_j},\ctrl_a e^{itQ_j},e^{-itP_j}, e^{itQ_j},(M_\beta)_j\}_{t\in (-\zeta,\zeta), \, \beta\in (\alpha^{-1},\alpha), \, 
j\in \{1,\ldots,m\} , \, 
a\in \{1,\ldots,n\}}\\
&\qquad \cup \{U_a,U_{a,b}\, |\ U_a, U_{a,b} \textrm{ one- or two-qubit unitary }\}_{a,\,
b\in \{1,\ldots,n\}}\ \label{eq:mygroupdefinitionmultiqubitbounded}
\end{align}
of elementary unitary operations on~$L^2(\mathbb{R})^{\otimes n} \otimes (\mathbb{C}^2)^{\otimes m}$ of strength bounded by parameters~$\alpha\geq 1$ and~$\zeta\geq 1$ for squeezing and displacement, respectively.

\subsection{Our contribution}
We propose new qubit-oscillator-based implementations of logical gates/unitaries 
for~$\ell$~logical qubits encoded in an approximate symmetrically squeezed GKP code~$\gkpcode{\kappa}{\star}{2^\ell}$ with squeezing parameter~$\kappa>0$. Our main result is the following. 

\begin{theorem}\label{thm:main}
Let~$T,\ell\in\bbN$. Let~$U=U_T\cdots U_1:(\mathbb{C}^2)^{\otimes \ell}\rightarrow (\mathbb{C}^2)^{\otimes \ell}$ be an~$\ell$-qubit unitary circuit which is the composition of~$T$ two-qubit gates~$U_1,\ldots,U_T$. 
Then there is a unitary circuit 
\begin{align}
W_U= W_{T'}\cdots W_1:L^2(\mathbb{R})^{\otimes 2}\otimes  (\mathbb{C}^2)^{\otimes 3}\rightarrow L^2(\mathbb{R})^{\otimes 2}\otimes  (\mathbb{C}^2)^{\otimes 3}
\end{align}
consisting of 
\begin{align}
T'=O(\ell^2 T)
\end{align}
elementary unitary operations~$W_1,\ldots,W_{T'}\in\Uelem^{2,3}$ on two oscillators and three qubits such that 
 \begin{align}
        \gateerror_{\cL_\kappa}( W_U,  U ) \le O(\ell T \kappa)\ . \label{eq:mainresult}
    \end{align}
    Here the code space~$\cL_\kappa\cong (\mathbb{C}^2)^{\otimes \ell}$ is 
    \begin{align}
    \cL_\kappa = \gkpcode{\kappa}{\star}{2^\ell}\otimes\mathbb{C}(\ket{\Psi})\subset L^2(\mathbb{R})^{\otimes 2}\otimes (\mathbb{C}^2)^{\otimes 3}\ ,\label{eq:codespacelkappa}
    \end{align}
    where~$\gkpcode{\kappa}{\star}{2^\ell}$ is the symmetrically squeezed GKP code  with parameter~$\kappa>0$ encoding a qudit of dimension~$2^\ell$ (see Eq.~\eqref{eq:GKPidealcode}), and where~$\Psi\in L^2(\mathbb{R})\otimes (\mathbb{C}^2)^{\otimes 3}$ is the state
    \begin{align}
    \ket{\Psi}&=|\gkp_{\kappa}^{2^{-(\ell+1)}}(0)_2\rangle\otimes \ket{0}^{\otimes 3}\ .
    \end{align}
\end{theorem}
Put more succinctly, Theorem~\ref{thm:main} demonstrates that any multi-qubit circuit~$U$ can approximately be recompiled (with linear overhead in the number of elementary gates as a function of the size of~$U$)    into a circuit~$W_U$  on two oscillators and three qubits. Furthermore, the 
logical gate error of the implementation vanishes in the limit~$\kappa\rightarrow 0$ of large squeezing.
In particular, in this limit, the elementary operations~$\Uelem^{2,3}$ together with the GKP resource state~$|\gkp_\kappa^{2^{-(\ell+1)}}(0)_2\rangle$ can be used to generate a dense subgroup of~$\ell$-qubit unitaries on an appropriate subspace of~$L^2(\mathbb{R})^{\otimes 2}\otimes (\mathbb{C}^2)^{\otimes 3}$. This means that a hybrid qubit-oscillator system consisting of two oscillators and three qubits gives a universal platform for realizing any~$\ell$-qubit (logical) unitary. 
We note that GKP states can be prepared with the set of all (including non-unitary) elementary operations~\eqref{it:firstlinearoptics}--\eqref{it: qubit controlled ops} using one additional auxiliary mode and one qubit, see~\cite{brenner2024complexity}.

\paragraph{Implementations of Cliffords.} As an example, we apply Theorem~\ref{thm:main} to obtain accurate implementations of the qudit Cliffords~$\Fgate, \Pgate, \cZ$ for (encoded) qudits of dimension~$d=2^\ell$, where~$\ell \in \mathbb{N}$. Here~$\Fgate$ is the Fourier transform acting as 
\begin{align}
    \label{eq:Fgate}
\Fgate \ket{x}& = 2^{-\ell/2} \sum_{y \in \mathbb{Z}_{2^\ell}} e^{2\pi i x y /2^{\ell}} \ket{y}\qquad\textrm{ for }\qquad x \in \mathbb{Z}_{2^\ell} :=\{0,\ldots, 2^\ell -1\}\ ,
\end{align} 
the unitary~$\Pgate$ is the one-qudit phase gate defined by Eq.~\eqref{eq:defPgate}, and~$\cZ$ is the two-qudit controlled phase gate acting as  
 \begin{align}
    \label{eq:CZgate}
 \cZ (\ket{x}\otimes \ket{y}) &= e^{2\pi i x y / 2^\ell} (\ket{x} \otimes \ket{y})\qquad\textrm{ for }\qquad x,y \in \mathbb{Z}_{2^\ell}\ .
 \end{align}
By a simple application of Theorem~\ref{thm:main}, we demonstrate that each unitary~$U \in \{\Fgate, \Pgate, \cZ\}$ has an implementation~$W_U=W_{T_U}\cdots W_1$ using~$T_U=O(\ell^3)$ elementary operations~$W_1,\ldots,W_{T_U}\in\Uelem^{2,3}$ where the qudit is encoded in a symmetrically squeezed GKP code~$\gkpcode{\kappa}{\star}{2^\ell}$, with a gate error which scales as~$O(\ell^2  \kappa)$. 

Combined with the results from~\cite{cliffordslinearoptics2025}  (which show that qudit Paulis~$X,Z$ have accurate implementations  by Gaussian unitaries in symmetrically squeezed GKP codes), we conclude the following: There is a complete set  of generators of the (logical) qudit Clifford group on~$\mathbb{C}^{2^\ell}$, where each generator has an efficient implementation  (i.e., one which uses~$\mathsf{poly}(\ell)$ elementary operations), and each implementation has a logical error vanishing linearly in~$\kappa$ with a polynomial prefactor (in~$\ell$). In particular, these implementations become exact in the limit~$\kappa\rightarrow 0$ of infinite squeezing.

\paragraph{Noisy implementations.} 
We note that our implementations are also robust to noise: The logical gate error of a noisy implementation is bounded in terms of a computable function of the underlying noise channel~$\cN$.

In more detail, consider the ideal implementation~$W_U=W_{T'}\cdots W_1$ of a logical unitary~$U=U_{T'}\cdots U_1$ (see Theorem~\ref{thm:main}). The corresponding noisy implementation 
\begin{align} 
\widetilde{\cW}_U&=(\cN\circ \cW_{T'})\circ\cdots \circ (\cN\circ \cW_1) \label{eq:noisyimplementationwidetildewU}
\end{align}
with noise channel~$\cN$ satisfies 
\begin{align}
    \label{eq:gateerrornoisyineq}
\gateerror_{\cL_\kappa}(\widetilde{\cW}_U,U)&\leq \gateerror_{\cL_\kappa}(\cW_U,U)+
\sum_{t=1}^{T'}\gateerror_{\cL_\kappa^{(t)}}(\cN,I_{\cL_\kappa^{(t)}}) \ ,
\end{align}
where we write~$\cW_t(\rho)=W_t\rho W_t^\dagger$, $
\cL_\kappa^{(t)}=W_{t-1}\cdots W_1\cL_\kappa \subset \cH$ for~$t\in \{1,\ldots,T'\}$
and we use~$I_{\cL_\kappa^{(t)}}$ to denote the identity on the subspace~$\cL_\kappa^{(t)}$. Eq.~\eqref{eq:gateerrornoisyineq} is an immediate consequence of the triangle inequality and the monotonicity of the diamond norm.

This implies for example that logical Clifford group generators of an encoded~$2^\ell$-qudit system satisfy the following: Associated noisy implementations have a logical gate error with a component that scales with~$O(\ell T\kappa)$ plus a term that depends on the noise.

\subsection{Outline}
This work is structured as follows.
In Section~\ref{sec:bitbasedunitaries} we construct a general bit-transfer unitary and demonstrate how it can be used to build implementations of multi-qubit gates.
In Section~\ref{sec:additionaloperations} we show how these operations can be implemented physically in ideal GKP codes using the elementary operations~\eqref{it:firstlinearoptics}--\eqref{it: qubit controlled ops}.
In Section~\ref{sec:logicalgateerrorsofimplementation} we prove the main result, i.e., Theorem~\ref{thm:main}, and derive bounds on implementations of Cliffords in hybrid qubit-oscillator systems. 

The appendices are structured as follows. Appendix~\ref{app:logicalgateerrorproperties} recalls the definition and properties of the logical gate error. Appendix~\ref{app:matrixelements} computes matrix elements of hybrid qubit-oscillator implementations of bit-manipulations maps. 
In Appendix~\ref{sec: lem: universal bound special-proof} we bound the logical gate error of the implementations of basic bit-manipulation maps proposed in Section~\ref{sec:additionaloperations}.
In Appendix~\ref{sec:approximateGKPcodescomb-app} we generalize the results obtained to the case of rectangular- (instead of Gaussian-) envelope GKP codes. 
Appendix~\ref{sec:generalizationmultiple} generalizes the construction of qubit-oscillator implementations from one- to two-mode encodings. 

\section{Construction of bit-based unitaries \label{sec:bitbasedunitaries}}
In this section we define multiple operations 
for manipulating individual bits of~$\ell$-bit strings encoded into a~$2^\ell$-dimensional space without a natural tensor product structure. We first introduce these operations on a logical level as a preparation to Section~\ref{sec:additionaloperations}, where we turn to physical implementations. In Section~\ref{sec:additionaloperations} the~$2^\ell$-dimensional space used here will correspond to a GKP code space, that is, a code space embedded in a single oscillator.

In Section~\ref{sec:multiquditsystemlogical} we define what we call basic bit-manipulation maps. 
In Section~\ref{sec:multiquditsystemlogicalderived} we then show how these basic bit-manipulation maps can be composed to obtain unitary circuits achieving what we call bit-transfer.
In Section~\ref{sec:multiqubitbitextraction} we translate unitaries on~$\ell$ qubits to unitary circuits on a~$2^\ell$-dimensional space using bit-transfer unitaries.
  
\subsection{Basic bit-manipulation maps \label{sec:multiquditsystemlogical}}

We define linear maps involving a qudit of dimension~$2^\ell$ (with orthonormal basis~$\{\ket{x}\}_{x\in \mathbb{Z}_{2^\ell}}$, $ \mathbb{Z}_{2^\ell}= \{0, \ldots, 2^\ell -1 \}$) and a single qubit, that is, unitaries on~$\mathbb{C}^{2^\ell}\otimes \mathbb{C}^2$. These correspond to associated bit manipulations.  
We represent elements of~$\mathbb{Z}_{2^\ell}$ by~$\ell$-bit strings as
\begin{align}
x&=[x_{\ell-1},\ldots,x_{0}]\qquad\textrm{ where }\qquad x=\sum_{j=0}^{\ell-1} 2^{j}x_j\in\mathbb{Z}_{2^\ell}\, .
\end{align}
The qubit-controlled modular shift~$\qCX{\ell}: \bbC^{2^\ell} \otimes \bbC^2 \rightarrow \bbC^{2^\ell} \otimes \bbC^2$ is defined as the unitary which acts as
\begin{align}
\qCX{\ell} (\ket{x}\otimes\ket{b})&=\ket{x\oplus b}\otimes\ket{b}\qquad\textrm{ for all }\qquad x\in \mathbb{Z}_{2^\ell}\textrm{ and }b\in \{0,1\}\ ,
\end{align} (here~$\oplus$ denotes addition modulo~$2^\ell$)
i.e., it is given by 
\begin{align}
\qCX{\ell}&=I\otimes \proj{0}+X\otimes \proj{1}\ 
\end{align}
where~$X:\mathbb{C}^{2^\ell}\rightarrow\mathbb{C}^{2^\ell}$ is defined as
\begin{align}
X\ket{x}&=\ket{x\oplus 1}\qquad\textrm{ for }\qquad x\in \mathbb{Z}_{2^\ell}\ .
\end{align}

Next, we define a unitary~$\lsb{\ell}$ on~$\mathbb{C}^{2^\ell}\otimes\mathbb{C}^2$ called the least significant bit-gate which extracts the least significant bit~$x_0$ of~$x$, copying it onto a qubit. Concretely, we use the orthonormal basis~$\{\ket{x}\otimes\ket{b}\}_{{x\in \mathbb{Z}_{2^\ell}, \,
b\in \{0,1\}}}$ of~$\mathbb{C}^{2^\ell}\otimes\mathbb{C}^2$. Then we set
\begin{align}
\lsb{\ell} (\ket{x}\otimes\ket{b})&=\ket{x}\otimes \ket{b\oplus x_0}\qquad \textrm{for }\qquad x=[x_{\ell-1},\ldots,x_0]\textrm{ and }b\in \{0,1\}\ ,
\end{align} 
where~$\oplus$ denotes addition modulo~$2$.

We need one more map. It is an isometry which embeds a~$2^\ell$-dimensional space ($\ell \in \bbN$) into a~$2^{\ell+1}$-dimensional space. To define it, consider the injective map
\begin{align}
\begin{matrix}
\embedmap_\ell: & \mathbb{Z}_{2^{\ell}}& \rightarrow & \mathbb{Z}_{2^{\ell+1}}\\
& x & \mapsto & 2x
\end{matrix} \ .
\end{align}
When applied to computational basis states, it gives rise to an isometry we denote as
\begin{align}
\begin{matrix}
\embedmap_\ell:& \mathbb{C}^{2^\ell}&\rightarrow&\mathbb{C}^{2^{\ell+1}}\\
&\ket{x}&\mapsto &\ket{2x}
\end{matrix} \ , 
\label{eq:embeddingmapdefinitionisometry}
\end{align}
with a slight abuse of notation.
We refer to~\eqref{eq:embeddingmapdefinitionisometry} as an ``embedding map''.  We note that -- when expressed in binary -- this operation appends a~$0$ as the least significant bit, i.e., it acts as 
\begin{align}
\embedmap_\ell\ket{[x_{\ell-1},\ldots,x_0]}&=\ket{[x_{\ell-1},\ldots,x_0,0]}\qquad\textrm{ for }\qquad x=[x_{\ell-1},\ldots,x_0]\in \mathbb{Z}_{2^\ell}\ .
\end{align}

See Table~\ref{tab:bitextractionunitary} for an illustration of these  (unitary respectively isometric) maps which we call basic bit-manipulation maps. We denote them by 
\begin{align} \label{eq:defcG}
    \cG_\ell = \left\{\qCX{\ell}, \lsb{\ell},\embedmap_{\ell}\right\}\, .
\end{align}
We also write
\begin{align}
\cG(\ell)=\bigcup_{k=1}^{\ell} \left(\cG_k\cup \cG_k^\dagger \right)\label{eq:bitmanipulationcollection}
\end{align}
for the collection of all basic bit-manipulation maps and their adjoints up to dimension~$2^\ell$. All the maps belonging to~$\cG(\ell)$ map between spaces of the form~$\mathbb{C}^{2^k}\otimes\mathbb{C}^2$ for~$k\in \{1,\ldots,\ell+1\}$.

\begin{table}[H]
\centering
\renewcommand{\arraystretch}{1.3} 
\setlength{\tabcolsep}{10pt} 
\rowcolors{1}{white}{gray!15} 
\centering
\resizebox{\textwidth}{!}{\begin{tabular}{c|c|c}
    \textbf{diagram} &\textbf{gate}   &\textbf{type}\\
        \raisebox{0.0cm}{\includegraphics{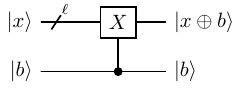}} &\raisebox{0.5cm}{qubit-controlled shift~$\qCX{\ell}$} &\raisebox{0.5cm}{unitary} \\
        \raisebox{0.0cm}{\includegraphics{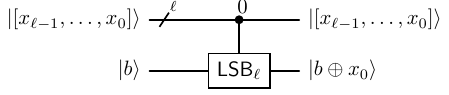}} & \raisebox{0.5cm}{least significant bit~$\lsb{\ell}$} &\raisebox{0.5cm}{unitary}\\
        \raisebox{0.0cm}{\includegraphics{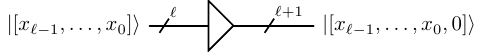}} & \raisebox{0.25cm}{embedding map~$\embedmap_\ell$} &\raisebox{0.25cm}{isometry}\\
    \end{tabular}}
\caption{The set~$\cG_\ell$ of basic bit-manipulation maps: The qubit-controlled modular shift unitary~$\qCX{\ell}$
and the least significant bit-unitary~$\lsb{\ell}$  on~$\mathbb{C}^{2^\ell}\otimes\mathbb{C}^2$, as well as the embedding isometry~$\embedmap_\ell:\bbC^{2^\ell} \rightarrow \bbC^{2^{\ell+1}}$ defined by Eq.~\eqref{eq:embeddingmapdefinitionisometry}.\label{tab:bitextractionunitary}}
\end{table}

\subsection{Circuits for bit-transfer unitaries in~$\mathbb{C}^{2^\ell} \otimes \mathbb{C}^2$}
\label{sec:multiquditsystemlogicalderived}
The basic bit-manipulation maps (i.e., operations belonging to~$\cG(\ell)$, see Eq.~\eqref{eq:bitmanipulationcollection}) are particularly powerful when used in a modular fashion. By composing these basic operations, one can construct more complex functionalities represented by their action on basis states.

We first consider what we call bit-transfer unitaries. Generalizing the~$\lsb{\ell}$-gate, we may consider unitaries which extract any individual bit~$x_j$,~$j\in \{0,\ldots,\ell-1\}$ in the binary representation of~$x\in \mathbb{Z}_{2^\ell}$, copying it onto a qubit. 

Concretely, consider again a Hilbert space of the form~$\mathbb{C}^{2^\ell}\otimes\mathbb{C}^2$ 
with orthonormal basis~$\{\ket{x}\otimes\ket{b}\}_{x\in \mathbb{Z}_{2^\ell} , \,
b\in \{0,1\}}$. For every~$j\in \{0,\ldots,\ell-1\}$, we then define the bit-transfer gate 
\begin{align}
\bittransfer{\ell}{j}:\mathbb{C}^{2^\ell}\otimes \mathbb{C}^2\rightarrow \mathbb{C}^{2^\ell}\otimes \mathbb{C}^2
\end{align}
by
\begin{align} \label{eq:defCjellX}
        \bittransfer{\ell}{j}(\ket{x}\otimes\ket{b})&=|x - 2^j (b \oplus x_j)\rangle \otimes \ket{ b \oplus x_j}\quad \textrm{ for } x=[x_{\ell-1},\ldots,x_0]\in\mathbb{Z}_{2^\ell} \textrm{ and } b\in \{0,1\}\, ,
\end{align} where the term~$x - 2^j (b \oplus x_j)$ is to be understood modulo~$2^\ell$, see Table~\ref{fig:bitextractionunitary} for an illustration. 
In particular, it follows that 
\begin{align} \label{eq:propertybiextract}
    \bittransfer{\ell}{j}(\ket{\base{}{x_{\ell-1},\ldots,x_0}}\otimes\ket{0})&=\ket{\base{}{x_{\ell-1},\ldots,x_{j+1},0,x_{j-1},\ldots,x_0}}\otimes \ket{x_j}\, 
\end{align}
for any~$[x_{\ell-1},\ldots,x_0]\in\mathbb{Z}_{2^\ell}$, i.e.,~$\bittransfer{\ell}{j}$ extracts the bit~$x_j$. 
Correspondingly, the adjoint map~$(\bittransfer{\ell}{j}{})^\dagger$ satisfies
\begin{align} \label{eq:propertybiextract_adjoint}
    (\bittransfer{\ell}{j}{})^\dagger \ket{\base{}{x_{\ell-1},\ldots,x_{j+1},0,x_{j-1},\ldots,x_0}}\otimes \ket{x_j}
    &= \ket{\base{}{x_{\ell-1},\ldots,x_0}}\otimes\ket{0} \ . 
\end{align}

\begin{table}[H]
    \centering
    \renewcommand{\arraystretch}{1.3}   
    \setlength{\tabcolsep}{10pt}        
    \rowcolors{1}{white}{gray!15}       
        \begin{tabular}{c|c|c}
        \textbf{diagram} & \textbf{gate}  &\textbf{type}\\
         \raisebox{0pt}{\includegraphics{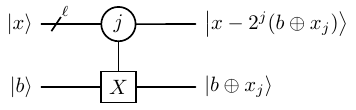}} 
             & \raisebox{0.5cm}{$\bittransfer{\ell}{j}$}& \raisebox{0.5cm}{unitary} \\
        \end{tabular}%
    \caption{The bit-transfer unitary~$\bittransfer{\ell}{j}$. The adjoint bit-transfer unitary~$(\bittransfer{\ell}{j}{})^\dagger$ is represented by a similar diagram with~$X^\dagger$ instead of~$X$.}
    \label{fig:bitextractionunitary}
\end{table}
The following lemma shows that bit-transfer unitaries can be realized using basic bit-ma\-ni\-pu\-la\-tion maps, i.e., operations belonging to the set~$\cG(\ell)$ (see Eq.~\eqref{eq:bitmanipulationcollection}).  The construction involves two auxiliary qubits which are in the state~$\ket{0}^{\otimes 2}$. The state of these auxiliary qubits  is unchanged by the circuit, i.e., it acts as a catalyst. 

\begin{lemma}[Circuit for the bit-transfer unitary~$\bittransfer{\ell}{j}$]   \label{lem: CjellX circuit}
Let~$\ell\in\mathbb{N}$ and~$j\in \{0,\ldots,\ell-1\}$ be arbitrary. There is a circuit~$V^j_\ell$ acting on~$(\mathbb{C}^{2^\ell}\otimes\mathbb{C}^2)\otimes (\mathbb{C}^2)^{\otimes 2}$
 composed of fewer than~\blue{$12j-4$} maps belonging to the set~$\cG(\ell)$ of basic bit-manipulation operations such that
\begin{align}
V^j_\ell \left((\ket{x}\otimes\ket{b})\otimes\ket{0}^{\otimes 2}\right)&=(\ket{x - 2^j(b \oplus x_j)}\otimes\ket{b\oplus x_j})\otimes\ket{0}^{\otimes 2}\\
&=\left(\bittransfer{\ell}{j}(\ket{x}\otimes\ket{b})\right)\otimes\ket{0}^{\otimes 2}
\end{align}
for all~$x=[x_{\ell-1},\ldots,x_0]\in \mathbb{Z}_{2^\ell}$ and~$b\in \{0,1\}$. 
In other words, the circuit~$V^j_\ell$ 
realizes the unitary~$\bittransfer{\ell}{j}$ on the subspace
\begin{align}
(\mathbb{C}^{2^\ell}\otimes\mathbb{C}^2)\otimes (\mathbb{C}\ket{0}^{\otimes 2})\subset (\mathbb{C}^{2^\ell}\otimes\mathbb{C}^2)\otimes (\mathbb{C}^2)^{\otimes 2}
\end{align}
 where the two auxiliary qubits are in the state~$\ket{0}$. 
\end{lemma}

\begin{figure}[b!]
    \centering
    \begin{subfigure}{1.0\textwidth}
        \centering
        \includegraphics{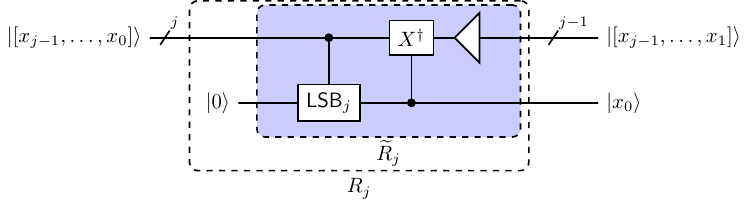}
        \caption{Circuit realizing the isometry~$R_j$ defined in Eq.~\eqref{eq:rjdefinitiona}.  It is based on a subcircuit~$\widetilde{R}_j$ (shaded in blue) acting on a bipartite system with the  auxiliary (qubit) system in the state~$\ket{0}$. The adjoint of the isometry~$\embedmap_{j-1}$  is represented by an inverted triangle.}
        \label{fig:RjRjtilde}
    \end{subfigure}\\
    \vspace{5ex}
    \begin{subfigure}{\textwidth}
        \centering
        \includegraphics[width=\textwidth]{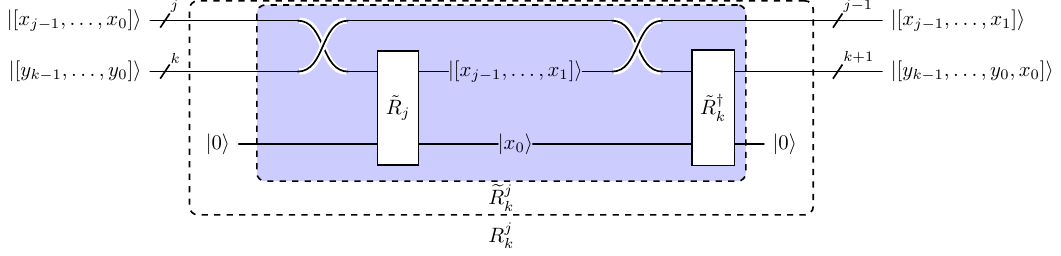}
        \caption{Circuit realizing the isometry~$R^j_k$ (see Eq.~\eqref{eq:rjkdefinitionb}).
  It is based on a subcircuit~$\widetilde{R}^j_k$  (blue) acting on an auxiliary (qubit) system in the state~$\ket{0}$.  
        }
        \label{fig:RjkRjktilde}
    \end{subfigure}
    \caption{
    Circuit realizing the isometry~$R^j_{k}$
      based on the isometry~$R_j$ and the adjoint of~$R_k$. All blue subcircuits only involve basic bit-manipulation maps from the set~$\cG(\ell)$. 
    \label{fig:rjkisometrydef}
    }
\end{figure}

\begin{proof}
For~$j\in \{1,\ldots,\ell\}$ consider the right-shift isomorphism
\begin{align}
\begin{matrix}
R_{j}:&\mathbb{C}^{2^j}  & \rightarrow & \mathbb{C}^{2^{j-1}} &\otimes& \mathbb{C}^{2} \\
&\ket{[x_{j-1},\ldots,x_0]} & \mapsto & \ket{[x_{j-1},\ldots,x_1]} &\otimes &\ket{x_0}\label{eq:rjdefinitiona}
\end{matrix}\end{align}
and for~$j\in \{1,\ldots,\ell\}$ and~$k \in \{1,\ldots,\ell\}$  the right-shift isomorphism
\begin{align}
\begin{matrix}
R^{j}_k:&\mathbb{C}^{2^j} &\otimes& \mathbb{C}^{2^k} & \rightarrow & \mathbb{C}^{2^{j-1}} &\otimes& \mathbb{C}^{2^{k+1}} \\
&\ket{[x_{j-1},\ldots,x_0]}& \otimes &\ket{[y_{k-1},\ldots,y_0]} & \mapsto & \ket{[x_{j-1},\ldots,x_1]} &\otimes &\ket{[y_{k-1},\ldots,y_0,x_0]}\ .
\end{matrix}\label{eq:rjkdefinitionb}
\end{align}
A circuit realization of these isometries using basic bit-manipulation operations (i.e., operations belonging to the set~$\cG(\ell)$) is given in Fig.~\ref{fig:rjkisometrydef}.

\begin{figure}[t!]
    \includegraphics[width=\textwidth]{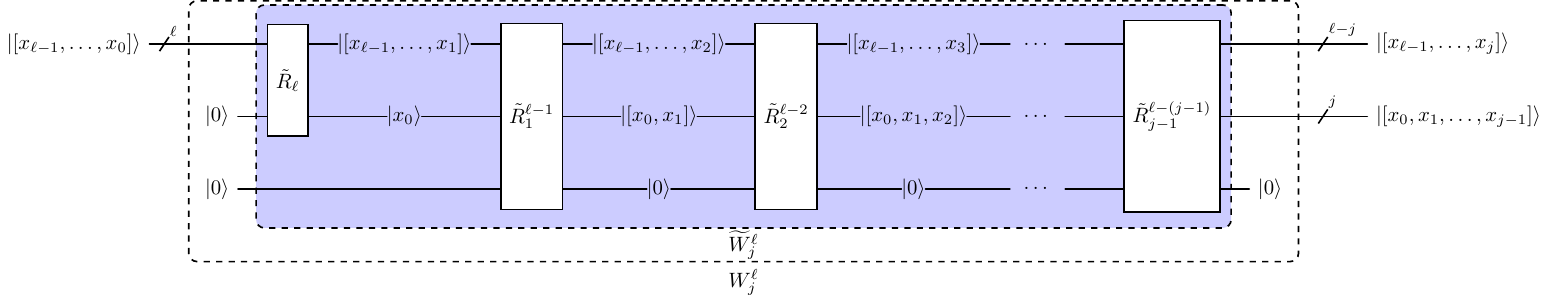} 
       \caption{Circuit for the isometry~$W^\ell_j$. It uses two auxiliary (qubit) systems each in the state~$\ket{0}$.
       The construction is presented in terms of the unitaries~$\widetilde{R}_j$ and~$\widetilde{R}^j_k$ introduced in Figs.~\ref{fig:RjRjtilde} and~\ref{fig:RjkRjktilde}, respectively. We denote the blue subcircuit by~$\widetilde{W}^{\ell}_j$.  \label{fig:welljdef}}
\end{figure}

\begin{figure}[t!]
    \centering
    \resizebox{\textwidth}{!}{\includegraphics{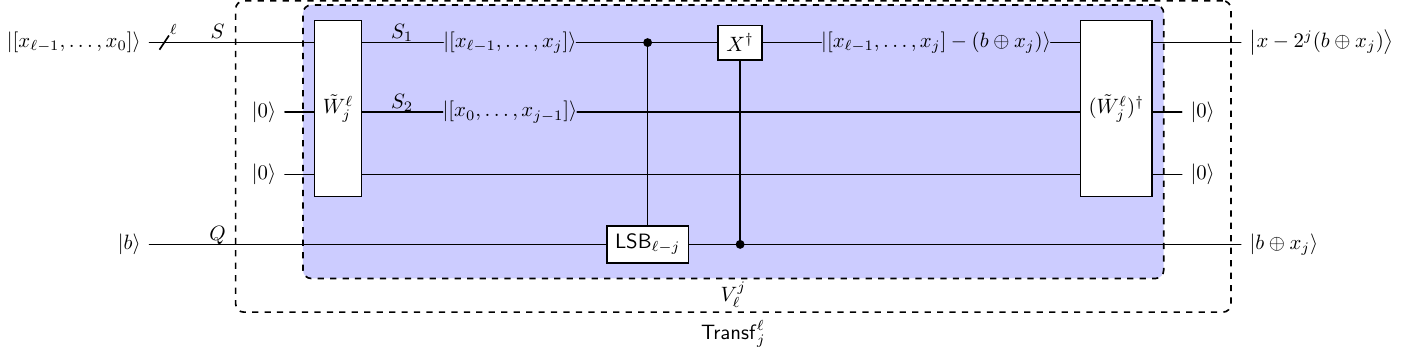}}
    \caption{Circuit for the bit-transfer unitary~$\bittransfer{\ell}{j}$ where~$x = [x_{\ell-1}, \dots, x_0]$. Here~$[x_{\ell-1}, \dots, x_j] - (b \oplus x_j)$ is understood modulo~$2^{\ell-j}$. The
 unitary~$\widetilde{W}^\ell_j$ is the blue subcircuit in Fig.~\ref{fig:welljdef}. The blue subcircuit in this figure is the circuit~$V^j_\ell$ which realizes~$\bittransfer{\ell}{j}$.
    \label{fig:bitextractionimplement}}
\end{figure}

For~$j\in \{0,\ldots,\ell-1\}$, define the unitary map
\begin{align}
W^\ell_j&=  R^{\ell-(j-1)}_{j-1}\cdots R^{\ell-1}_{1} R_{\ell}: \mathbb{C}^{2^\ell} \rightarrow \mathbb{C}^{2^{\ell-j}}\otimes \mathbb{C}^{2^j}\ .
\end{align} 
A circuit realizing the map~$W^\ell_j$ is given in  Fig.~\ref{fig:welljdef}. It acts on computational basis states (in the binary representation) as 
\begin{align}
W^\ell_j \ket{[x_{\ell-1},\ldots,x_0]}&=\ket{[x_{\ell-1},\ldots,x_{j}]}\otimes \ket{[x_{0},\ldots,x_{j-1}]}  \textrm{ for } [x_{\ell-1},\ldots,x_0]\in \{0,\ldots,2^\ell-1\} \ \label{eq:Welljmapdefinition}
\end{align}
and thus 
\begin{align}
    (W_j^\ell)^\dagger (\ket{[x_{\ell-1},\ldots,x_{j}]}\otimes \ket{[x_{0},\ldots,x_{j-1}]}) = \ket{2^{j} [x_{\ell-1},\ldots,x_{j}] + [x_{j-1},\ldots,x_{0}]} \ .  \label{eq:Welljdagger}
\end{align}
It follows from Eqs.~\eqref{eq:Welljmapdefinition} and~\eqref{eq:Welljdagger} that the map
\begin{align}
(W_j^\ell)^\dagger_{S_1S_2\rightarrow S} (\qCX{\ell})_{Q\rightarrow S_1}(\LSB^{\ell-j})_{S_1\rightarrow Q} (W_j^\ell)_{S\rightarrow S_1S_2}
\end{align}
realizes the desired functionality, where~$Q$ denotes the qubit system, whereas~$S$,~$S_1$ and~$S_2$ are of dimensions~$2^\ell$,~$2^{\ell-j}$ and~$2^{j}$, respectively (see Fig.~\ref{fig:bitextractionimplement}).

The circuits~$\widetilde{R}_j$ and~$\widetilde{R}_{k}^j$ require~$3$ and~$6$ maps from~$\cG(\ell)$, respectively. The circuit~$\widetilde{W}_j^\ell$ in Fig.~\ref{fig:welljdef} consists of~$1$ circuit of the type~$\widetilde{R}_j$ and~$j-1$ of the type~$\widetilde{R}_{k}^j$, amounting to a total of~$6(j-1)+3=6j-3$ maps from~$\cG(\ell)$. Then the implementation of~$\bittransfer{\ell}{j}$ in Fig.~\ref{fig:bitextractionimplement} which includes two isometries of the type~$W_j^\ell$ and two additional operations from~$\cG(\ell)$ uses up to~$2(6j-3)+2=12j-4$ operations in~$\cG(\ell)$.
\end{proof}

\subsection{Bit-transfer-based implementation of multi-qubit gates \label{sec:multiqubitbitextraction}}

In this section we consider a multi-qubit unitary~$U: \left(\mathbb{C}^2\right)^{\otimes 
\ell} \rightarrow\left(\mathbb{C}^2\right)^{\otimes 
\ell}$ and show how it can be realized on a system of the form~$\mathbb{C}^{2^\ell}$ using bit-transfer.
(We consider that~$U$ acts non-trivially on two qubits and, by notational abuse, call~$U$ a two-qubit unitary.)
Since the qudit dimension is of the form~$d=2^\ell$ with~$\ell\in\mathbb{N}$, we may also think of the qudit as consisting of~$\ell$ qubits with a suitable isomorphism~$\mathbb{C}^{2^\ell}\cong (\mathbb{C}^2)^{\otimes \ell}$. 
Combined with property~\eqref{eq:propertybiextract}, this allows us to outsource the computation of any two-qubit unitary to two auxiliary qubit systems 
which are initialized in the computational basis state~$\ket{0}$.

In more detail, our construction uses a system of the form~$\mathbb{C}^{2^\ell}\otimes (\mathbb{C}^2)^{\otimes 4}$, i.e.,~$4$~additional auxiliary qubits 
in addition to the system~$\mathbb{C}^{2^\ell}$ which encodes the logical information (i.e.,~$\ell$ logical qubits). To address the latter, we use the  isomorphism~$J_\ell: \left(\mathbb{C}^2\right)^{\otimes \ell} \rightarrow \mathbb{C}^{2^\ell}$ defined by
\begin{align}
    J_\ell\left(\ket{x_{\ell-1}} \otimes \dots \otimes \ket{x_0}\right) = \ket{[x_{\ell-1}, \dots, x_0]}\quad \textrm{for} \quad (x_0,\ldots,x_{\ell-1})\in \{0,1\}^\ell\ .
\end{align}
For better readability we sometimes write~$\ket{a}\ket{b}$ for the tensor product~$\ket{a}\otimes \ket{b}$.

\begin{lemma}[Bit-transfer-based implementation of two-qubit gates]
 \label{lem:general2qubitbitextract}
Let~$\ell\geq 2$ be an integer. Consider a ``logical''~$\ell$-qubit system denoted 
\begin{align}
A_0\cdots A_{\ell-1}=(\mathbb{C}^2)^{\otimes \ell}\ .
\end{align}
Let~$j<k$ and~$j,k\in \{0,\ldots,\ell-1\}$ be arbitrary.
Let~$U_{A_j A_k}: \left(\mathbb{C}^2\right)^{\otimes 
    \ell} \rightarrow\left(\mathbb{C}^2\right)^{\otimes 
    \ell}$ be a unitary 
    acting non-trivially only on the two qubits~$A_jA_k$.  
Let 
    \begin{align}
    SQ_1Q_2Q_3Q_4=\mathbb{C}^{2^\ell}\otimes (\mathbb{C}^2)^{\otimes 4}
    \end{align}
    be a system consisting of a~$2^\ell$-dimensional (``code'') space~$S$ and~$4$ auxiliary qubits~$Q_1\cdots Q_4$. 
    Then there is a circuit~$V_{U_{A_jA_k}}$ on~$SQ_1\cdots Q_4$ consisting of fewer than~\blue{$48\ell-16$} operations belonging to the set~$\cG(\ell)$ and a single two-qubit operation which satisfies 
    \begin{align}
       (J_\ell^{-1}  \otimes I) V_{U_{A_jA_k}} (J_\ell \otimes I)\left((\ket{x_{\ell-1}}\cdots \ket{x_0}) \otimes  \ket{0}^{\otimes 4}\right) = \left( U_{A_jA_k} (\ket{x_{\ell-1}}\cdots \ket{x_0})\right)\otimes \ket{0}^{\otimes 4}  \ \label{eq:circuitidellqubitgate}
    \end{align} 
    for all~$x=(x_0,\ldots,x_{\ell-1})\in \{0,1\}^\ell$. In other words, the circuit~$V_{U_{A_jA_k}}$ realizes the unitary~$U_{A_jA_k}$ on the~$2^\ell$ dimensional subspace 
    \begin{align}
        \mathbb{C}^{2^\ell}\otimes (\mathbb{C}\ket{0}^{\otimes 4})\subset SQ_1\cdots Q_4=\mathbb{C}^{2^\ell}\otimes (\mathbb{C}^2)^{\otimes 4}
        \end{align}
        where the four auxiliary qubits are in the state~$\ket{0}$.
\end{lemma}
In Appendix~\ref{sec:generalizationmultiple} we give a generalization (see Lemma~\ref{lem:general2qubitbitextracttwo}) of this result
to the case where~$2\ell$ qubits are encoded into a bipartite Hilbert space of the form~$(\mathbb{C}^{2^\ell})^{\otimes 2}$.
\begin{proof}
    The expressions~\eqref{eq:propertybiextract} and~\eqref{eq:propertybiextract_adjoint} imply
    \begin{align}
        &\left(\bittransfer{\ell}{j}\right)^\dagger_{SQ_1} \left(\bittransfer{\ell}{k}\right)^\dagger_{SQ_2}(I \otimes U_{Q_1Q_2}) \left(\bittransfer{\ell}{k}\right)_{SQ_2} \left(\bittransfer{\ell}{j}\right)_{SQ_1} \left( \ket{x}\otimes \ket{0}^{\otimes 2}\right)\\
        &\qquad = \left(J_\ell U_{A_jA_k} J_\ell^{-1}\ket{x}\right)\otimes \ket{0}^{\otimes 2} \ , \label{eq:tranfid}
    \end{align}
    for all~$j,k \in \{0,\dots, \ell-1\}$ and for all~$x \in \{0,\dots, 2^\ell -1\}$, where~$S$,~$Q_1$ and~$Q_2$ are systems of dimension~$2^\ell$,~$2$ and~$2$ respectively, see Fig.~\ref{fig:general2qubitimplement}.

    \begin{figure}[H]
        \centering
        \resizebox{\textwidth}{!}{\includegraphics{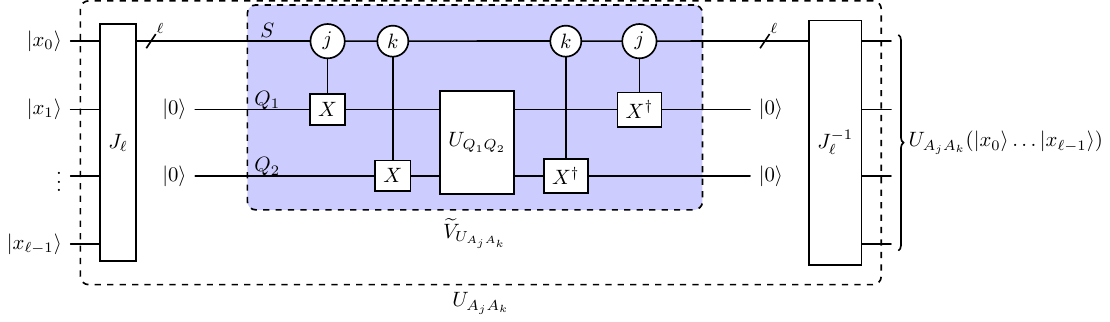}}
        \caption{Circuit implementing the two-qubit unitary~$U_{A_jA_k}$. It uses the bit-transfer unitaries~$\bittransfer{\ell}{j}$ and~$\bittransfer{\ell}{k}$, two additional qubit systems and the two-qubit unitary~$U_{Q_1Q_2}$ acting on the auxiliary two-qubit system~$Q_1Q_2$. The circuit in blue is~$\widetilde{V}_{U_{A_jA_k}}$. Utilizing the implementation given in Lemma~\ref{lem: CjellX circuit} for the bit-transfer unitaries results in the use of two additional qubits, in which case the total number of additional auxiliary qubits is four.
        \label{fig:general2qubitimplement}}
    \end{figure}

    The construction of~$V_{U_{A_jA_k}}$ relies on the identity~\eqref{eq:tranfid}, which directly implies that the unitary
    \begin{align}
        \widetilde{V}_{U_{A_jA_k}} =  \left(\bittransfer{\ell}{j}\right)^\dagger_{SQ_1} \left(\bittransfer{\ell}{k}\right)^\dagger_{SQ_2}(I \otimes U_{Q_1Q_2}) \left(\bittransfer{\ell}{k}\right)_{SQ_2} \left(\bittransfer{\ell}{j}\right)_{SQ_1} 
    \end{align} satisfies
    \begin{align}
        \widetilde{V}_{U_{A_jA_k}}\left(\ket{x}\otimes \ket{0}^{\otimes 2}\right) = \left(J_\ell U_{A_jA_k} J_\ell^{-1} \ket{x} \right)\otimes \ket{0}^{\otimes 2} \qquad \textrm{for all} \qquad x \in \{0,\dots, 2^\ell -1\}\, .\label{eq:imptrans-aux0}
    \end{align}
    Then, substituting each bit-transfer unitary~$\bittransfer{\ell}{j}$ and its adjoint in~$\widetilde{V}_{U_{A_jA_k}}$ 
    by the circuit given in Lemma~\ref{lem: CjellX circuit} results in a circuit~$V_{U_{A_jA_k}}$ which uses two additional auxiliary qubits~$Q_3$ and~$Q_4$. The transfer gates acting on~$SQ_1$ make use of the additional qubit~$Q_3$ which is reutilized in each application of a transfer gate, and similarly the transfer gates acting on~$SQ_2$ use the additional qubit~$Q_4$. Hence, four additional auxiliary qubits are used in total. 
    
    By Eq.~\eqref{eq:imptrans-aux0}, the circuit~$V_{U_{A_jA_k}}$~satisfies property~\eqref{eq:circuitidellqubitgate}.
    Furthermore,~$V_{U_{A_jA_k}}$ uses a two-qubit gate and four bit transfer unitaries, each of which requires~$12\ell-4$ basic bit-manipulation maps in~$\cG(\ell)$, totalling~$48\ell-16$ operations in~$\cG(\ell)$. This concludes the proof.
\end{proof}

\section{Implementation of logical unitaries in ideal GKP codes\label{sec:additionaloperations}}

In this section we show how logical unitaries can be implemented in ideal GKP codes. Although the code spaces are spanned by unphysical (unnormalizable) states, the action of 
certain Gaussian unitaries and phase-space displacements can conveniently be analyzed using position-eigenstates. 
We use this framework to derive exact circuit identities for ideal GKP codes in order to motivate the construction of circuits we expect to exhibit similar behavior when acting on approximate (finitely squeezed) GKP codes. 

In Section~\ref{sec:bitmanipulationimplement} we show how elementary operations in the hybrid qubit-oscillator model can be used to implement the basic bit-manipulation maps introduced in Section~\ref{sec:multiquditsystemlogical} and the bit-transfer unitaries introduced in Section~\ref{sec:multiquditsystemlogicalderived} in ideal GKP codes.
In Section~\ref{sec:GKPimplementmultiqubit} we extend this result by showing how two-qubit unitaries can be implemented in ideal GKP code spaces.

In the following, we use the encoding map 
\begin{align}
\begin{matrix}
\encmapgkp[d]: & \mathbb{C}^d & \rightarrow & \gkpcode{}{}{d}\\
 & \ket{j} & \mapsto &\ket{\gkp(j)_d}
 \end{matrix}\qquad\textrm{ for }\qquad j\in\mathbb{Z}_d\,
\end{align}
between~$\mathbb{C}^d$ and the ideal GKP code~$\gkpcode{}{}{d}$ for an integer~$d\ge 2$ (see Eq.~\eqref{eq: def ideal gkp} for the definition of the code states).
In a circuit we represent this map by the diagram
\begin{center}
\begin{tabular}{c}
    \includegraphics{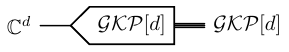} \,.
\end{tabular}
\end{center}
Similarly, the corresponding decoding map is defined by its action on basis states
\begin{align}
\begin{matrix}
\decmapgkp[d]=(\encmapgkp[d])^{-1}:&\gkpcode{}{}{d}&\rightarrow& \mathbb{C}^d\\
&\ket{\gkp(j)_d}& \rightarrow & \ket{j} 
\end{matrix}\qquad\textrm{ for }\qquad j\in \mathbb{Z}_d\ ,
\end{align}
linearly extended to all of~$\gkpcode{}{}{d}$.
We represent it as 
\begin{center}
    \begin{tabular}{c}
        \includegraphics{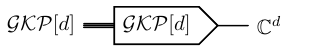} \,.
    \end{tabular}
    \end{center}

\subsection{Basic bit-manipulation and bit-transfer maps\label{sec:bitmanipulationimplement}}

Here we give physical realizations of basic bit-manipulation operations belonging to the set~$\cG(\ell)$ 
 in appropriate codes. 
 Recall that 
 the considered basic bit-manipulation operations are defined in terms of a set~$\cG_\ell$ of maps, which consists of  (see Eq.~\eqref{eq:defcG})
 \begin{enumerate}[1)]
    \item the qubit-controlled modular shift~$\qCX{\ell}$,
    \item the least significant bit-gate~$\lsb{\ell}$\ , 
\end{enumerate}
both of which are  unitaries on~$\mathbb{C}^{2^\ell}\otimes\mathbb{C}^2$, as well as 
\begin{enumerate}[3)]
    \item the embedding isometry~$\embedmap_\ell:\mathbb{C}^{2^\ell} \rightarrow \mathbb{C}^{2^{\ell+1}}$.
\end{enumerate}
The set~$\mathcal{G}(\ell)$ of all bit-manipulation maps consists of all such maps and their adjoints, see Eq.~\eqref{eq:bitmanipulationcollection}.
  
Here we consider logical implementations of these maps when 
$\mathbb{C}^{2^\ell}\otimes\mathbb{C}^2$ is replaced by the subspace~$\gkpcode{}{}{2^\ell}\otimes\mathbb{C}^2 \subset L^2(\bbR)\otimes\bbC^2$ of a qubit-oscillator system.
In other words, the code is a product of an ideal, i.e., infinitely squeezed GKP code~$\gkpcode{}{}{2^\ell}$ (encoding a~$2^\ell$-dimensional qudit) and a ``bare'' (i.e., unencoded) physical qubit. 

Recall that~$\Uelem^{n,m}$ denotes the set of elementary unitary operations~\eqref{it:secondlinearoptics}--\eqref{it: qubit controlled ops} acting on the system~$L^2(\mathbb{R})^{\otimes n} \otimes (\mathbb{C}^2)^{\otimes m}$ (see Section~\ref{sec:intro}). 

\begin{lemma}[Implementation of bit-manipulation maps in ideal GKP codes] \label{lem: gkpLSB}
Let~$\ell\in \mathbb{N}$ be an integer. Then the following holds.
\begin{enumerate}[i)]
\item The (logical) unitaries~$\qCX{\ell}$ and~$\lsb{\ell}$ for the code~$\gkpcode{}{}{2^\ell} \otimes \mathbb{C}^2$ can be implemented exactly by a unitary circuit on~$L^2(\mathbb{R}) \otimes \mathbb{C}^2$ using one, respectively at most~$\ell + 6$ elementary operations in~$\Uelem^{1,1}$.
\item The (logical) isometry~$\embedmap_\ell:\gkpcode{}{}{2^\ell}\rightarrow \gkpcode{}{}{2^{\ell+1}}$ has an exact implementation by a unitary circuit on~$L^2(\mathbb{R})$ using one elementary operation in~$\Uelem^{1,0}$.
\end{enumerate}
See Table~\ref{fig:Gellimplement} for the corresponding circuits.
    \begin{table}[H] 
    \centering
    \renewcommand{\arraystretch}{1.3} 
    \rowcolors{1}{white}{gray!15} 
    \centering
    \resizebox{\textwidth}{!}{\begin{tabular}{c|c}
    \textbf{\makecell{logical \\ operation}} & \textbf{implementation} \\
           \raisebox{0.5cm}{$\qCX{\ell}$} & \includegraphics[scale= 0.9]{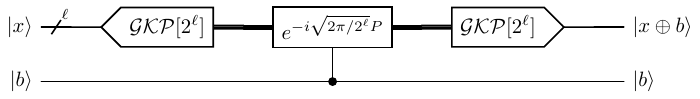}\\
            \raisebox{0.5cm}{$\lsb{\ell}$}
            & \includegraphics[scale= 0.9]{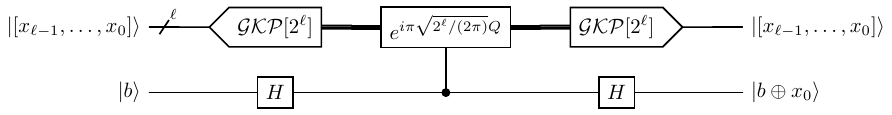}\\
            \raisebox{0.2cm}{$\embedmap_\ell$}
            & \includegraphics[scale= 0.9]{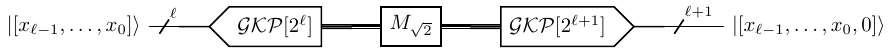}
        \end{tabular}}
        \caption{Exact implementations of the unitary (logical) operations~$\qCX{\ell}$
         and~$\lsb{\ell}$ 
        in the code~$\gkpcode{}{}{2^\ell}\otimes \mathbb{C}^2$,  and of the isometry~$\embedmap_\ell$ embedding the code~$\gkpcode{}{}{2^\ell}$ into the code~$\gkpcode{}{}{2^{\ell+1}}$. 
        Here we write~$M_{\sqrt{2}}= e^{-i (\log 2) (QP +PQ)/4}$ for brevity.
         All these implementations are defined by unitary maps on~$L^2(\mathbb{R})\otimes\mathbb{C}^2$, except for~$\embedmap_\ell$ which is defined on~$L^2(\mathbb{R})$. The associated logical operations constitute the set~$\cG_\ell$.   To obtain circuits composed of elementary gates, the controlled phase shift in the implementation of~$\lsb{\ell}$ needs to be decomposed, see the proof of Lemma~\ref{lem: gkpLSB}.
 \label{fig:Gellimplement}}
\end{table}

\end{lemma}
\begin{proof}
The circuit identities in Table~\ref{fig:Gellimplement} are straightforward to verify (see e.g., \cite[Section 1.1.4]{brenner2024factoring} for the case of the~$\lsb{\ell}$-gate). The circuits for~$\qCX{\ell}$ and~$\mathsf{Embed_\ell}$ consist of a single constant-strength unitary, thus giving the desired implementation.

The circuit for 
$\lsb{\ell}$ involves a qubit-controlled phase space displacement where the displacement is of the form
\begin{align}
    W((0,\alpha))=e^{i\alpha Q}
\end{align}
(see Eq.~\eqref{eq:displacementrelationsetup}) with~$ \alpha = \pi \sqrt{2^\ell/(2\pi)}$. This is not of constant strength for general~$\ell\in\bbN$ but can be decomposed into constant-strength operations in the set~$\Uelem^{1,0}$ as follows. 
Define the squeezing operator~$M_\alpha= e^{-i (\log \alpha) (QP +PQ)/2}$. 
We have 
\begin{align}
e^{i \alpha Q}  = (M_{\alpha'}^\dagger)^{n} e^{i Q} \left(M_{\alpha'}\right)^{n} \qquad \textrm{ where } 
\alpha' = e^{\frac{\log\alpha}{\ceil{|\log \alpha|}}}
\text{ and }
n = \left\lceil| \log \alpha | \right\rceil \ . \label{eq: squeezing trick}
\end{align} 
Eq.~\eqref{eq: squeezing trick} gives a decomposition of the unitary~$e^{i\alpha Q}$ into~$2n + 1$ constant-strength unitaries. Notice that the circuit for~$\lsb{\ell}~$ uses two additional Hadamard gates, giving a total of $2n+3$ elementary operations in $\Uelem^{1,1}$. The claim follows from the fact that 
\begin{align}
    I \otimes\proj{0} + e^{i\alpha Q}\otimes \proj{1}  = ( M_{\alpha'}^\dagger  \otimes I)^n\left(I \otimes \proj{0} +  e^{iQ}\otimes \proj{1}\right)\left(M_{\alpha'} \otimes I\right)^n 
\end{align}
where each factor on the rhs. is a constant-strength unitary and 
~$n = \big\lceil |\log(\pi \sqrt{2^\ell/(2\pi)} )| \big\rceil  \le (\ell -1)/2 + 2$ which gives $2n+3 \leq \ell + 6$. 
\end{proof}

An immediate consequence of Lemma~\ref{lem: gkpLSB} is a circuit implementation of the bit-transfer unitary~$\bittransfer{\ell}{j}$ introduced in Section~\ref{sec:multiquditsystemlogicalderived}.
Namely, we have the following.

\begin{lemma}[Implementation of bit-transfer in ideal GKP codes]\label{lem:GKPbittransfer}
Let~$\ell \in \mathbb{N}$ be an integer and~$j \in \{0,\dots,\ell-1\}$.
Then there is a unitary circuit~$W_{\bittransfer{\ell}{j}}$ on~$L^2(\mathbb{R})^{\otimes 2}\otimes (\mathbb{C}^2)^{\otimes 2}$ such that the following holds.
\begin{enumerate}[i)]
    \item \label{it:transclaim1} The unitary~$W_{\bittransfer{\ell}{j}}$ consists of fewer than~\blue{$85\ell^2$} operations in~$\Uelem^{2,2}$ and satisfies  
    \begin{align}
        &W_{\bittransfer{\ell}{j}} \left(\ket{\gkp(x)_{2^\ell}} \otimes \ket{\gkp(0)_2}\otimes\ket{0} \otimes \ket{b} \right)\\
        &\qquad\qquad=  |\gkp(x- 2^j (b \oplus x_j))_{2^\ell}\rangle\otimes \left( \ket{\gkp(0)_2}\otimes\ket{0}\right)  \otimes \ket{b \oplus x_j} 
    \end{align} for all ~$x=[x_{\ell-1},\ldots,x_0]\in\mathbb{Z}_{2^\ell}$ and~$b\in \{0,1\}$. That is, the unitary~$W_{\bittransfer{\ell}{j}}$ implements the bit-transfer unitary~$\bittransfer{\ell}{j}: \mathbb{C}^{2^\ell} \otimes \mathbb{C}^2 \rightarrow \mathbb{C}^{2^\ell} \otimes \mathbb{C}^2$ exactly on 
    the space 
    \begin{align}
        \gkpcode{}{}{2^\ell} \otimes \mathbb{C} \left( \ket{\gkp(0)_2}\otimes\ket{0}\right) \otimes \mathbb{C}^2 \, .
    \end{align}
\item \label{it:transclaim2} Define~$\zeta = \sqrt{\pi} \cdot 2^{(\ell-1)/2}$. There exists a family of~\blue{$L \le 36 \ell$} unitaries~$\{W^{(t)}\}_{t=1}^{L}$ with
\begin{align}
    W^{(t)} &\in \Uelem^{2,2}(2,\zeta)\qquad\textrm{ for all }\qquad t\in \{1,\ldots,L\}\ 
\end{align}
such that~$W_{\bittransfer{\ell}{j}} = W^{(L)} \cdots W^{(1)}$. (Recall from Eq.~\eqref{eq:mygroupdefinitionmultiqubitbounded} that~$\Uelem^{2,2}(\alpha,\zeta)$
is the set of elementary unitary operations on~$L^2(\mathbb{R})^{\otimes 2}\otimes (\mathbb{C}^2)^{\otimes 2}$ with squeezing and displacements bounded by~$\alpha\geq 1$ and~$\zeta\geq 1$, respectively.)
\end{enumerate}
\end{lemma}
\begin{proof}
    The physical realization~$W_{\bittransfer{\ell}{j}}$ of the bit-transfer~$\bittransfer{\ell}{j}$ is obtained by applying Lem\-ma~\ref{lem: CjellX circuit} (cf. Fig.~\ref{fig:bitextractionimplement}) which gives a circuit for the latter, and replacing each bit-manipulation map~$M\in \cG(\ell)$
    by the implementation provided in Lemma~\ref{lem: gkpLSB} (cf. Table~\ref{fig:Gellimplement}). 
    We note here that since~$\gkpcode{}{}{2^k}$ is a code defined on a single oscillator for any~$k\in\mathbb{N}$, qudits of different dimensions occurring e.g., in the course of applying  the circuit for the bit-transfer unitary~$\bittransfer{\ell}{j}$ (see Lemma~\ref{lem: CjellX circuit}) can be 
    embedded into a single oscillator. In other words, these substitutions
    give an implementation of~$\bittransfer{\ell}{j}$
    whereby the circuit~$V^j_\ell$ introduced in  Lemma~\ref{lem: CjellX circuit}
    is turned into the circuit~$W_{\bittransfer{\ell}{j}}$ acting on~$L^2(\mathbb{R})\otimes  L^2(\mathbb{R})\otimes \mathbb{C}^2\otimes \mathbb{C}^2 := SBQQ'$ (i.e., two oscillators~$S,B$ and two qubits~$Q,Q'$). This means that while~$\bittransfer{\ell}{j}$ acts on two systems, the implementation~$W_{\bittransfer{\ell}{j}}$ acts on four systems. 
    Namely, the implementation~$W_{\bittransfer{\ell}{j}}$ acts by bit-transfer inside the  subspace
    \begin{align}
    \gkpcode{}{}{2^\ell}\otimes  \mathbb{C}\left(\ket{\gkp(0)_2}\otimes\ket{0}\right) \otimes \mathbb{C}^2 \subset L^2(\mathbb{R})^{\otimes 2}\otimes (\mathbb{C}^2)^{\otimes 2} \ .
    \end{align}

    To bound the total number of elementary operations  we note that the circuit~$V_\ell^j$ is composed of at most~$12\ell -4$ basic bit-manipulation operations belonging to~$\cG(\ell)$ (see Lemma~\ref{lem: CjellX circuit}), and each 
    such operation can be implemented in the code of interest using at most~$\ell + 6$ elementary physical operations (see Lemma~\ref{lem: gkpLSB}).
    Therefore, we need at most~$(12\ell -4)(\ell+6)  \leq 12\ell \cdot 7\ell< 85\ell^2$ elementary operations. This shows Claim~\eqref{it:transclaim1}.

    Claim~\eqref{it:transclaim2} follows immediately from Lemma~\ref{lem: gkpLSB} by observing that the physical implementation of each bit manipulation map can be written as a circuit of at most three elements contained in the set $\cG(\ell)$ in Eq.~\eqref{eq:bitmanipulationcollection}.
    Therefore we can bound~$L \le 3(12\ell-4)\le 36\ell$.
\end{proof}

We use the diagram given in Table~\ref{fig:bitextractionunitary_implementation} to represent~$W_{\bittransfer{\ell}{j}}$.

\begin{table}[H]
    \centering
    \renewcommand{\arraystretch}{1.3}   
    \setlength{\tabcolsep}{10pt}        
    \rowcolors{1}{white}{gray!15}       
        \begin{tabular}{c|c|c}
        \textbf{diagram} & \textbf{gate}  &\textbf{type}\\
         \raisebox{-10pt}{\includegraphics{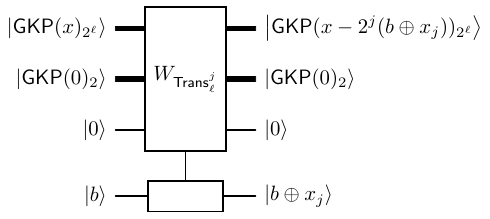}} 
             & \raisebox{1.5cm}{$W_{\bittransfer{\ell}{j}}$}& \raisebox{1.5cm}{unitary} \\
        \end{tabular}%
    \caption{Diagram of the implementation~$W_{\bittransfer{\ell}{j}}$ given in Lemma~\ref{lem:GKPbittransfer} of the bit-transfer unitary~$\bittransfer{\ell}{j}$. It moves the $j$-th bit of the encoded $\ell$-bit string~$x$ into a physical qubit.}
    \label{fig:bitextractionunitary_implementation}
\end{table}

\subsection{Two-qubit gates\label{sec:GKPimplementmultiqubit}}
The implementations of bit-transfer unitaries obtained in Section~\ref{lem: gkpLSB} can be used to implement logical unitaries following the procedure given in Section~\ref{sec:multiqubitbitextraction}. 
Namely, we have the following results.

\begin{lemma}[Implementation of two-qubit unitaries in ideal GKP codes] \label{lem: multiqubitGKPimplement} 
    Let~$\ell\geq 2$ be an integer. Consider a ``logical''~$\ell$-qubit system denoted 
    \begin{align}
    A_0\cdots A_{\ell-1}=(\mathbb{C}^2)^{\otimes \ell}\ .
    \end{align}
    Let~$j<k$ and~$j,k\in \{0,\ldots,\ell-1\}$ be arbitrary.
    Let~$U := U_{A_j A_k}: \left(\mathbb{C}^2\right)^{\otimes 
    \ell} \rightarrow\left(\mathbb{C}^2\right)^{\otimes 
    \ell}$ be a unitary 
    acting non-trivially only on the two qubits~$A_jA_k$.  
    Let 
    \begin{align}
    SBQQ_1Q_2 = L^2(\mathbb{R}) \otimes \left(L^2(\mathbb{R})\otimes (\mathbb{C}^2)^{\otimes 3}\right)
    \end{align}
    be a system consisting of a~$2^\ell$-dimensional code space~$S$,~$1$ auxiliary oscillator~$B$ and~$3$ auxiliary qubits~$Q,Q_1,Q_2$. 
    Then there is a unitary circuit~$W_U$ on~$SBQQ_1Q_2$ consisting of at most~\blue{$340\ell^2$} elementary operations in~$\Uelem^{2,3}$
            such that
            \begin{align}
            W_U\left(\ket{\gkp(x)_{2^\ell}}\otimes \left(\ket{\gkp(0)_2}\otimes\ket{0}^{\otimes 3}\right)\right)
            = \encodergkp{}{}{2^\ell}\left( J_\ell U J_\ell^{-1}\ket{x} \right)  \otimes \left(\ket{\gkp(0)_2}\otimes\ket{0}^{\otimes 3}\right)
            \end{align}
            for all~$x\in \{0,\ldots,2^\ell-1\}$. 
    In other words, the circuit~$W_U$ implements the gate~$U$ exactly on the subspace
    \begin{align}
        \gkpcode{}{}{2^\ell}\otimes\mathbb{C}(\ket{\Psi}) \subset L^2(\mathbb{R}) \otimes (L^2(\mathbb{R})\otimes (\mathbb{C}^2)^{\otimes 3}) \ ,
    \end{align}
    where~$\ket{\Psi}=\ket{\gkp(0)_2}\otimes\ket{0}^{\otimes 3}\in L^2(\mathbb{R})\otimes(\mathbb{C}^2)^{\otimes 3}$, i.e., the auxiliary oscillator is in the ideal GKP state~$\ket{\gkp(0)_2}$ and the three auxiliary qubits are in the state~$\ket{0}^{\otimes 3}$.
    \end{lemma}

We show an illustration of the implementation of a two-qubit unitary in Fig.~\ref{fig:W_Uimplement}.
We give a generalization of Lemma~\ref{lem: multiqubitGKPimplement} in Appendix~\ref{sec:generalizationmultiple} (see Lemma~\ref{lem:multiqubitGKPimplementbipartite}) where~$2\ell$ qubits are encoded into two ideal GKP codes~$(\gkpcode{}{}{2^\ell})^{\otimes 2}$.

\begin{proof} 
The  circuit~$W_U$ is obtained from the circuit~$V_U$ introduced in Lemma~\ref{lem:general2qubitbitextract} by replacing each bit-transfer map~$\bittransfer{\ell}{j}$ by its implementation~$W_{\bittransfer{\ell}{j}}$ given in Lemma~\ref{lem:GKPbittransfer} (and similarly for the adjoint of the bit-transfer map). (We note the two-qubit gate~$U$ in the circuit~$V_U$ introduced in Lemma~\ref{lem:general2qubitbitextract}
 acts as a two-qubit gate in the circuit~$W_U$).
Concretely, 
the circuit~$W_U$ is (leaving out identities on the remaining systems)  
\begin{align}
    W_{U} = \left(W_{\bittransfer{\ell}{j}}\right)_{SBQQ_1}^\dagger  \left(W_{\bittransfer{\ell}{k}}\right)^\dagger_{SBQQ_2} U_{Q_1Q_2} \left(W_{\bittransfer{\ell}{k}}\right)_{SBQQ_2} \left(W_{\bittransfer{\ell}{j}}\right)_{SBQQ_1}\,. \label{eq:WUimplement}
\end{align}
Each transfer gate uses the auxiliary oscillator~$B$ and the auxiliary qubit~$Q_3$ as per Lemma~\ref{lem:GKPbittransfer}.
By reusing the auxiliary oscillator~$B$ for each circuit realizing different bit-transfers (and their adjoints), we obtain a circuit~$W_U$ on~$L^2(\mathbb{R})\otimes (L^2(\mathbb{R})\otimes (\mathbb{C}^2)^{\otimes 3})$ with the desired action. 

The circuit~$W_U$ uses a single two-qubit gate, which is an elementary operation in our setup. It further uses four implementations of bit-transfer unitaries and their adjoints, each of which can be implemented using fewer than~$85\ell^2$ elementary operations by Lemma~\ref{lem:GKPbittransfer}.
Therefore, the total number of elementary operations necessary to implement~$W_U$ is less than~$4\cdot 85\ell^2 + 1 = 340 \ell^2 + 1$, i.e., it is less than or equal to~$340\ell^2$.
\end{proof}

\begin{figure}[H]
\centering
\includegraphics[width=\textwidth]{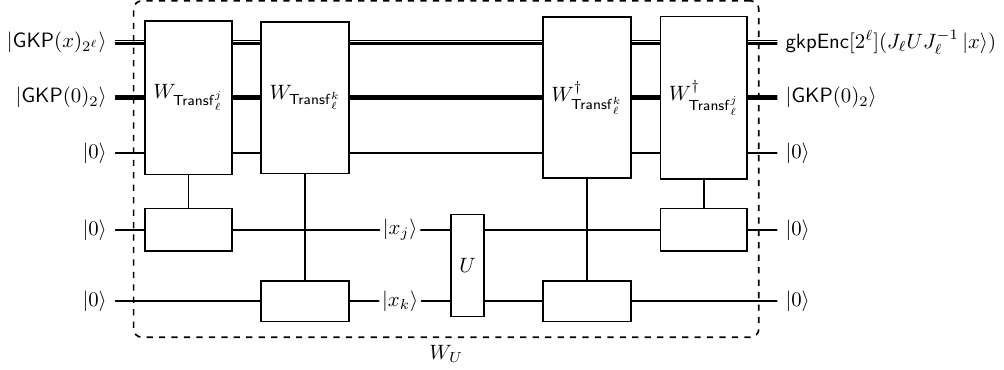}
\caption{
Illustration of the implementation~$W_U$ in Lemma~\ref{lem: multiqubitGKPimplement}  (cf. Eq.~\eqref{eq:WUimplement}) of a logical two-qubit unitary~$U=U_{A_jA_k}$ acting on the~$j$-th and~$k$-th qubits for~$j,k \in \{0,\dots, \ell-1\}$. 
}
\label{fig:W_Uimplement}
\end{figure}

\section{Logical gate errors of implementations in \\ approximate GKP codes} \label{sec:logicalgateerrorsofimplementation}

In this section we show that the circuit identities derived in Section~\ref{sec:GKPimplementmultiqubit} for ideal GKP codes are approximately satisfied when working with finitely squeezed GKP states.
More concretely, we quantify the error arising from the use of approximate GKP states
$\ket{\gkp_{\kappa,\Delta}^{\varepsilon}(j)_{2^\ell}}$, $j\in\{0,\ldots,2^\ell-1\}$
defined in Eq.~\eqref{eq:truncatedvarsvdkap} by bounding the logical gate error of several implementations of logical unitaries in the approximate GKP codes
$\gkpcode{\kappa,\Delta}{\varepsilon}{2^\ell}=\mathsf{span}\{\gkp^{\varepsilon}_{\kappa,\Delta}(j)_2^\ell\}_{j\in\{0,\ldots,2^\ell-1\}}$.
Recall the definition of the logical gate error in Appendix~\ref{app:logicalgateerrorproperties}.

In Section~\ref{sec:gateerrorbitmanmain} we derive bounds on the logical gate error of implementations of basic bit-manipulation maps as well as bit-transfer unitaries in symmetrically squeezed GKP codes. In Section~\ref{sec:gaterrorbitCliffordmain} we use these results together with subadditivity to establish bounds on the logical gate error of implementations of two- and multi-qubit unitaries, proving Theorem~\ref{thm:main}.
In Section~\ref{sec:gateerrorscliffords} we apply these bounds to implementations of Cliffords.

\subsection{Basic bit-manipulation and bit-transfer maps  \label{sec:gateerrorbitmanmain}}

In the following, we upper bound the logical gate error for the implementations of the bit manipulation maps given in Table~\ref{fig:Gellimplement}.
Recall we denote by~$\Uelem^{n,m}$ the set of elementary unitary operations in~\eqref{it:secondlinearoptics}--\eqref{it: qubit controlled ops} acting on the system~$L^2(\mathbb{R})^{\otimes n} \otimes (\mathbb{C}^2)^{\otimes m}$.

\begin{lemma}[Implementation of 
$\qCX{\ell}$,~$\lsb{\ell}$ and~$\embedmap_\ell$]
  \label{lem: universal bound special} Let~$\ell \in \mathbb{N}$. Let~$\kappa,\Delta \in (0,1/4)$ and~$\varepsilon \in(0, 2^{-(\ell+1)}]$. 
Consider the space~$\gkpcode{\kappa,\Delta}{\varepsilon}{2^\ell} \otimes \mathbb{C}^{2} \simeq \bbC^{2^\ell} \otimes \bbC^2$. Then we have the following.
\begin{enumerate}[i)]
    \item The physical implementation of the unitary~$\qCX{\ell}$ on~$ \mathbb{C}^{2^\ell} \otimes \mathbb{C}^2$given by 
    \begin{align}
        W_{\qCX{\ell}} = \mathsf{ctrl} e^{-i\sqrt{2\pi/2^{\ell}}P}
    \end{align} can be realized by a single elementary operation in~$\Uelem^{1,1}$. It satisfies 
    \begin{align}
        \max\left\{\gateerror_{  \gkpcode{\kappa,\Delta}{\varepsilon}{2^\ell} \otimes \mathbb{C}^{2} }( W_{\qCX{\ell}}, \qCX{\ell}) , \, 
        \gateerror_{  \gkpcode{\kappa,\Delta}{\varepsilon}{2^\ell} \otimes \mathbb{C}^{2} }( W_{\qCX{\ell}}^\dagger, \qCX{\ell}^\dagger)  \right\} &\le 8 \kappa \ .
    \end{align}
    \item The physical implementation of the unitary~$\lsb{\ell}$ on~$ \mathbb{C}^{2^\ell} \otimes \mathbb{C}^2$ given by 
    \begin{align}
        W_{\lsb{\ell}}&=(I\otimes H)(M_{\alpha'}^\dagger \otimes I)^n \mathsf{ctrl}e^{iQ}(M_{\alpha'}\otimes I)^n (I\otimes H) \label{eq:expansionlsbimplement}
    \end{align} 
    where 
    \begin{align}
        \label{eq:alphadef}
    \alpha'&=e^{\frac{\log \alpha}{\left\lceil |\log \alpha|\right\rceil}}\qquad\textrm{ and }\qquad n=\left\lceil |\log \alpha|\right\rceil\qquad\textrm{ with }\qquad \alpha=\sqrt{\pi}2^{(\ell-1)/2}
    \end{align}
     defines
     a circuit~$W_{\lsb{\ell}}=W_T\cdots W_1$ composed of~$T \le \ell + 6$~elementary gates in~$\Uelem^{1,1}$. It satisfies
     \begin{align}
        \max\left\{\gateerror_{ \gkpcode{\kappa,\Delta}{\varepsilon}{2^\ell}\otimes \mathbb{C}^{2} }( W_{\lsb{\ell}}, \lsb{\ell}) , \, 
        \gateerror_{ \gkpcode{\kappa,\Delta}{\varepsilon}{2^\ell}\otimes \mathbb{C}^{2} }( W_{\lsb{\ell}}^\dagger, \lsb{\ell}^\dagger)  \right\}
        &\le 16 \cdot  2^\ell \Delta + 32 (\Delta/\varepsilon)^2 \ .
     \end{align}
\end{enumerate}
Consider the spaces~$\gkpcode{\kappa,\Delta}{\varepsilon}{2^\ell} \simeq \bbC^{2^\ell}$ and~$\gkpcode{\kappa,\Delta}{\varepsilon}{2^{\ell+1}} \simeq \bbC^{2^{\ell+1}}$. Then we have the following. 
\begin{enumerate}
    \setlength{\leftskip}{-0.5em}
    \item[iii)] The isometry~$\embedmap_\ell: \mathbb{C}^{2^\ell} \rightarrow \mathbb{C}^{2^{\ell+1}}$ can be implemented by the unitary~$W_{\embedmap_\ell} = M_{\sqrt{2}}$. It uses a single elementary operation in~$\Uelem^{1,0}$ and satisfies 
    \begin{align}
        \max\Big\{
            &\gateerror_{\gkpcode{\kappa,\Delta}{\varepsilon}{2^\ell}, \gkpcode{\kappa,\Delta}{\varepsilon}{2^{\ell+1}}} \left(W_{\embedmap_\ell}, \embedmap_\ell\right)  , \\
        &\gateerror_{\gkpcode{\kappa,\Delta}{\varepsilon}{2^{\ell+1}},\gkpcode{\kappa,\Delta}{\varepsilon}{2^\ell}} \left(W_{\embedmap_\ell}^\dagger, \embedmap_\ell^\dagger \right) \Big\}= 0\, .
    \end{align}
\end{enumerate}
\end{lemma}

We prove Lemma~\ref{lem: universal bound special} in Appendix~\ref{sec: lem: universal bound special-proof}.

An immediate consequence of the bounds on the logical gate error of the basic bit-manipulation maps in Lemma~\ref{lem: universal bound special} is the following result on the physical implementation of bit-transfer unitaries~$\bittransfer{\ell}{j}$.

\begin{lemma}[Implementation of bit-transfer unitaries]\label{lem:GKPbittransferimplement}
Let~$\ell \in \mathbb{N}$ be an integer and~$j \in \{0,\dots,\ell-1\}$. Let~$\kappa\in(0,1/4)$.
Then there is a unitary circuit~$W_{\bittransfer{\ell}{j}}$ on~$L^2(\mathbb{R})\otimes \mathbb{C}^2 \otimes L^2(\mathbb{R}) \otimes \mathbb{C}^2$ 
which consists of fewer than~\blue{$85\ell^2$} operations in~$\Uelem^{2,2}$. It satisfies  
    \begin{align}
        \max\left\{\gateerror_\cL\big(W_{\bittransfer{\ell}{j}}, \bittransfer{\ell}{j}\big), \gateerror_\cL\big((W_{\bittransfer{\ell}{j}})^\dagger, (\bittransfer{\ell}{j})^\dagger\big)\right\}\le \blue{96} \ell \kappa\,,
    \end{align}
where the code space is given by~$\cL = \gkpcode{\kappa}{\star}{2^\ell}\otimes \mathbb{C}^2 \otimes \mathbb{C}\big(|\gkp_{\kappa}^{2^{-(\ell+1)}}(0)_2\rangle\otimes \ket{0}\big)$.
\end{lemma} 
\begin{proof}
    Let~$W_{\bittransfer{\ell}{j}}$ be the circuit implementation of~$\bittransfer{\ell}{j}$ introduced in Lemma~\ref{lem:GKPbittransfer}.
    By the proof of Lemma~\ref{lem:GKPbittransfer} (cf. Lemma~\ref{lem: CjellX circuit}), the implementation~$W_{\bittransfer{\ell}{j}}$ is constructed by replacing at most
    \begin{align}
        12\ell - 4 \le 12 \ell
        \label{eq:aux3}
    \end{align}
    basic bit-manipulation maps from the set~$ \cG(\ell) =\cup_{k=1}^{\ell} (\cG_k\cup \cG_k^\dagger )$ (see Eq.~\eqref{eq:defcG}) by suitable implementations.
    By Lemma~\ref{lem: universal bound special}, there is a physical implementation~$W_V$ for every basic bit-manipulation map~$V \in \cG_k\cup \cG_k^\dagger$ such that 
    \begin{align}
     \gateerror_{\cL_{in}, \cL_{out}}(W_V, V) \le \max\left\{8\kappa, 16 \cdot 2^k \Delta + 32(\Delta/\varepsilon)^2\right\}\, , \label{eq:aux1}
    \end{align} 
    where if~$V = \embedmap_k$ the logical spaces are~$\cL_{in} = \gkpcode{\kappa,\Delta}{\varepsilon}{2^k}$ and~$ \cL_{out} = \gkpcode{\kappa,\Delta}{\varepsilon}{2^{k+1}}$, otherwise we have~$\cL_{in} = \cL_{out} = \gkpcode{\kappa,\Delta}{\varepsilon}{2^k} \otimes \mathbb{C}^2$.

    We have~$k \in \{0,\dots, \ell-1\}$ and, by definition of~$\gkpcode{\kappa}{\star}{2^k}$, $\Delta = \kappa/(2\pi \cdot 2^{\ell})$ and~$\varepsilon = 2^{-(\ell+1)}$ (see Section~\ref{sec:intro}).
    Therefore, we may upper bound Eq.~\eqref{eq:aux1} as 
    \begin{align}
        \gateerror_{\cL_{in}, \cL_{out}}(W_V, V) &\le  \max_{k \in \{1,\dots,\ell\}}  \max \left\{8\kappa,  16\cdot 2^k \kappa/(2\pi \cdot 2^\ell)+ 32(\kappa/\pi)^2\right\}\\
        &=  \max \left\{8\kappa, 16\cdot 2^\ell \kappa/(2\pi \cdot 2^\ell)+ 32(\kappa/\pi)^2\right\}\\
        & \le 8 \kappa \, , \label{eq:boundGellerrorCZ}
    \end{align}
    for~$\kappa\in(0,1/4)$.

    Combining  Eqs.~\eqref{eq:aux3} and~\eqref{eq:boundGellerrorCZ} with the subadditivity of the logical gate error under composition, see~Lemma~\ref{lem: additivity gate error}, gives
    \begin{align}
        \max\left\{\gateerror_\cL\left(W_{\bittransfer{\ell}{j}}, \bittransfer{\ell}{j}\right)\right\}
        &\le 12 \ell \cdot \max_{V \in \cG(\ell)}\gateerror_{\cL_{in}, \cL_{out}}(W_V, V) \\
        &\le 12 \ell \cdot 8\kappa = 96 \ell \kappa\ .
    \end{align}
    The bound for the adjoint follows by the same arguments, giving the claim.
\end{proof}

\subsection{Multi-qubit gates \label{sec:gaterrorbitCliffordmain}}

Here we prove the main result stated as Theorem~\ref{thm:main}. 
We start by analysing implementations of two-qubit unitaries in symmetrically squeezed GKP codes. We then generalize the result to implementations of (multi-qubit) circuits composed of~$T$ two-qubit gates.

\begin{theorem}[Implementation of a two-qubit gate in symmetrically squeezed GKP codes] \label{lem:uniformboundgeneralmultiqubit}
    Let~$\ell\geq 2$ be an integer. Let~$\kappa \in (0,1/4)$. Consider a ``logical''~$\ell$-qubit system denoted 
    \begin{align}
    A_0\cdots A_{\ell-1}=(\mathbb{C}^2)^{\otimes \ell}\ .
    \end{align}
    Let~$j<k$ and~$j,k\in \{0,\ldots,\ell-1\}$ be arbitrary.
    Let~$U := U_{A_j A_k}: \left(\mathbb{C}^2\right)^{\otimes 
    \ell} \rightarrow\left(\mathbb{C}^2\right)^{\otimes 
    \ell}$ be a unitary 
    acting non-trivially only on the two qubits~$A_jA_k$.  
    Then, there exists a physical unitary implementation~$W_U$ acting on 
    the Hilbert space~$\cH = L^2(\mathbb{R}) \otimes L^2(\mathbb{R}) \otimes (\mathbb{C}^{2})^{\otimes 3}$ of the logical unitary~$U$ using at most~\blue{$340\ell^2$} elementary operations in~$\Uelem^{2,3}$. It satisfies
\begin{align}
    \gateerror_\cL( W_U, J_{\ell} U J_{\ell}^{-1}) \le \blue{400}  \ell \kappa \, 
\end{align}
with code space~$\cL = \gkpcode{\kappa}{\star}{2^\ell}\otimes \mathbb{C}\big(|\gkp_{\kappa}^{2^{-(\ell+1)}}(0)_2\rangle\otimes \ket{0}^{\otimes 3}\big)$.
\end{theorem}

We give a generalization of Theorem~\ref{lem:uniformboundgeneralmultiqubit} in Appendix~\ref{sec:generalizationmultiple} (see Theorem~\ref{thm: uniform bound general multiqubitbipartite}) where~$2\ell$ qubits are encoded into two symmetrically squeezed GKP codes~$(\gkpcode{\kappa}{\star}{2^\ell})^{\otimes 2}$.

\begin{proof}
    Let~$W_U$ be the unitary implementation of~$U$ introduced in Lemma~\ref{lem: multiqubitGKPimplement}. It can be realized using at most~$340\ell^2$ elementary operations. 
    Moreover, it consists of four physical implementations of a bit-transfer operations of the form~$W_{\bittransfer{\ell}{j}},(W_{\bittransfer{\ell}{j}})^\dagger$ (together with a single two-qubit unitary acting on two-physical qubits which does not contribute to the logical gate error).
    Recall Lemma~\ref{lem:GKPbittransferimplement} which states
    \begin{align}
        \max\left\{\gateerror_\cL\big(W_{\bittransfer{\ell}{j}}, \bittransfer{\ell}{j}\big), \gateerror_\cL\big((W_{\bittransfer{\ell}{j}})^\dagger, (\bittransfer{\ell}{j})^\dagger\big)\right\}\le \blue{96} \ell \kappa\ .
    \end{align}
    This together together with the subadditivity of the logical gate error under composition, see~Lemma~\ref{lem: additivity gate error}, gives the claim
    \begin{align}
        \gateerror_\cL( W_U, J_{\ell} U J_{\ell}^{-1}) \le 4 \cdot \blue{96} \ell \kappa \leq \blue{400}  \ell \kappa \ .
    \end{align}
\end{proof}

With this preparation we can give the proof of the main result of this work, see Theorem~\ref{thm:main}.

\begin{proof}[Proof of Theorem~\ref{thm:main}]
Theorem~\ref{lem:uniformboundgeneralmultiqubit} bounds the logical gate error of the implementation of a single two-qubit gate. The claim follows using the the subadditivity of the logical gate error under composition (see Lemma~\ref{lem: additivity gate error}), that is, using the triangle inequality of the diamond norm $T$ times.
\end{proof}

\subsection{Qudit Clifford-gates \label{sec:gateerrorscliffords}}
In this section we demonstrate how Theorem~\ref{thm:main} (and its extension to the case of~$2\ell$ qubits encoded in two modes, see Theorem~\ref{thm: uniform bound general multiqubitbipartite}) can be leveraged to find implementations of Cliffords.

For completeness, we first introduce qudit Paulis and Cliffords.
Let~$d\ge 2$ be an integer. We define the Pauli operators~$X,Z$ for a qudit as
 \begin{align}
 \begin{aligned}
 X\ket{x}&=\ket{x\oplus 1}\\
 Z\ket{x}&=\omega_d^x \ket{x}
 \end{aligned}\qquad\textrm{ for all }x\in\mathbb{Z}_d\ ,
 \end{align}
 where~$\oplus$ denotes addition modulo~$d$ and~$\omega_d=e^{i 2\pi/d}$ is a~$d$-th root of unity. The Pauli operators~$X,Z$ satisfy the commutation relation 
 \begin{align}
ZX&=\omega_d X Z
 \end{align}
 and  generate the single-qudit Pauli group.
The~$n$-qudit Pauli group acting on the tensor product space~$(\mathbb{C}^d)^{\otimes n}$ is the group~$\langle X_j,Z_j\ |\ j=1,\ldots,n\rangle$, i.e., it is the group generated by single-qudit Pauli-operators on each qudit~$j\in \{1,\ldots,n\}$. The~$n$-qudit Clifford group is the normalizer of the~$n$-qudit Pauli group in the unitary group. It is given by
 \begin{align}
 \langle \left\{\Pgate_j,\Fgate_j\ |\ j=1,\ldots,n\right\}\cup \{\CZ_{j,k}\ |\ j\neq k, j,k\in \{1,\ldots,n\}\} \rangle\, ,
\end{align}
where~$\Pgate_j$ is the phase gate in Eq.~\eqref{eq:defPgate} acting on the~$j$-th qubit,~$\Fgate_j$ is the Fourier transform in Eq.~\eqref{eq:Fgate} acting on the~$j$-th qubit and~$\CZ_{j,k}$ is the two-qudit controlled phase gate in Eq.~\eqref{eq:CZgate} with the~$j$-th and the~$k$-th qubits the control and target, respectively.

To generate a universal gate set, we need additional gates. We use the phase gate~$Z(e^{i\theta})$ defined for~$\theta\in\mathbb{R}$ by its action
\begin{align}
Z(e^{i\theta})\ket{x}&=e^{i\theta x}\ket{x}\qquad\textrm{ for all }x\in\mathbb{Z}_d
\end{align}
on computational basis states. We note that~$Z(e^{i\theta})$ is a power of the Pauli-$Z$ operator whenever~$\theta$ is an integer multiple of~$2\pi /d$.  

The generators of the qudit Clifford group, as well as the phase gate~$Z(e^{i\theta})$ are illustrated in Fig.~\ref{tab:singlequditops}.
\begin{table}[H]
\renewcommand{\arraystretch}{1.3} 
\setlength{\tabcolsep}{10pt} 
\rowcolors{1}{white}{gray!15} 
    \centering
    \begin{center}
    \begin{tabular}{c|c|c}
        \textbf{diagram} & \textbf{gate} &  \textbf{type of unitary} \\
         \includegraphics{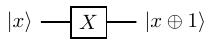}  &  \raisebox{0.25cm}{modular  shift ~$X$}  & \raisebox{0.25cm}{Pauli}\\
        \includegraphics{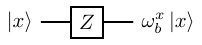}  &\raisebox{0.25cm}{ modular phase shift~$Z$} & \raisebox{0.25cm}{Pauli}\\
                \includegraphics{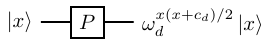}  & \raisebox{0.25cm}{phase gate~$\Pgate$}      & \raisebox{0.25cm}{Clifford} \\ 
                            \raisebox{0.0cm}{\includegraphics{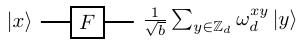}}  &  \raisebox{0.25cm}{Fourier transform~$\Fgate$} & \raisebox{0.25cm}{Clifford} \\
                                             \raisebox{0.0cm}{\includegraphics{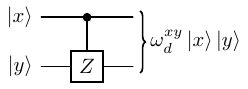}}  &    \raisebox{0.5cm}{controlled modular phase  ~$\cZ$} & \raisebox{0.5cm}{Clifford} \\
                    \includegraphics{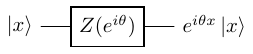}  &  \raisebox{0.25cm}{phase shift~$Z(e^{i\theta})$} & \raisebox{0.25cm}{Clifford for~$\theta \in (\pi/d) \mathbb{Z}$}                \end{tabular}        
                    \end{center}
           \caption{Action of single- and two-qudit operations
    on an orthonormal basis~$\{\ket{x}\}_{x\in\mathbb{Z}_d}$ of~$\mathbb{C}^d$. Here~$c_d = d \mod 2$.}\label{tab:singlequditops}
\end{table}
In the special case~$d = 2^\ell$, we give implementations of a complete set of generators of the Clifford group in terms of qubit-oscillator operations. More precisely, we have the following.

\begin{corollary}[Implementations of Cliffords in symmetrically squeezed GKP codes]\label{cor:gaterrorFPZ}
    Let~$\ell \in \mathbb{N}$ and~$\kappa \in (0,1/4)$. Then there is a universal constant~$C>0$ (independent of~$\ell$) such that following holds:
    \begin{enumerate}[i)]
        \item \label{it:claimfirst} Let~$U \in \{\Fgate, \Pgate, Z(e^{i\theta})\}$ be a unitary acting on~$\mathbb{C}^{2^\ell}$ as in Table~\ref{tab:singlequditops}. Then there exists a physical unitary implementation~$W_U$ acting on 
        the Hilbert space~$L^2(\mathbb{R}) \otimes L^2(\mathbb{R}) \otimes (\mathbb{C}^{2})^{\otimes 3}$ which uses~$O(\ell^4)$ elementary operations in~$\Uelem^{2,3}$. It satisfies
    \begin{align}
        \gateerror_\cL( W_U, J_{\ell} U J_{\ell}^{-1}) \le C  \ell^3 \kappa \, ,
    \end{align}
    with code space~$\cL = \gkpcode{\kappa}{\star}{2^\ell}\otimes \mathbb{C}\left(|\gkp_{\kappa}^{2^{-(\ell+1)}}(0)_2\rangle\otimes \ket{0}^{\otimes 3}\right)$.
    \item \label{it:claimsec} Consider the unitary~$\cZ$ acting on~$\mathbb{C}^{2^\ell} \otimes \mathbb{C}^{2^\ell}$ as in Table~\ref{tab:singlequditops}. Then there exists a physical unitary implementation~$W_{\cZ}$ acting on 
    the Hilbert space~$L^2(\mathbb{R})^{\otimes 2} \otimes L^2(\mathbb{R}) \otimes (\mathbb{C}^{2})^{\otimes 3}$ which uses~$O(\ell^4)$ elementary operations in~$\Uelem^{3,3}$. It satisfies
\begin{align}
    \gateerror_\cL\left( W_{\cZ}, J_{\ell}^{\otimes 2} \cZ \left(J_{\ell}^{\otimes 2}\right)^{-1}\right) \le C  \ell^3  \kappa \, ,
\end{align}
    with code space~$\cL = (\gkpcode{\kappa}{\star}{2^\ell})^{\otimes 2}\otimes \mathbb{C}\left(|\gkp_{\kappa}^{2^{-(\ell+1)}}(0)_2\rangle\otimes \ket{0}^{\otimes 3}\right)$.
\end{enumerate}
\end{corollary}

\begin{proof}
    We will show that for~$U \in \{\Fgate, \Pgate, Z(e^{i\theta})\}$ a unitary on~$\mathbb{C}^{2^\ell}$,
    the~$\ell$-qubit unitary~$\widetilde{U} = J_\ell^{-1} U J_\ell$ on~$\left(\mathbb{C}^2\right)^{\otimes \ell}$ has decomposition
    \begin{align}
        \widetilde{U} = \prod_{t=1}^{T} U^{(t)}_{A_jA_k}
    \end{align}
    where~$T = O(\ell^2)$,~$A_j$ is a two-dimensional system for all~$j \in \{0,\dots,\ell-1\}$ and~$U^{(t)}_{A_jA_k}$ is a (two-qubit unitary) acting only on the system~$A_jA_k$ for all~$t \in \{1,\dots, T\}$. 
    
    For~$U = \Fgate$ the decomposition is the standard realization of the quantum Fourier-transform for~$d= 2^\ell$ on~$\ell$-qubits using only (single-qubit) Hadamard-gates,~$\mathsf{SWAP}$-gates and qubit-controlled rotations of the form 
    \begin{align}
        \proj{0} \otimes I + \proj{1} \otimes R(\omega_{2^k}) \qquad \textrm{for} \qquad k \in \{0, \dots, \ell-1\}\, .
    \end{align}
    where we define the single-qubit rotation with angle~$\alpha\in\bbR$ as
    \begin{align}
        \label{eq:1qbitrotation}
        R(e^{i\alpha})\ket{b} = e^{i \alpha b} \ket{b} \qquad\text{ for } b \in \{0,1\} \ .
    \end{align} 
    We refer to~\cite{Nielsen_Chuang_2010} for a detailed discussion.
    For completeness we give explicit decompositions for the remaining unitaries.

    Let~$U = \Pgate$. As~$d=2^\ell$ is even, we have~$c_d=0$. Then we have~$\Pgate\ket{x}=\omega_{2^\ell}^{x^2}\ket{x}$
    which can be rewritten as
    \begin{align}
        \Pgate \ket{x} &= \left( \prod_{j,k=0}^{\ell-1} \omega_{2^\ell}^{2^{j+k}x_j x_k} \right) \ket{x}                    \qquad \text{ for } \qquad x = [x_{\ell-1}, \ldots, x_0]\ .
            \end{align}
            This implies that 
            \begin{align}
                \widetilde{U}&=\prod_{j,k=0}^{\ell-1} U(\omega_{2^\ell}^{2^{j+k}})_{A_jA_k}
            \end{align}
      where we introduced the (two-qubit) unitary 
       ~$U(e^{i\theta})_{A_jA_k}$ acting as 
        \begin{align}
            \begin{matrix}
                U(e^{i\theta})_{A_jA_k}: & \left(\mathbb{C}^2\right)^{\otimes \ell} & \rightarrow &\left(\mathbb{C}^2\right)^{\otimes \ell}\\
                & \ket{x_{\ell-1}}\otimes \cdots \otimes \ket{ x_0} & \mapsto & e^{i\theta x_j x_k} \ket{x_{\ell-1}}\otimes  \dots \otimes \ket{ x_0}\, .
            \end{matrix}
        \end{align}
        Note that if~$j=k$ we can regard~$U(e^{i\theta})_{A_jA_j}$ as a two-qubit unitary.

        Finally, let~$U = Z(e^{i\theta})$. Then we have 
        \begin{align}
            Z(e^{i\theta}) \ket{x} = \left( \prod_{j=0}^{\ell-1} e^{i \theta x_j 2^j } \right) \ket{x}  \qquad \text{ for } \qquad x = [x_{\ell-1}, \ldots, x_0]\ .
        \end{align}
        This implies that 
            \begin{align}
                \widetilde{U}&=\prod_{j=0}^{\ell-1} R(e^{i \theta 2^j })_{A_j} \ ,
            \end{align}
      where the subscript~$A_j$ indicates that the (single-qubit) rotation~$R(e^{i \theta 2^j })$ acts on the~$j$-th qubit, see Eq.~\eqref{eq:1qbitrotation}.

    Similarly, the~$2\ell$-qubit unitary~$\widetilde{U} = \left(J_{\ell}^{\otimes 2}\right)^{-1} \cZ J_{\ell}^{\otimes 2}$ on~$\left(\mathbb{C}^2\right)^{\otimes \ell} \otimes \left(\mathbb{C}^2\right)^{\otimes \ell}$ has a decomposition
    \begin{align}
        \widetilde{U} = \prod_{t=1}^{T} U^{(t)}_{A_jB_k}
    \end{align}
    where~$T = O(\ell^2)$, $A_j$ and $B_k$ are (distinct) two-dimensional systems for all~$j,k \in \{0,\dots,\ell-1\}$ and~$U^{(t)}_{A_jB_k}$ is a two-qubit unitary acting only the system~$A_jB_k$ for~$t \in \{1,\dots, T\}$. 
    A concrete decomposition can be obtained as follows. We have~$\cZ\left(\ket{x} \otimes\ket{y}\right)=\omega_{2^\ell}^{xy}\left(\ket{x} \otimes \ket{y}\right)$ or 
    \begin{align}
        \cZ \left(\ket{x}\otimes \ket{y}\right) =\left(\prod_{j,k=0}^{\ell-1} \omega_{2^\ell}^{i x_j y_k 2^{j+k}}\right) \left(\ket{x}\otimes \ket{y}\right)
        \qquad\text{for all}\qquad x,y \in \mathbb{Z}_{2^\ell} \ .
    \end{align}
    This implies that 
            \begin{align}
                \widetilde{U}&=\prod_{j,k=0}^{\ell-1} U(\omega_{2^\ell}^{2^{j+k}})_{A_jB_k}
            \end{align}
      where we introduced the (two-qubit) unitary 
       ~$U(e^{i\theta})_{A_jB_k}:  \left(\mathbb{C}^2\right)^{\otimes \ell} \otimes \left(\mathbb{C}^2\right)^{\otimes \ell} \rightarrow \left(\mathbb{C}^2\right)^{\otimes \ell} \otimes \left(\mathbb{C}^2\right)^{\otimes \ell}$ acting as 
        \begin{align}
                \ket{x_{\ell-1}}\otimes \dots \otimes \ket{ x_0} \otimes \ket{y_{\ell-1}}\otimes \dots \otimes\ket{ y_0}  & \mapsto  e^{i x_j y_k\theta } \left(\ket{x_{\ell-1}}\otimes  \dots \otimes \ket{ x_0}\otimes \ket{y_{\ell-1}}\otimes \dots \otimes\ket{ y_0}\right)\, .
        \end{align}
    Claim~\ref{it:claimfirst}) follows from Theorem~\ref{thm:main} with $T=O(\ell^2)$. Claim~\ref{it:claimsec}) follows from a straightforward generalization of Theorem~\ref{thm:main} for $2\ell$ qubits encoded into two approximate codes $(\gkpcode{\kappa}{\star}{2^\ell})^{\otimes 2}$ which is obtained by combining Lemma~\ref{lem:multiqubitGKPimplementbipartite}, Theorem~\ref{thm: uniform bound general multiqubitbipartite} and Lemma~\ref{lem: additivity gate error}.
\end{proof}

By combining Corollary~\ref{cor:gaterrorFPZ} with~\cite[Theorem 6.1]{cliffordslinearoptics2025}, we find that there exist physical implementations of a complete generating set of the qudit Clifford group on~$\bbC^d \otimes \bbC^d$
(for~$ d = 2^\ell$) on a code space defined by symmetrically squeezed GKP codes tensored with a fixed auxiliary state, i.e.,
\begin{align}
    \left(\gkpcode{\kappa}{\star}{2^\ell}\right)^{\otimes 2} \otimes \left(\ket{\gkp_{\kappa}^{2^{-(\ell+1)}}(0)_2} \otimes \ket{0}^{\otimes 3} \right) \ .
\end{align}
That is, the auxiliary system consists of one additional oscillator and three physical qubits. Each of these gate implementations requires~$O(\ell^4)$ elementary operations and achieves a logical gate error of order~$O(\ell^3 \kappa)$.

\section*{Acknowledgments}

LB, BD and RK gratefully acknowledge support by the European Research Council under
grant agreement no. 101001976 (project EQUIPTNT), as well as the Munich Quantum
Valley, which is supported by the Bavarian state government through the Hightech Agenda
Bayern Plus. 

\newpage

\appendix


\section{The logical gate error and its properties~\label{app:logicalgateerrorproperties}} 

In this section we briefly recall the definition and several properties of the logical gate error defined in~\cite{cliffordslinearoptics2025}.
Let~$d\ge 2$ be an integer and let~$U:\mathbb{C}^d \rightarrow\mathbb{C}^d$ be a unitary.
Correspondingly, we define the CPTP map
\begin{alignat}{2}
\cU:\cB(\mathbb{C}^d) & \ \rightarrow\   && \cB(\mathbb{C}^d)\\
\rho & \ \mapsto\  && \cU(\rho)=U\rho U^\dagger\ .
\end{alignat}
We are interested in implementing the logical CPTP map~$\cU$ on a larger ``physical''  space by a suitable implementation when we consider~$\mathbb{C}^d$ as the codespace of some code. 

More precisely, let~$\cH_{in}, \cH_{out}$ be  Hilbert spaces, and let~$\cL_{in}\subseteq\cH_{in}$ ~$\cL_{out}\subseteq \cH_{out}$ be  (``logical'') subspaces isomorphic to~$\mathbb{C}^d$. 
We define an (isometric) encoding map~$\encmap_{\cL_{in}}:\mathbb{C}^d \rightarrow \cL_{in}$, and a corresponding decoding map~$\decmap_{\cL_{out}} :\cL_{out} \rightarrow\mathbb{C}^d$. We define the map 
\begin{align}
    \encoded{U}=\encmap_{\cL_{out}}  U \decmap_{\cL_{in}}  :\cL_{in}  \rightarrow \cL_{out}  \ .
\end{align}
The map~$\encoded{U}$ induces the CPTP map
\begin{alignat}{2}
    \label{eq:Ubold}
    \encodedC{U} : \cal{B}(\cal{L}_{in}) &\ \rightarrow\ && \cal{B}(\cal{L}_{out}) \\
    \rho &\ \mapsto\ && \encoded{U}\rho \encoded{U}^\dagger \ .
\end{alignat}
In the following we denote by~$\pi_\cL$ the orthogonal projection onto the subspace~$\cL$ of a Hilbert space~$\cH$, and we define the corresponding map~$\Pi_\cL$ as 
$\Pi_{\cL}(\rho)=\pi_\cL\rho\pi_\cL$.

The logical gate error of an implementation~$\cW: \cB(\cH_{in})  \rightarrow \cB(\cH_{out})$ of a logical unitary~$U: \mathbb{C}^d \rightarrow \mathbb{C}^d$ on input and output code spaces~$\cL_{in}$  and~$\cL_{out}$, respectively, is defined as 
\begin{align}
\gateerror_{\cL_{in}, \cL_{out}}(\cW,U)=\left\|(\cW-\encodedC{U})\circ \Pi_{\cL}\right\|_\diamond\, .
\end{align}
We frequently consider the case where~$\cH_{in} = \cH_{out} = \cH$,~$\cL_{in} = \cL_{out}$ and an implementation of the form~$\cW(\rho)=W\rho W^\dagger$ where~$W$ is a unitary on~$\cH$. In this case, we write 
\begin{align}
    \gateerror_{\cL}(W,U)=\gateerror_{\cL}(\cW,U)
\end{align} by a slight abuse of notation.

In the following we collect several results about the logical gate error.

\begin{lemma}[Subadditivity of the logical gate error,~\cite{cliffordslinearoptics2025}] \label{lem: additivity gate error} Let~$\cH_{in}^{(1)}, \cH_{out}^{(1)}, \cH_{out}^{(2)},\cK$ be Hilbert spaces. Set~$\cH_{in}^{(2)} := \cH_{out}^{(1)}$.
    Let~$U_1,U_2:\cK \rightarrow \cK$ be unitaries and let~$\cal{W}_1 : \cal{B}(\cH^{(1)}_{in}) \rightarrow \cal{B}(\cH^{(1)}_{out})$ and~$\cal{W}_2 : \cal{B}(\cH^{(2)}_{in}) \rightarrow \cal{B}(\cH^{(2)}_{out})$ be CPTP maps. Let~$\cal{L}^{(1)}_{in} \subseteq \cal{H}^{(1)}_{in}, \cal{L}^{(1)}_{out} \subseteq \cal{H}^{(1)}_{out},\cal{L}^{(2)}_{in} \subseteq \cal{H}^{(2)}_{in}$ and~$\cal{L}^{(2)}_{out} \subseteq \cal{H}^{(2)}_{out}$ be code spaces with~$\cal{L}^{(2)}_{in} := \cal{L}^{(1)}_{out}$, and~$\cLin^{(1)} \simeq \cLout^{(1)} \simeq \cLout^{(2)} \simeq \cK$. We have
    \begin{align}
        \label{eq:def_logicalerrror}
    \gateerror_{\cL^{(1)}_{in},\cL^{(2)}_{out}}(\cW_2\circ\cW_1,U_2U_1)&\leq 
    \gateerror_{\cL^{(1)}_{in}, \cL_{out}^{(1)}}(\cW_1,U_1)+\gateerror_{\cL^{(2)}_{in}, \cL_{out}^{(2)}}(\cW_2,U_2) \ .
    \end{align}
    \end{lemma}

    The logical gate error is invariant under tensoring additional systems. More precisely, we have the following.
    \begin{lemma}[See~\cite{cliffordslinearoptics2025}] \label{lem:invariancetensoredidentity}
        Let~$\cH_1,\cH_2$ be Hilbert spaces. Let~$\cL_j \subset \cH_j$ be a subspace for~$j=1,2$ and let~$\cW:\cB(\cH_1) \rightarrow \cB(\cH_1)$ be a 
        CPTP map. 
        Then 
        \begin{align}
            \gateerror_{\cL_1\otimes \cL_2}(\cW \otimes \mathsf{id}_{\cB(\cH_2)}, U \otimes I_{\mathbb{C}^d})&=\gateerror_{\cL_1}(\cW,U)\ .
        \end{align} 
    \end{lemma}

We will bound the logical gate error of a (unitary) implementation in terms of the matrix elements with respect to a choice of basis of the logical subspaces. 
More precisely, let~$\{\ket{j}\}_{j=0}^{d-1}$ be an orthonormal basis of~$\bb{C}^d$,~$\{\ket{\encoded{j}}_{\cLin}\}_{j=0}^{d-1}$  an orthonormal basis of~$\cLin$, and~$\{\ket{\encoded{j}}_{\cLout}\}_{j=0}^{d-1}$  an orthonormal basis of~$\cLout$.
Then we can  fix a concrete encoding of~$\mathbb{C}^d$ into~$\cLin\subseteq\cHin$ by mapping basis states as 
\begin{align}
\begin{matrix}
\encmap_{\cLin}: & \mathbb{C}^d & \rightarrow & \cHin\\
             & \ket{j} & \mapsto &\ket{\encoded{j}}_{\cLin}
             \end{matrix}\qquad\textrm{ for }\qquad j\in \{0,\ldots,d-1\}\ ,
\end{align}
and linearly extending to all of~$\mathbb{C}^d$. 
We denote  the inverse (``decoding'') map by~$\decmap_{\cLin}=(\encmap_\cLin)^{-1}$.
For the subspace~$\cLout \subseteq \cHout$, we define a pair of maps~$(\encmap_{\cLout},\decmap_{\cLout})$ in an analogous manner.
The following definition will be useful.
\begin{definition}[Definition 3.6 in \cite{cliffordslinearoptics2025}]\label{def:diagonalunitary}
We define the operator 
\begin{align}
    B=B^U_{\cLin,\cLout}(W,U) = 
    \encoded{U}^\dagger \pi_{\cLout}W\pi_{\cLin}:\cL_{in}\rightarrow\cL_{in}\, .
\end{align} 
It is characterized by the matrix elements with respect to the basis~$\{\ket{\encoded{k}}_{\cL_{in}}\}_{k=0}^{d-1}$ given by 
\begin{align}
 B_{j,k}=\langle\encoded{j}|_{\cLin}
\encoded{U}^\dagger \pi_{\cLout}W\pi_{\cLin} \ket{\encoded{k}}_{\cLin}
\end{align}
for~$j,k\in \{0,\ldots,d-1\}$. 
\end{definition}
\noindent 
It is straightforward to check that 
\begin{align}
B_{j,k}&:= \sum_{m=0}^{d-1} \overline{U_{m,j}}  \langle \encoded m|_{\cLout} W |\encoded{k}\rangle_{\cL_{in}} \qquad\textrm{for}\qquad j,k\in \{0,\ldots,d-1\}\ .\label{eq:bjkdefinition}
\end{align}
We call a matrix~$B\in\mathsf{Mat}_{d\times d}(\mathbb{C})$~$s$-sparse if the number of non-zero entries in each row and column is at most~$s$. We need the following bound on the logical gate error.
\begin{corollary}[Gate error of unitary implementations in terms of matrix elements,~\cite{cliffordslinearoptics2025}]
\label{cor:shortmatrixBstatement}
    Let~$W:\cHin\rightarrow\cHout$ be unitary and~$\cW(\rho)=W\rho W^\dagger$. Let~$U:\mathbb{C}^d\rightarrow\mathbb{C}^d$ be unitary, and let ~$B=B^{U}_{\cLin,\cLout}(W,U)$ be the operator introduced in Definition~\ref{def:diagonalunitary}.
    Suppose that~$B$ is~$s$-sparse and has real, non-zero entries~$\{B_{j,j}\}_{j=0}^{d-1}$ on the diagonal. 
    Then
    \begin{align}
        \max\left\{\gateerror_{\cLin,\cLout} (W,U), \gateerror_{\cLout,\cLin} (W^\dagger,U^\dagger) \right\} &\leq 8 \left((1-\min_j |B_{j,j}|)+(s-1)\max_{j,\ell:j\neq \ell }|B_{j,\ell}|\right)^{1/2}\ .
    \end{align}
    \end{corollary}

\section{Matrix elements of implementations of \\ basic bit-manipulation maps \label{app:matrixelements}}

In this section we bound the matrix elements of physical implementations introduced in  Section~\ref{sec:bitmanipulationimplement} of the~$\cX_\ell$-gate, the~$\LSB_\ell$-gate and the embedding isometry. 

Let~$d\ge 2$ be an integer,~$\kappa,\Delta>0$ and~$\varepsilon\le 1/2$. 
We introduce the integer centered (normalized) GKP states~$\ket{\gkp_{\kappa,\Delta}^\varepsilon}$ in position-space as
\begin{align}
    \gkp^\varepsilon_{\kappa,\Delta}(x) \propto \sum_{s \in \mathbb{Z}} e^{-\kappa^2 s^2/2} \cdot e^{-x^2/(2 \Delta^2)} 
\end{align} if~$\mathsf{dist}(x, \mathbb{Z}) \le \varepsilon$ and~$\gkp^\varepsilon_{\kappa,\Delta}(x)=0$ otherwise.
We note that we can relate the GKP code states~$\ket{\gkp_{\kappa,\Delta}(j)_{2^\ell}}$ (cf. Eq.~\eqref{eq: def ideal gkp}) to the state~$\ket{\gkp_{\kappa,\Delta}^\varepsilon}$ as
\begin{align}
    \label{eq:defGKPlogicalstate}
    \ket{\gkp_{\kappa,\Delta}(j)_{2^\ell}} = e^{-i \sqrt{2\pi/{2^\ell}}jP} M_{\sqrt{2^{\ell+1}\pi}} \ket{\gkp^\varepsilon_{\kappa,\Delta}}\qquad \textrm{for} \qquad j \in \{0,\dots,{2^\ell}-1\}\, ,
\end{align} 
where we use the short-hand notation~$M_\alpha = e^{-i \log\alpha (PQ + QP)/2}$ for~$\alpha>0$. We give matrix elements of implementations with respect to the basis consisting of the states
\begin{align}
    \ket{\gkp_{\kappa,\Delta}(j)_{2^\ell}} \otimes \ket{b} \in L^2(\bbR) \otimes \bbC^2 \qquad \text{ for }\qquad j\in\mathbb{Z}_{2^\ell}, b \in \{0,1\} \ .
\end{align}

\subsection{The~$\qCX{\ell}$-gate}

We bound the matrix elements of the implementation of the~$\qCX{\ell}$-gate given in Table~\ref{fig:Gellimplement}.

\begin{lemma}[Matrix elements of the~$\qCX{\ell}$-gate]\label{lem: matrix elements gen CZm}
    Let~$\ell \in \mathbb{N}$. Let~$\kappa,\Delta \in (0,1/4)$ and~$\varepsilon \in (0,2^{-(\ell+1)}]$. Define the unitary~$W_{\qCX{\ell}}$ acting on~$  L^2(\mathbb{R}) \otimes \mathbb{C}^2$ as
    \begin{align}
        W_{\qCX{\ell}} = I \otimes \proj{0}  +  e^{-i \sqrt{2\pi/2^\ell}P} \otimes \proj{1}  \, .
    \end{align}
    For~$j,k \in \mathbb{Z}_{2^\ell}$ and~$b,b' \in \{0,1\}$, define the matrix element 
    \begin{align}
        M_{(j,b),(k,b')} = \left( \langle \gkp_{\kappa,\Delta}^\varepsilon(j)_{2^\ell}| \otimes \langle b|\right) W_{\qCX{\ell}}\left( |\gkp_{\kappa,\Delta}^\varepsilon(k)_{2^\ell}\rangle  \otimes \ket{b'} \right)
    \end{align}
    of the operator~$W_{\qCX{\ell}}$ with respect to the basis
    \begin{align}
        \left\{ \ket{\gkp_{\kappa,\Delta}^\varepsilon(j)_{2^\ell}} \otimes \ket{b}  \right\}_{(j,b) \in  \mathbb{Z}_{2^\ell} \times \{0,1\}} \qquad \textrm{of} \qquad  \gkpcode{\kappa,\Delta}{\varepsilon}{2^\ell} \otimes \mathbb{C}^2 \, . 
    \end{align} 
    Then we have
    \begin{align}
          (1- \kappa^2)\cdot \delta_{b, b'} \cdot \delta_{j,k\oplus b}  \le &M_{(j,b),(k,b')} \le   \delta_{b, b'} \cdot \delta_{j,k\oplus b} 
        \qquad 
        \textrm{ for all }   j,k \in \mathbb{Z}_d  \text{ and }  b,b' \in \{0,1\}\, .
    \end{align}
    In particular, the matrix~$M = (M_{(j,b),(k,b')})$ is~$1$-sparse.
    \end{lemma}
    \begin{proof}
    A straightforward computation shows that 
    \begin{align}
       M_{(j,b),(k,b')} &= \delta_{b,b'}\left((1-b)\delta_{j,k} +  b\cdot  \langle \gkp_{\kappa,\Delta}^\varepsilon(j)_{2^\ell}, e^{-i\sqrt{2\pi/2^\ell}P} \gkp_{\kappa,\Delta}^\varepsilon(k)_{2^\ell} \rangle \right)\, .
    \end{align}
    The claim follows from~\cite[Lemma C.1]{cliffordslinearoptics2025} which states 
    \begin{align}
        (1 - \kappa^2)\cdot\delta_{j,k \oplus 1} \le \langle \gkp_{\kappa,\Delta}^\varepsilon(j)_{2^\ell}, e^{-i\sqrt{2\pi/2^\ell}P} \gkp_{\kappa,\Delta}^\varepsilon(k)_{2^\ell} \rangle & \le \delta_{j,k \oplus 1}\, . 
    \end{align}
    \end{proof}

\subsection{The~$\lsb{\ell}$-gate}
    
We bound the matrix elements of the implementation of the~$\lsb{\ell}$-gate given in Table~\ref{fig:Gellimplement}.

\begin{lemma}[Matrix elements of the LSB-gate] 
\label{lem:matrix elements lsb}
    Let~$\ell \in \mathbb{N}$. Let~$\kappa,\Delta>0$ and~$\varepsilon \in (0,2^{-(\ell+1)}]$.
Define the unitary~$W_{\lsb{\ell}}$ acting on~$ L^2(\mathbb{R}) \otimes \mathbb{C}^2$ as 
\begin{align}
    W_{\lsb{\ell}} &=  \left(  I \otimes H\right)\left(I \otimes \proj{0}     +    e^{i\pi \sqrt{2^\ell/(2\pi)} Q } \otimes \proj{1}  \right) \left( I \otimes H \right)\, .
    \label{eq:WLSB}
\end{align}
For~$j,k \in \mathbb{Z}_{2^\ell}$ and~$b,b' \in \{0,1\}$, define the matrix element 
\begin{align}
    M_{(j,b),(k,b')} = \left(   \langle \gkp_{\kappa,\Delta}^\varepsilon(j)_{2^\ell} \otimes \langle b| \right) W_{\lsb{\ell}}\left(  |\gkp_{\kappa,\Delta}^\varepsilon(k)_{2^\ell}\rangle \otimes \ket{b'} \right)
\end{align}
of the operator~$W_{\lsb{\ell}}$  with respect to the basis
\begin{align}
    \left\{  \gkp_{\kappa,\Delta}^\varepsilon(j) \otimes \ket{b}  \right\}_{(j,b) \in   \mathbb{Z}_{2^\ell} \times \{0,1\} }  \qquad \textrm{of} \qquad \gkpcode{\kappa,\Delta}{\varepsilon}{2^\ell} \otimes \mathbb{C}^2\, .
\end{align}
Then 
\begin{align}
    \label{eq:claim-implement-CX}
    M_{(j,b),(j,b\oplus j_0)} & \ge  1 - 2 \cdot 2^{2\ell} \Delta^2 - 8 (\Delta/\varepsilon)^4\\
    |M_{(j,b),(j,b\oplus j_0\oplus 1)}| &\le  2 \cdot 2^{2\ell}\Delta^2 + 8 (\Delta/\varepsilon)^4\\
    M_{(j,b),(k,b')} & = 0 \qquad \textrm{otherwise} 
\end{align}
for all~$j,k \in \mathbb{Z}_{2^\ell}$ and~$b,b' \in \{0,1\}$, where~$j_0 = j \mod 2$ is the least significant bit in the binary representation of~$j$. In particular, the matrix~$M = (M_{(j,b),(k,b')})$ is~$2$-sparse.
\end{lemma}
\begin{proof}
By Eq.~\eqref{eq:defGKPlogicalstate} we have 
\begin{align}
    \left(  \langle \gkp_{\Delta,\kappa}^\varepsilon(j)_{2^\ell} | \otimes \langle b|  \right) W_{\lsb{\ell}}\left(    |\gkp_{\kappa,\Delta}(k)^\varepsilon_{2^\ell}\rangle \otimes \ket{b'} \right) &=\left(    \langle \gkp_{\Delta,\kappa}^\varepsilon|\otimes \langle b| \right) U_{j,k}^{(\ell)}\left(   |\gkp_{\kappa,\Delta}^\varepsilon\rangle \otimes \ket{b'}  \right)
\end{align}
where
\begin{align}
    U_{j,k}^{(\ell)} 
    = &\left( \left(M_{\sqrt{2\pi \cdot 2^\ell}}\right)^\dagger \left(e^{-i\sqrt{2\pi/2^{\ell}}jP}\right)^\dagger \otimes I \right) W_{\lsb{\ell}} \left(  e^{-i\sqrt{2\pi/2^\ell}kP} M_{\sqrt{2\pi \cdot 2^\ell}} \otimes I  \right) \ .
    \label{eq:aux4}
\end{align}
Plugging Eq.~\eqref{eq:WLSB} into Eq.~\eqref{eq:aux4} and using the identities 
\begin{align}
\begin{aligned}
    \left(M_\alpha\right)^\dagger P M_\alpha
    &= P/\alpha\\
    \left(M_\alpha\right)^\dagger Q M_\alpha 
    &= \alpha Q\\
    \left(e^{-i\beta P}\right)^\dagger Q e^{-i\beta P}
    &= Q + \beta I
\end{aligned}
    \quad\textrm{ for }\qquad \alpha>0, \beta \in \mathbb{R}
\end{align} 
to derive the transformation of the quadrature operators, we obtain
\begin{align}
    U_{j,k}^{(\ell)} &= \left( I \otimes H\right) \left(  I \otimes \proj{0}   +  e^{i\pi \sqrt{2^\ell/(2\pi)}\left(\sqrt{2\pi d} Q + \sqrt{2\pi/2^\ell}j\right)} \otimes \proj{1}  \right) \left(I \otimes H \right) \left( e^{-i(k-j)P/2^\ell} \otimes I \right)\\
    &=  \left(  I \otimes H \right) \left( I \otimes \proj{0} +    e^{i\pi j} e^{i\pi \cdot 2^\ell Q} \otimes \proj{1} \right) \left( I \otimes H \right)\left(   e^{-i(k-j)P/2^\ell} \otimes I\right)\, ,
\end{align} 
Using~$H \ket{b} = (\ket{0} + (-1)^b \ket{1})/\sqrt{2}$ gives
\begin{align}
    M_{(j,b), (k,b')} &= \frac{1}{2} \left(\delta_{j,k} + (-1)^b (-1)^{b'} e^{i\pi j} \langle \gkp_{\kappa,\Delta}^\varepsilon,  e^{i\pi \cdot 2^\ell Q} e^{-i(k-j)P/2^\ell} \gkp_{\kappa,\Delta}^\varepsilon \rangle \right)\\
        &=  \frac{1}{2}  \delta_{j,k} \cdot \left(1 +  (-1)^{b+ b' + j_0} \langle \gkp_{\kappa,\Delta}^\varepsilon,  e^{i\pi \cdot 2^\ell Q}  \gkp_{\kappa,\Delta}^\varepsilon \rangle \right) \label{eq:LSBproofaux1}\, ,
\end{align}
where we used~\cite[Lemma~7.6]{brenner2024factoring}, which implies that 
\begin{align}
    \langle \gkp_{\kappa,\Delta}^\varepsilon,   e^{-i(k-j)P/2^\ell} \gkp_{\kappa,\Delta}^\varepsilon \rangle \neq 0 
    \qquad\text{ if and only if }\qquad 
    j = k \qquad \text{ for } \varepsilon\le 2^{-(\ell+1)} \ ,
\end{align}
together with the fact that the unitary~$e^{i\pi d Q}$ does not change the support of the state~$\gkp_{\kappa,\Delta}^{\varepsilon}$.
Moreover, we used that
\begin{align}
    (-1)^b (-1)^{b'} e^{i\pi j} = e^{i\pi (b+ b' + j)} = (-1)^{b+ b' + j_0}
\end{align}
where~$j_0 = j \mod 2$ is the least significant bit of~$j$ in its binary representation, i.e., $j=0$ is~$j$ is even and~$j\neq 1$ if~$j$ is odd.
By~\cite[Lemma B.14]{cliffordslinearoptics2025} with~$z = 2^{\ell-1}$ we have 
\begin{align}
    \label{eq:LSBproofaux2} \langle \gkp_{\kappa,\Delta}^\varepsilon,  e^{i\pi \cdot 2^\ell Q}  \gkp_{\kappa,\Delta}^\varepsilon \rangle 
    &= 1 - 10\cdot 2^{2\ell -2} \Delta^2 - 16(\Delta/\varepsilon)^4\, .
\end{align}
Inserting Eq.~\eqref{eq:LSBproofaux2} into Eq.~\eqref{eq:LSBproofaux1} gives the claim~\eqref{eq:claim-implement-CX}.

Finally, Eq.~\eqref{eq:LSBproofaux1} implies that for fixed~$j \in \mathbb{Z}_{2^\ell}$ only the matrix elements~$M_{(j,0),(j,0)}$, $M_{(j,0),(j,1)}$, $M_{(j,1),(j,0)}$ and~$M_{(j,1),(j,1)}$ are non-zero, which implies that the matrix~$M$ is~$2$-sparse. 
\end{proof}

\subsection{The embedding isometry}

In the following lemma, we establish that the squeezing operation~$M_{\sqrt{2}}$ implements the embedding map~$\embedmap_\ell$ defined in Eq.~\eqref{eq:embeddingmapdefinitionisometry} perfectly. 

\begin{lemma} \label{lem:matrixelementembedmap}
    Let~$\ell \in \mathbb{N}$. Let~$\kappa,\Delta>0$ and~$\varepsilon \in (0,2^{-(\ell+1)}]$. Then we have 
    \begin{align}
        M_{\sqrt{2}} \ket{\gkp_{\kappa,\Delta}^\varepsilon(j)_{2^\ell}} = \ket{\gkp_{\kappa,\Delta}^\varepsilon(2j)_{2^{\ell+1}}}\qquad \textrm{for all} \qquad j \in \mathbb{Z}_d\, .
    \end{align}
\end{lemma}

\begin{proof}
    By definition, we have 
    \begin{align}
        M_{\sqrt{2}} \ket{\gkp_{\kappa,\Delta}^\varepsilon(j)_{2^{\ell}}} &= M_{\sqrt{2}} e^{-i \sqrt{2\pi/2^\ell}j P} M_{\sqrt{2\pi \cdot 2^\ell}} \ket{\gkp_{\kappa,\Delta}^\varepsilon}\\
        &= e^{-i \sqrt{4\pi /2^\ell}j P} M_{\sqrt{2}} M_{\sqrt{2\pi \cdot2^\ell}} \ket{\gkp_{\kappa,\Delta}^\varepsilon}\\
        &= e^{-i \sqrt{2\pi/(2^{\ell+1})}2j P} M_{\sqrt{2\pi\cdot 2^{\ell+1}}} \ket{\gkp_{\kappa,\Delta}^\varepsilon} \\
        &= \ket{\gkp_{\kappa,\Delta}^\varepsilon(2j)_{2^{\ell+1}}} \, ,
    \end{align}
    where we used that~$M_\alpha P M_\alpha^\dagger = \alpha P$ for all~$\alpha>0$. 
\end{proof}

\section{Logical gate errors of implementations of \\ basic bit-manipulation maps \label{sec: lem: universal bound special-proof}}

Here we proof Lemma~\ref{lem: universal bound special} which bounds the logical gate error of the implementations given in Tab.~\ref{fig:Gellimplement} of the basic bit-manipulation maps. The proof uses Corollary~\ref{cor:shortmatrixBstatement} and relies on the bounds for the matrix elements of the corresponding implementations derived in Appendix~\ref{app:matrixelements}.

\begin{proof}[Proof Lemma~\ref{lem: universal bound special}]
The claim about the number of elementary operations which compose the respective physical implementation~$W_U$ of~$U \in \{\qCX{\ell},\lsb{\ell}, \embedmap_\ell\}$ follows directly from Lemma~\ref{lem: gkpLSB}.
Next, we derive the bounds on the logical gate error. 

Define
\begin{align}
\ket{\encoded{k,b}}=  \ket{\gkp_{\kappa,\Delta}^\varepsilon(k)_{2^\ell}} \otimes \ket{b}   \qquad\textrm{ for }\qquad   k\in \mathbb{Z}_{2^\ell}, \,b\in \{0,1\} \ .
\end{align}
We first focus on the logical~$\qCX{\ell}$-gate defined by its matrix elements~$(\qCX{\ell})_{(j,b),(k,c)}=\delta_{b,c} \cdot \delta_{j,k\oplus b}$ and its (approximate) implementation~$W_{\qCX{\ell}}$.
The matrix elements of the corresponding matrix~$B$ in Definition~\ref{def:diagonalunitary} are 
\begin{align}
    B_{(r,a),(s,b)}&=\sum_{c=0}^{1}\sum_{j=0}^{d-1} \overline{(\qCX{\ell})_{(j,c),(r,a)}}\langle \encoded{j,c}|W_{\qCX{\ell}} |\encoded{s,b}\rangle\\
&= \langle \encoded{ r \oplus a, a}|W_{\qCX{\ell}} |\encoded{s,b}\rangle \qquad\textrm{ for }\qquad  r,s\in \mathbb{Z}_{2^\ell} \text{ and } a,b \in \{0,1\}\ .
\end{align}
In particular, the diagonal elements are
\begin{align}
    B_{(r,a),(r,a)}&= \langle\encoded{ r\oplus  a, a}|W_{\qCX{\ell}} |\encoded{r,a}\rangle\,  .
\end{align}
According to Lemma~\ref{lem: matrix elements gen CZm} the matrix~$B$ is~$1$-sparse ($s=1$) and
\begin{align}
    B_{(r,a),(r,a)} \ge 1- \kappa^2 
\end{align}
for every~$r \in \mathbb{Z}_{2^\ell}$ and~$a \in \{0,1\}$.
Using Corollary~\ref{cor:shortmatrixBstatement} with $s=1$ the logical gate error is bounded by 
\begin{align}
\gateerror_{  \gkpcode{\kappa,\Delta}{\varepsilon}{2^\ell} \otimes \mathbb{C}^{2}}(W_{\qCX{\ell}},\qCX{\ell})
&\le 8 \left(1-\min_{(r,a)} |B_{(r,a),(r,a)}|\right)^{1/2} \\
&\leq 8 \kappa\,. \label{eq:gateerrorqCX}
\end{align}

Similarly, we bound the logical gate error of the~$\lsb{\ell}$-gate. 
Recall that the logical~$\lsb{\ell}$-gate is defined by its matrix elements~$(\lsb{\ell})_{(j,b),(k,c)}=\delta_{b, c \oplus j_0}\cdot \delta_{j,k}$ where~$j_0 = j \mod 2$.
The matrix elements of the corresponding matrix~$B$ in Definition~\ref{def:diagonalunitary} are given by
\begin{align}
    B_{(r,a),(s,b)}&=\sum_{c=0}^{1} \sum_{j=0}^{d-1} \overline{(\lsb{\ell})_{(j,c),(r,a)}}\langle \encoded{j,c}|W_{\lsb{\ell}} |\encoded{s,b}\rangle\\
&= \langle \encoded{r, a\oplus r_0}|W_{\lsb{\ell}} |\encoded{s,b}\rangle \qquad\textrm{ for }\qquad  \ r,s\in \mathbb{Z}_{2^\ell} \text{ and }  a,b \in \{0,1\}\ .
\end{align}
Therefore
\begin{align}
    B_{(r,a),(r,a)}&= \langle\encoded{r,a\oplus r_0}|W_{\lsb{\ell}} |\encoded{r,a}\rangle\,  ,
\end{align}
where~$r_0 = r \mod 2$.
According to Lemma~\ref{lem:matrix elements lsb} the matrix~$B$ is~$2$-sparse, and the only non-vanishing matrix elements are bounded by
\begin{align}
    \begin{aligned}
    B_{(r,a),(r,a)} &\ge 1- 2 \cdot2^{2\ell} \Delta^2 - 8 (\Delta/\varepsilon)^4 \\
    |B_{(r,a),(r,a \oplus 1)}| &\le 2 \cdot2^{2\ell} \Delta^2 + 8 (\Delta/\varepsilon)^4    
    \end{aligned}
    \qquad \text{for every}\qquad\text{~$r \in \mathbb{Z}_{2^\ell}, a \in \{0,1\}$ .}
\end{align}
Using Corollary~\ref{cor:shortmatrixBstatement} (with~$s=2$), the logical gate error is bounded by
\begin{align}
\gateerror_{  \gkpcode{\kappa,\Delta}{\varepsilon}{2^\ell}\otimes \mathbb{C}^{2}}( W_{\lsb{\ell}},\lsb{\ell})
&\leq \left((1-\min_{(r,a)} |B_{(r,a),(r,a)}|)+(s-1)\max_{r,a,s,b:(r,a)\neq (s,b) }|B_{(r,a),(s,b)}|\right)^{1/2} \\
&\leq 8 \left( 4 \cdot2^{2\ell} \Delta^2 + 16 (\Delta/\varepsilon)^4\right)^{1/2} \\
&\le 16 \cdot 2^{\ell} \Delta + 32 (\Delta/\varepsilon)^2 \, , \label{eq:gateerrorlsb}
\end{align} where we used that~$\sqrt{x + y} \le \sqrt{x} + \sqrt{y}$ for~$x, y \ge 0$.

The claim about the implementation of the isometry~$\embedmap_\ell$ follows from Lemma~\ref{lem:matrixelementembedmap}.

Finally, 
the fact that the bounds on the logical gate error of the implementations were derived using Corollary~\ref{cor:shortmatrixBstatement} implies that the same bounds hold for the implementations of the adjoint operations.
\end{proof}

\section{Logical gate error for rectangular-envelope GKP codes \label{sec:approximateGKPcodescomb-app}}
In this section, we argue that  similar conclusions apply to (truncated) rectangular-envelope GKP states (or comb states) and the corresponding approximate codes.
In Section~\ref{sec:approximateGKPcodescomb} we define the corresponding states and codes.  In Section~\ref{sec:matrixelementbox} we compute matrix elements of implementations of basic bit-manipulations. In Section~\ref{sec:logicalgatescombstates} we then establish bounds on the corresponding logical gate errors. 

\subsection{Definition of rectangular-envelope GKP codes\label{sec:approximateGKPcodescomb}}
Here we define so-called comb-states. Instead of a Gaussian envelope parametrized by a squeezing parameter~$\kappa>0$, these are defined in terms of a rectangular envelope whose width is determined by an integer~$L$.  We still use a parameter~$\Delta>0$ for the width of individual peaks (local maxima), which are (truncated) Gaussians.
Recall that we use the Gaussian
\begin{align}
\Psi_{\Delta}(x)=\frac{1}{\left(\pi \Delta^2\right)^{1 / 4}} e^{-x^2 /\left(2 \Delta^2\right)}
\end{align}
and its truncated variant
\begin{align}
\Psi_{\Delta}^{\varepsilon}=\frac{\Pi_{[-\varepsilon, \varepsilon]} \Psi_{\Delta}}{\left\|\Pi_{[-\varepsilon, \varepsilon]} \Psi_{\Delta}\right\|}
\end{align}
where~$\|\cdot\|$ is the Euclidean norm, and where~$\varepsilon \in (0,1/2)$. 
Defining the translated versions ~$\chi_{\Delta}^\varepsilon(z)(\cdot)=\Psi_\Delta^\varepsilon(\cdot-z)$ , the integer-spaced comb state (or ``rectangular-envelope GKP state'') is defined as 
\begin{align}
\left|\Sha_{L, \Delta}^{\varepsilon}\right\rangle=\frac{1}{\sqrt{L}} \sum_{z=-L / 2}^{L / 2-1}\left|\chi_{\Delta}^{\varepsilon}(z)\right\rangle\ .\label{eq:integerspacedcomb}
\end{align}
Support properties are especially important for our considerations. 
The following is an immediate consequence of the definitions: We have 
\begin{align}
\mathsf{supp}\left(e^{-i\delta P}\Sha_{L, \Delta}^{\varepsilon}\right)\cap
\mathsf{supp}\left(\Sha_{L, \Delta}^{\varepsilon}\right)=\emptyset\qquad\textrm{  if }\qquad \varepsilon< |\delta|<1-2\varepsilon\ .
\end{align}
We can also show that the state~$\Sha_{L, \Delta}^{\varepsilon}$ is fixed by 
momentum-translations which are small integer multiples of~$2\pi$.
\begin{lemma}[Momentum-translated rectangular-envelope GKP states] \label{lem:shiftinvariancemomemntumsha}
Let~$z\in\mathbb{Z}$. Let~$L\in 2\mathbb{N}$,~$\Delta>0$ and~$\varepsilon\leq 1/(2d)$.
Then
\begin{align}
\langle \Sha^\varepsilon_{L,\Delta},e^{2\pi i zQ}\Sha^\varepsilon_{L,\Delta}\rangle &
\geq 1- 10z^2\Delta^2-16(\Delta/\varepsilon)^4\ .
\end{align}

\end{lemma}
\begin{proof}
We have
\begin{align}
\langle \Sha^\varepsilon_{L,\Delta},e^{2\pi i zQ}\Sha^\varepsilon_{L,\Delta}\rangle &
=\frac{1}{L}\sum_{y=-L/2}^{L/2-1}
\langle \chi^\varepsilon_\Delta(y),e^{2\pi i z Q}\chi^\varepsilon_\Delta(y)\rangle
\end{align}
since the operator~$e^{2\pi i z Q}$ does not change the support,
and the functions~$\{\chi^\varepsilon_\Delta(y)\}_{y\in\mathbb{Z}}$ have pairwise disjoint support by definition. A straightforward calculation (see the proof of Lemma~B.14 in~\cite{cliffordslinearoptics2025}) gives 
\begin{align}
\langle \chi^\varepsilon_\Delta(y),e^{2\pi i z Q}\chi^\varepsilon_\Delta(y)\rangle
&\geq 1-10z^2\Delta^2-16(\Delta/\varepsilon)^4\qquad\textrm{ for any }\qquad y\in\mathbb{Z}\ .
\end{align}
The claim follows from this. 
\end{proof}

As for GKP codes with a Gaussian envelope, we use the state defined by Eq.~\eqref{eq:integerspacedcomb} as a starting point to define approximate (rectangular-envelope) GKP codes:
 For any integer~$d\geq 2$,we  define the~$d$ approximate GKP states
\begin{align}
\ket{\Sha_{L,\Delta}^\varepsilon(j)_d}&= e^{-i\sqrt{2\pi/d}jP} M_{\sqrt{2\pi d}}\ket{\Sha_{L,\Delta}^\varepsilon}\qquad\textrm{ for }\qquad j\in \{0,\ldots,d-1\}\ . \label{eq:symmetricallysqueezedstates}
\end{align}In this definition, we have used the single-mode squeezing operator~$M_\alpha$ for~$\alpha>0$, a Gaussian unitary which acts on~$\Psi\in L^2(\mathbb{R})$ as
$\left(M_\alpha \Psi\right)(x)=\frac{1}{\sqrt{\alpha}} \Psi(x / \alpha)$ for~$x \in \mathbb{R}$. A straightforward computation gives
\begin{align}
M_{\alpha}\ket{\chi^\varepsilon_\Delta(z)}&=\chi_{\Delta\alpha}^{\varepsilon \alpha}(z\alpha)\ 
\end{align}
and thus 
\begin{align}
e^{-i\sqrt{2\pi/d}jP} M_{\sqrt{2\pi d}}\ket{\chi^\varepsilon_\Delta(z)}&=\chi^{\varepsilon\sqrt{2\pi d}}_{\Delta\sqrt{2\pi d}}\left(z\sqrt{2\pi d}+\sqrt{\frac{2\pi}{d}} j\right)\ .
\end{align}
It follows that
\begin{align}
\ket{\Sha_{L,\Delta}^\varepsilon(j)_d}&=\frac{1}{\sqrt{L}}\sum_{z=-L/2}^{L/2-1}
\chi^{\varepsilon\sqrt{2\pi d}}_{\Delta\sqrt{2\pi d}}\left(z\sqrt{2\pi d}+\sqrt{\frac{2\pi}{d}} j\right)\label{eq:shaldeltaexplicitz}
\end{align}
for any~$j\in \{0,\ldots,d-1\}$. 
In particular, the states~$\{\Sha_{L,\Delta}^\varepsilon(j)_d\}_{j=\{0,\ldots,d-1\}}$ have essentially pairwise disjoint support for~$\varepsilon\le 1/(2d)=:\varepsilon_d$ and thus they are orthogonal.

For an integer~$d\geq 2$,~$L\in 2\mathbb{N}$,~$\Delta>0$ and~$\varepsilon\leq 1/(2d)$, we define
the (rectangular-envelope truncated) GKP code with parameters~$L,\Delta,\varepsilon$ by
\begin{align}
\gkpcoderect{L,\Delta}{\varepsilon}{d}:=\mathsf{span}\{\Sha_{L,\Delta}^\varepsilon(j)_d\}_{j=\{0,\ldots,d-1\}}\ .
\end{align}
Corresponding isometric encoding- and decoding maps are defined as before using the orthonormal basis~$\{\Sha_{L,\Delta}^\varepsilon(j)_d\}_{j=\{0,\ldots,d-1\}}$. 

As before, we write
\begin{align}
\gkpcoderect{L,\Delta}{\star}{d}:=\gkpcoderect{L,\Delta}{\varepsilon_d}{d}\ 
\end{align}
for the GKP code with optimal truncation parameter~$\varepsilon_d=1/(2d)$.
We typically choose the integer~$L\in\mathbb{N}$ as a certain function of~$L\in\mathbb{N}$, i.e., we set
\begin{align}
    \label{eq:defbigL}
L_{\Delta,d}&=2^{2(\lceil \log_2 1/\Delta\rceil - \lfloor \log_2 d\rfloor)} \, ,
\end{align} where we assume that~$\Delta < 1/d$.
With these choices, we end up with a one-parameter family of approximate GKP codes depending only on the parameter~$\Delta>0$. We write
\begin{align}
\gkpcoderect{\Delta}{\star}{d}:=\gkpcoderect{L_{\Delta,d},\Delta}{\varepsilon_d}{d}\  
\end{align}
for the code associated with~$\Delta>0$, and call this the approximate (rectangular-envelope truncated) GKP code with parameter~$\Delta$.

For later use, we observe that (for sufficiently large~$L$), the operator~$e^{-i\sqrt{2\pi/d} P}$ approximately acts as a logical-Pauli operator, as expressed by the following matrix elements.
\begin{lemma} \label{lem:shiftinvariance Sha}
Let~$d\geq 2$ be an integer,~$\varepsilon \in (0,1/(2d))$,~$\Delta>0$ and~$L\in2\mathbb{N}$. 
For any~$j,k\in \{0,\ldots,d-1\}$ we then have 
\begin{align}
\langle \Sha_{L,\Delta}^\varepsilon(j)_d,e^{-i\sqrt{2\pi/d} P}\Sha_{L,\Delta}^\varepsilon(k)_d\rangle &=
\begin{cases}
1-2/L\qquad &\textrm{ if }\qquad j=k\oplus 1\\
0&\textrm{ otherwise }\ .
\end{cases}
\end{align}
where~$\oplus$ denotes addition modulo~$d$.
\end{lemma}
\begin{proof}
This follows immediately from Eq.~\eqref{eq:shaldeltaexplicitz} 
using the fact that the states
$\chi^{\varepsilon\sqrt{2\pi d}}_{\Delta\sqrt{2\pi d}}\left(z\sqrt{2\pi d}+\sqrt{\frac{2\pi}{d}} j\right)$
and 
$\chi^{\varepsilon\sqrt{2\pi d}}_{\Delta\sqrt{2\pi d}}\left(z'\sqrt{2\pi d}+\sqrt{\frac{2\pi}{d}} j'\right)$
have disjoint support unless~$(z,j)=(z',j')$.
\end{proof}

\subsection{Matrix elements of hybrid qubit-oscillator implementations \label{sec:matrixelementbox}}
\begin{lemma}[Matrix elements of the~$\qCX{\ell}$-gate]\label{lem:matrixelementCXsha}
    Let~$\ell \in \mathbb{N}$,~$\Delta >0$,~$L \in 2\mathbb{N}$ and~$\varepsilon \in (0,2^{-(\ell+1)}]$. Define the unitary~$W_{\qCX{\ell}}$ acting on~$  L^2(\mathbb{R}) \otimes \mathbb{C}^2$ as
    \begin{align}
        W_{\qCX{\ell}} = I \otimes \proj{0}  +  e^{-i \sqrt{2\pi/2^\ell}P} \otimes \proj{1}  \, .
    \end{align}
    For~$j,k \in \mathbb{Z}_{2^\ell}$ and~$b,b' \in \{0,1\}$, define the matrix element 
    \begin{align}
        M_{(j,b),(k,b')} = \left( \langle \Sha_{L,\Delta}^{\varepsilon}(j)_{2^\ell}| \otimes \langle b|\right) W_{\qCX{\ell}}\left( |\Sha_{L,\Delta}^{\varepsilon}(k)_{2^\ell}\rangle  \otimes \ket{b'} \right)
    \end{align}
    of the operator~$W_{\qCX{\ell}}$ with respect to the basis
    \begin{align}
        \left\{ |\Sha_{L,\Delta}^\varepsilon(j)_{2^\ell}\rangle \otimes \ket{b}  \right\}_{(j,b) \in  \mathbb{Z}_{2^\ell} \times \{0,1\}} \qquad \textrm{of} \qquad  \Sha\gkpcode{L,\Delta}{\varepsilon}{2^\ell} \otimes \mathbb{C}^2 \, . 
    \end{align} 
    Then we have
    \begin{align}
          (1- 2/L)\cdot \delta_{b, b'} \cdot \delta_{j,k\oplus b}  \le &M_{(j,b),(k,b')} \le   \delta_{b, b'} \cdot \delta_{j,k\oplus b} 
        \qquad 
        \textrm{ for all }   j,k \in \mathbb{Z}_d  \text{ and }  b,b' \in \{0,1\}\, .
    \end{align}
    In particular, the matrix~$M = (M_{(j,b),(k,b')})$ is~$1$-sparse.
    \end{lemma}
    \begin{proof}
        This follows by the same arguments as Lemma~\ref{lem: matrix elements gen CZm} together with Lemma~\ref{lem:shiftinvariance Sha}.
    \end{proof}

\begin{lemma}[Matrix elements of the LSB-gate] \label{lem:matrixelementlsbSha}
    Let~$\ell \in \mathbb{N}$,~$\Delta>0$,~$L\in 2\mathbb{N}$ and~$\varepsilon \in (0,2^{-(\ell+1)}]$.
    Define the unitary~$W_{\lsb{\ell}}$ acting on~$ L^2(\mathbb{R}) \otimes \mathbb{C}^2$ as 
    \begin{align}
        W_{\lsb{\ell}} &=  \left(  I \otimes H\right)\left(I \otimes \proj{0}     +    e^{i\pi \sqrt{2^\ell/(2\pi)} Q } \otimes \proj{1}  \right) \left( I \otimes H \right)\, .
    \end{align}
    For~$j,k \in \mathbb{Z}_{2^\ell}$ and~$b,b' \in \{0,1\}$ , define the matrix element 
    \begin{align}
        M_{(j,b),(k,b')} = \left(   \langle \Sha_{L,\Delta}^\varepsilon(j)_{2^\ell} \otimes \langle b| \right) W_{\lsb{\ell}}\left(  |\Sha_{L,\Delta}^\varepsilon(k)_{2^\ell}\rangle \otimes \ket{b'} \right)
    \end{align}
    for matrix element of the operator~$W_{\lsb{\ell}}$  with respect to the basis
    \begin{align}
        \left\{  \Sha_{\kappa,\Delta}^\varepsilon(j) \otimes \ket{b}  \right\}_{(j,b) \in   \mathbb{Z}_{2^\ell} \times \{0,1\} }  \qquad \textrm{of} \qquad \Sha\gkpcode{L,\Delta}{\varepsilon}{2^\ell} \otimes \mathbb{C}^2\, .
    \end{align}
    Then 
\begin{align}
    \label{eq:claim-implement-CX-SHA}
    M_{(j,b),(j,b\oplus j_0)} & \ge  1 - 2 \cdot 2^{2\ell} \Delta^2 - 8 (\Delta/\varepsilon)^4\\
    |M_{(j,b),(j,b\oplus j_0\oplus 1)}| &\le  2 \cdot 2^{2\ell}\Delta^2 + 8 (\Delta/\varepsilon)^4\\
    M_{(j,b),(k,b')} & = 0 \qquad \textrm{otherwise} 
\end{align}
for all~$j,k \in \mathbb{Z}_{2^\ell}$ and~$b,b' \in \{0,1\}$, where~$j_0 = j \mod 2$ is the least significant bit in the binary representation of~$j$.
In particular, the matrix~$M = (M_{(j,b),(k,b')})$ is~$2$-sparse.
    \end{lemma}
    \begin{proof}
        This follows by the same arguments as Lemma~\ref{lem:matrix elements lsb} together with Lemma~\ref{lem:shiftinvariancemomemntumsha}.
    \end{proof}

\subsection{Logical gate errors for rectangular-envelope GKP codes  \label{sec:logicalgatescombstates}}

We find the following bounds on the logical gate error of implementations of basic bit-manipulation maps in the code~$\Sha\gkpcode{L,\Delta}{\varepsilon}{2^\ell}$.

\begin{lemma}[Implementation of 
   ~$\qCX{\ell}$,~$\lsb{\ell}$ and~$\embedmap_\ell$ in rectangular-envelope GKP codes] \label{lem:bitmanipulationSha}Let~$\ell \in \mathbb{N}$,~$\Delta >0$,~$L \in 2\mathbb{N}$ and~$\varepsilon \in(0, 2^{-(\ell+1)}]$. 
    Then we have the following:
    \begin{enumerate}[i)]
        \item The physical implementation of the unitary~$\qCX{\ell}$ on~$ \mathbb{C}^{2^\ell} \otimes \mathbb{C}^2$given by 
        \begin{align}
            W_{\qCX{\ell}} = \mathsf{ctrl} e^{-i\sqrt{2\pi/2^{\ell}}P}
        \end{align} can be realized by a single operation in~$\Uelem^{1,1}$ and satisfies 
        \begin{align}
            \gateerror_{  \Sha\gkpcode{L,\Delta}{\varepsilon}{2^\ell} \otimes \mathbb{C}^{2} }(W_{\qCX{\ell}}, \qCX{\ell}) \le 12 L^{-1/2}\, .
        \end{align}
        \item The physical implementation of the unitary~$\lsb{\ell}$ on~$ \mathbb{C}^{2^\ell} \otimes \mathbb{C}^2$ given by 
        \begin{align}
            W_{\lsb{\ell}}&=(I\otimes H)(M_{\alpha'}^\dagger \otimes I)^n \mathsf{ctrl}e^{iQ}(M_{\alpha'}\otimes I)^n (I\otimes H) 
        \end{align} where 
        \begin{align}
        \alpha'&=e^{\frac{\log \alpha}{\lceil |\log \alpha|\rceil}}\qquad\textrm{ and }\qquad n=\lceil |\log \alpha|\rceil\qquad\textrm{ with }\qquad \alpha=\sqrt{\pi}2^{(\ell-1)/2} 
        \end{align}
         defines
         a circuit~$W_{\lsb{\ell}}=W_T\cdots W_1$ composed of~$T \le \ell + 6$~elementary gates in~$\Uelem^{1,1}$ such that 
         \begin{align}
            \gateerror_{ \Sha\gkpcode{L,\Delta}{\varepsilon}{2^\ell}\otimes \mathbb{C}^{2} }( W_{\lsb{\ell}}, \lsb{\ell}) \le 16 \cdot  2^\ell \Delta + 32 \cdot (\Delta/\varepsilon)^2 \ .
         \end{align}
        \item The isometry~$\embedmap_\ell: \mathbb{C}^{2^\ell} \rightarrow \mathbb{C}^{2^{\ell+1}}$ can be implemented by a unitary~$W_{\embedmap_\ell} = M_{\sqrt{2}}$ acting on~$L^2(\mathbb{R})$ using one elementary operation in~$\Uelem^{1,0}$ such that 
        \begin{align}
            \gateerror_{\Sha\gkpcode{L,\Delta}{\varepsilon}{2^\ell}, \Sha\gkpcode{L,\Delta}{\varepsilon}{2^{\ell+1}}} \left(W_{\embedmap_\ell}, \embedmap_\ell\right) = 0\, .
        \end{align}
    \end{enumerate}
    \end{lemma}

\begin{proof}
Comparing Lemmas~\ref{lem: matrix elements gen CZm} and~\ref{lem:matrix elements lsb} with Lemmas~\ref{lem:matrixelementCXsha} and~\ref{lem:matrixelementlsbSha} respectively
we observe that the proof follows similarly to that of  Lemma~\ref{lem: universal bound special} where we replace the parameter~$\kappa$ with~$\sqrt{2/L}$.
\end{proof}

    With the necessary modifications, the proofs from Section~\ref{sec:gaterrorbitCliffordmain} can be applied. This leads to the following result concerning the implementation of multi-qubit circuits using rectangular-envelope GKP codes.
    It is the counterpart of Theorem~\ref{thm:main} for (regular) symmetrically squeezed GKP codes.

    \begin{corollary}[Implementations of multi-qubit circuits in rectangular-envelope GKP codes]\label{cor:multiqubitrectangular}
        Let~$\ell \in \mathbb{N}$ and~$\Delta \in (0,2^{-\ell})$. Define~$L_{\Delta,2^\ell} = 2^{2(\lceil\log_2 1/\Delta\rceil - \ell)}$. 
        Let~$U = U_T \cdots U_1$ be a unitary on~$(\mathbb{C}^{2})^{\otimes \ell }$ 
    where each unitary~$U_t$ acts non-trivially only on (any) two qubits  for~$t \in \{1,\dots,T\}$.
 There exists a physical unitary implementation~$W_U$ using at most~\blue{$340\ell^2 T$} elementary operations in~$\Uelem^{2,3}$ and a universal constant~$C>0$ (independent of~$\ell$) such that 
\begin{align}
    \gateerror_\cL\left( W_U, J_{\ell} U \left(J_{\ell}^{-1}\right)\right) \le 600 \cdot 2^{2\ell} \cdot T  \Delta\, ,
\end{align}
where the code space is given by~$\cL = \Sha\gkpcode{\Delta}{\star}{2^\ell}\otimes \mathbb{C}\left(\ket{\Sha_{L_{\Delta,2^\ell}, \Delta}^{2^{-(\ell+1)}}(0)_2}\otimes \ket{0}^{\otimes 3}\right)$.
\end{corollary}
\begin{proof} 
Consider a  two-qubit unitary~$U$ on~$(\mathbb{C}^2)^{\otimes\ell}$. 
    Following the same arguments as in the proof of Theorem~\ref{lem:uniformboundgeneralmultiqubit}
    (treating a single two-qubit unitary) and Theorem~\ref{thm:main} (treating the case of a circuit)
    we can construct an implementation~$W_U$ 
    using at most~$340\ell^2 T$ elementary operations. To obtain an upper bound on the logical gate error, we use  Lemma~\ref{lem:bitmanipulationSha} (for rectangular-envelope GKP codes) and  find 
    \begin{align}
        \gateerror_\cL( W_U, J_{\ell} U J_{\ell}^{-1}) \le \blue{600}\ell\cdot  2^\ell \cdot T  \Delta  \, ,
    \end{align}
    where we used the definition \eqref{eq:defbigL} of $L_{\Delta,2^\ell}$.
    The claim follows from this because $ \ell \cdot 2^\ell \le 2^{2\ell}$ for all~$\ell \in \mathbb{N}$.
\end{proof}

\section{Generalization to two-mode encodings of~$2\ell$ logical qubits\label{sec:generalizationmultiple}}

The following generalizes Lemma~\ref{lem:general2qubitbitextract} to a scenario where
we have an even number~$2\ell$ of qubits encoded in a bipartite Hilbert space of the form~$(\mathbb{C}^2)^{\otimes \ell}\otimes (\mathbb{C}^2)^{\otimes \ell}$.
 The construction again uses~$4$~auxiliary qubits. For brevity, we write
 \begin{align}
\ket{\vec{x}}&=\ket{x_{\ell-1}}\cdots \ket{x_0}\qquad\textrm{ for }\qquad  \vec{x}=(x_0,\ldots,x_{\ell-1})\in \{0,1\}^\ell\ .
 \end{align}

\begin{lemma}[Bit-transfer-based implementation of multi-qubit gates -- bipartite case] \label{lem:general2qubitbitextracttwo}
Let~$\ell\geq 2$ be an integer. Consider a system of~$2\ell$~``logical'' qubits denoted
\begin{align}
(A_0\cdots A_{\ell-1})(B_0\cdots B_{\ell-1})=(\mathbb{C}^2)^{\otimes \ell}\otimes (\mathbb{C}^2)^{\otimes \ell}\ .
\end{align}
Let~$j,k\in \{0,\ldots,\ell-1\}$ be arbitrary. Let~$U_{A_jB_k}:(\mathbb{C}^2)^{\otimes \ell}\otimes (\mathbb{C}^2)^{\otimes \ell}
\rightarrow (\mathbb{C}^2)^{\otimes \ell}\otimes (\mathbb{C}^2)^{\otimes \ell}$ be a unitary acting non-trivially only on the two qubits~$A_jB_k$.
Let 
    \begin{align}
    S_1S_2Q_1Q_2Q_3Q_4=\mathbb{C}^{2^\ell}\otimes \mathbb{C}^{2^\ell}\otimes (\mathbb{C}^2)^{\otimes 4}
    \end{align}
    be a system consisting of a bipartite~$2^{2\ell}$-dimensional (``code'') space~$S_1S_2\cong (\mathbb{C}^{2^\ell})^{\otimes 2}$ and~$4$ auxiliary qubits~$Q_1\cdots Q_4$. 
        Then there is a circuit~$V_{U_{A_jB_k}}$ on~$S_1S_2Q_1\cdots Q_4$ consisting of fewer than~$\blue{48\ell-16}$ operations belonging to the set~$\cG(\ell)$ and a single two-qubit operation such that 
        \begin{align}
((J_{\ell}^{-1})^{\otimes 2}  \otimes I) V_{U_{A_jB_k}} (J_{\ell}^{\otimes 2} \otimes I)\left((\ket{\vec{x}}\ket{\vec{y}} )\otimes  \ket{0}^{\otimes 4}\right) &= \left( U_{A_jB_k} (\ket{\vec{x}}\ket{\vec{y}})\right)\otimes \ket{0}^{\otimes 4} 
        \end{align} for all~$\vec{x}=(x_0,\ldots,x_{\ell-1}),\vec{y}=(y_0,\ldots,y_{\ell-1})\in \{0,1\}^\ell$.
                 That is, the circuit~$V_{U_{A_jB_k}}$ realizes the unitary~$U_{A_jB_k}$ on the~$2^\ell$ dimensional subspace 
        \begin{align}
            \mathbb{C}^{2^\ell} \otimes \mathbb{C}^{2^\ell}\otimes (\mathbb{C}\ket{0}^{\otimes 4})\subset S_1S_2Q_1\cdots Q_4=\mathbb{C}^{2^\ell} \otimes \mathbb{C}^{2^\ell}\otimes (\mathbb{C}^2)^{\otimes 4}
            \end{align}
            where the four auxiliary qubits are in the state~$\ket{0}$.
\end{lemma}
\begin{proof}
    The construction of~$V_{U_{A_jB_k}}$  follows using the same reasoning 
    as used in the proof of Lemma~\ref{lem:general2qubitbitextract}, but relies on  the circuit identity shown in Fig.~\ref{fig:general2qubitimplementblock}
    instead of the circuit given in Fig.~\ref{fig:general2qubitimplement}
    \begin{figure}[H]
        \centering
        \includegraphics{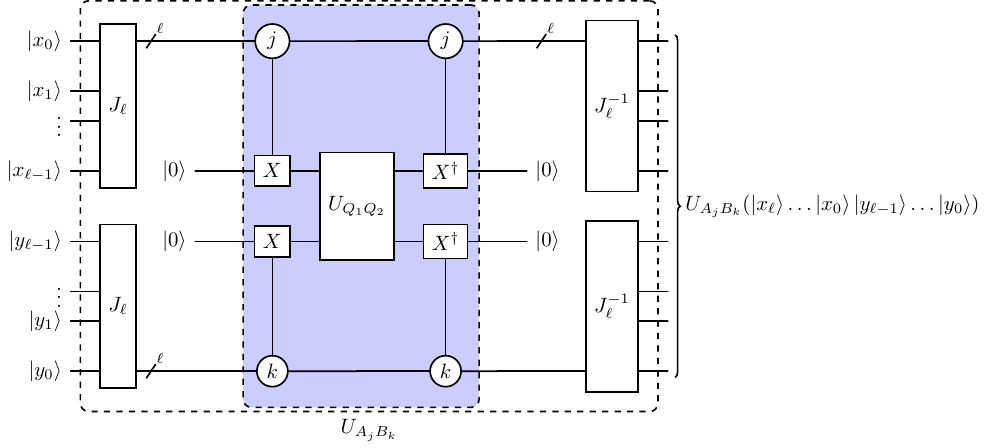}
        \caption{Circuit implementing the 2-qubit unitary~$U_{A_jB_k}$ using the bit-transfer unitary~$\bittransfer{\ell}{j}$ and~$\bittransfer{\ell}{k}$ and two additional qubit systems. 
        \label{fig:general2qubitimplementblock}}
    \end{figure}
\end{proof}

The following two lemmas generalize Lemma~\ref{lem: multiqubitGKPimplement} and Theorem~\ref{lem:uniformboundgeneralmultiqubit}, respectively. Lemma~\ref{lem:multiqubitGKPimplementbipartite} considers~$2\ell$ qubits encoded into two ideal GKP codes~$(\gkpcode{}{}{2^\ell})^{\otimes 2}$. Theorem~\ref{thm: uniform bound general multiqubitbipartite} considers symmetrically squeezed GKP codes~$(\gkpcode{\kappa}{\star}{2^\ell})^{\otimes 2}$ instead of ideal ones. In both cases, the constructions use one auxiliary oscillator and three auxiliary qubits. 

\begin{lemma}[Implementation of multi-qubit unitaries in ideal GKP codes -- bipartite case] \label{lem:multiqubitGKPimplementbipartite} 
Let~$\ell\geq 2$ be an integer. Consider a system of~$2\ell$~``logical'' qubits denoted
\begin{align}
(A_0\cdots A_{\ell-1})(B_0\cdots B_{\ell-1})=(\mathbb{C}^2)^{\otimes \ell}\otimes (\mathbb{C}^2)^{\otimes \ell}\ .
\end{align}
Let~$j,k\in \{0,\ldots,\ell-1\}$ be arbitrary. Let~$U := U_{A_jB_k}:(\mathbb{C}^2)^{\otimes \ell}\otimes (\mathbb{C}^2)^{\otimes \ell}
\rightarrow (\mathbb{C}^2)^{\otimes \ell}\otimes (\mathbb{C}^2)^{\otimes \ell}$ be a unitary acting non-trivially only on the two qubits~$A_jB_k$.
Let 
    \begin{align}
    S_1S_2BQ_1Q_2Q_3 = L^2(\mathbb{R}) \otimes L^2(\mathbb{R}) \otimes \left(L^2(\mathbb{R})\otimes (\mathbb{C}^2)^{\otimes 3}\right)
    \end{align}
    be a system consisting of a bipartite~$2^{2\ell}$-dimensional (``code'') space~$S_1S_2\cong (L^2(\mathbb{R}))^{\otimes 2}$, $1$ auxiliary oscillator~$B$ and~$3$ auxiliary qubits~$Q_1 Q_2 Q_3$. 
    Then there is a unitary circuit~$W_U$  on~$L^2(\mathbb{R})^{\otimes 2}\otimes \left(L^2(\mathbb{R})\otimes (\mathbb{C}^2)^{\otimes 3}\right)$ consisting of  at most~$\blue{340\ell^2}$ elementary operations in~$\Uelem^{3,3}$
            such that
            \begin{align}
            &W_U\left(\ket{\gkp(x)_{2^\ell} \otimes \gkp(y)_{2^\ell}}\otimes \left(\ket{\gkp(0)_2}\otimes\ket{0}^{\otimes 3}\right)\right)\\
            &\qquad\qquad= \encodergkp{}{}{2^\ell}^{\otimes 2}  \left( J_{\ell}^{\otimes 2} U \left(J_{\ell}^{\otimes 2}\right)^{-1}\left(\ket{x} \otimes \ket{y}\right) \right)  \otimes \left(\ket{\gkp(0)_2}\otimes\ket{0}^{\otimes 3}\right)
            \end{align}
            for all~$x,y\in \{0,\ldots,2^\ell-1\}$. 
    In other words, the circuit~$W_U$ implements the gate~$U$ exactly on the subspace
    \begin{align}
        \gkpcode{}{}{2^\ell}^{\otimes 2}\otimes\mathbb{C}(\ket{\Psi}) \subset  L^2(\mathbb{R})^{\otimes 2}\otimes (L^2(\mathbb{R})\otimes (\mathbb{C}^2)^{\otimes 3}) \ ,
    \end{align}
    where~$\ket{\Psi}=\ket{\gkp(0)_2}\otimes\ket{0}^{\otimes 3}\in L^2(\mathbb{R})\otimes(\mathbb{C}^2)^{\otimes 3}$, i.e., the auxiliary oscillator is in the ideal GKP state~$\ket{\gkp(0)_2}$ and the three auxiliary qubits are in the state~$\ket{0}^{\otimes 3}$.
    \end{lemma}

    \begin{proof}
        The construction of~$U_{A_jB_k}$ is analogous to the construction of~$U_{A_jA_k}$ in the proof of Lemma~\ref{lem: multiqubitGKPimplement}. 
It is illustrated in Fig.~\ref{fig:constructionUajbk}.
    \end{proof}
    
\begin{figure}
\centering
\includegraphics[width=\textwidth]{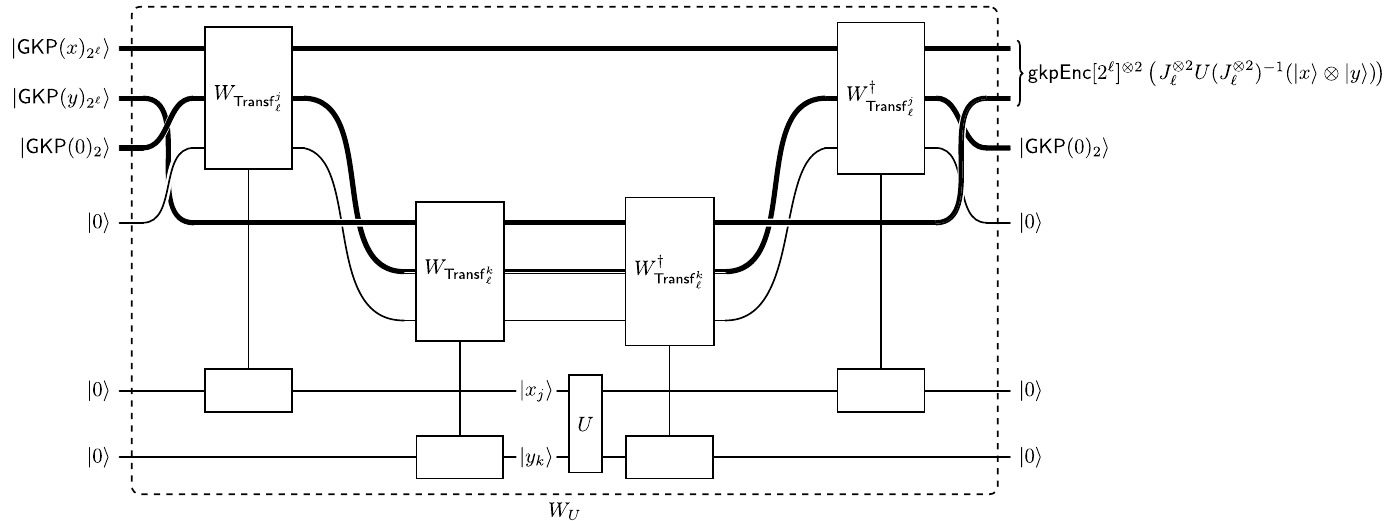}
    \caption{The implementation~$W_U$ of a logical two-qubit unitary~$U=U_{A_jB_k}$,
    acting on the~$j$-th qubit~$A_j$ and~$k$-th qubit~$B_k$ for~$j,k \in \{0,\dots, \ell-1\}$, see Lemma~\ref{lem:multiqubitGKPimplementbipartite}. Here the first~$\ell$ qubits~$A_0\cdots A_{\ell-1}$ are  encoded in a GKP code on the first mode, whereas the remaining qubits~$B_0\cdots B_{\ell-1}$ are encoded in a second mode. In contrast to the construction given in Fig.~\ref{fig:W_Uimplement} where the entire logical information is encoded in a single mode, bit-transfer gates (and their adjoints) are used which act on the respective modes.
         \label{fig:constructionUajbk}}
    \end{figure}

    \newcommand*{\bosonicmode}{\mathsf{B}}
Combining the constructions illustrated in Fig.~\ref{fig:W_Uimplement} and Fig.~\ref{fig:constructionUajbk}, we obtain the following.
\begin{theorem}[Implementation of multi-qubit gates in approximate GKP codes -- bipartite case] \label{thm: uniform bound general multiqubitbipartite}
    Let~$\ell\geq 2$ be an integer. Consider a system of~$2\ell$~``logical'' qubits denoted
$Q_{1}\cdots Q_{2\ell}\cong (\mathbb{C}^2)^{\otimes \ell}$.
Let $\alpha,\beta\in \{1,\ldots,2\ell\}$, $\alpha<\beta$ be arbitrary.
Let $U:=U_{Q_\alpha Q_\beta}:(\mathbb{C}^2)^{\otimes 2\ell}\rightarrow (\mathbb{C}^2)^{\otimes 2\ell}$
be a unitary acting non-trivially only on the pair of qubits $Q_\alpha Q_\beta$. 
     Then there exists a physical unitary implementation~$W_U$ using at most~$\blue{340\ell^2}$ elementary operations in~$\Uelem^{3,3}$ acting on 
    the Hilbert space~$L^2(\mathbb{R})^{\otimes 2} \otimes L^2(\mathbb{R}) \otimes (\mathbb{C}^{2})^{\otimes 3}$ such that the following holds.
    \begin{enumerate}[(i)]
    \item \label{it:glpstateimplementationv}Let~$\kappa \in (0,1/4)$. For the code space
    \begin{align}
    \cL = (\gkpcode{\kappa}{\star}{2^\ell})^{\otimes 2}\otimes \mathbb{C}\left(\ket{\gkp_{\kappa}^{2^{-(\ell+1)}}(0)_2}\otimes \ket{0}^{\otimes 3}\right)\subset L^2(\mathbb{R})^{\otimes 2}\otimes L^2(\mathbb{R})\otimes (\mathbb{C}^2)^{\otimes 3}\ 
    \end{align}
    defined in terms of the symmetrically squeezed GKP-code~$\gkpcode{\kappa}{\star}{2^\ell}$, we have 
        \begin{align}
    \gateerror_\cL\left( W_U, J_{\ell}^{\otimes 2} U \left(J_{\ell}^{\otimes 2}\right)^{-1}\right) \le \blue{400} \ell  \kappa \, . \label{eq:2modegateerrorGaussian}
    \end{align}
    \item\label{it:shastateimplementationv}
    Let $\ell\in \mathbb{N}$ and $\Delta\in (0,2^{-\ell})$.  For  the code space
    \begin{align}
    \cL = (\Sha\gkpcode{\Delta}{\star}{2^\ell})^{\otimes 2}\otimes \mathbb{C}\left(
    \ket{\Sha_{L_{\Delta,2^\ell}, \Delta}^{2^{-(\ell+1)}}(0)_2}\otimes \ket{0}^{\otimes 3}\right)\subset L^2(\mathbb{R})^{\otimes 2}\otimes L^2(\mathbb{R})\otimes (\mathbb{C}^2)^{\otimes 3}\ 
    \end{align}
    defined in terms of the rectangular-envelope GKP-code~$\Sha\gkpcode{\Delta}{\star}{2^\ell}$
    and $L_{\Delta,2^\ell} = 2^{2(\lceil\log_2 1/\Delta\rceil - \ell)}$, we have
        \begin{align}
    \gateerror_\cL\left( W_U, J_{\ell}^{\otimes 2} U \left(J_{\ell}^{\otimes 2}\right)^{-1}\right) \le \blue{600}\cdot 2^{2\ell}\cdot \Delta\ . \label{eq:gaterrorboundrectangle}
    \end{align}
    \end{enumerate}
    
\end{theorem}
At first sight, the bounds on the logical gate error from Theorem~\ref{thm: uniform bound general multiqubitbipartite} suggest a different scaling for symmetrically squeezed Gaussian-envelope and rectangular-envelope GKP codes. This is not the case.
We note that the bound for the logical gate error of an implementation using symmetrically squeezed Gaussian-envelope GKP-codes in Eq.~\eqref{eq:2modegateerrorGaussian} requires $\kappa = O(1/\mathsf{poly}(\ell))$ in order to be non-trivial. When expressed in terms of the corresponding parameter $\Delta = \kappa/{(2\pi \cdot 2^\ell)}$, this implies that $\Delta = 1/\exp(\Omega(\ell))$.
This is in accordance with the scaling suggested in Eq.~\eqref{eq:gaterrorboundrectangle} for rectangular-envelope GKP codes.

We note that -- again by the composability of the logical gate error, Theorem~\ref{thm: uniform bound general multiqubitbipartite} implies an analogous statement for an (tensor product) encoding of $m\ell$ qubits into a code space of the form
\begin{align}
    \cL = (\gkpcode{\kappa}{\star}{2^\ell})^{\otimes m}\otimes \mathbb{C}\left(\ket{\gkp_{\kappa}^{2^{-(\ell+1)}}(0)_2}\otimes \ket{0}^{\otimes 3}\right)\subset L^2(\mathbb{R})^{\otimes m}\otimes L^2(\mathbb{R})\otimes (\mathbb{C}^2)^{\otimes 3}\ 
\end{align}
respectively
    \begin{align}
    \cL = (\Sha\gkpcode{\Delta}{\star}{2^\ell})^{\otimes m}\otimes \mathbb{C}\left(
    \ket{\Sha_{L_{\Delta,2^\ell}, \Delta}^{2^{-(\ell+1)}}(0)_2}\otimes \ket{0}^{\otimes 3}\right)\subset L^2(\mathbb{R})^{\otimes m}\otimes L^2(\mathbb{R})\otimes (\mathbb{C}^2)^{\otimes 3}\ \ .
    \end{align}
That is, any two-qubit unitary~$U$ acting on any pair of  logical qubits 
can be implemented with $340\ell^2$ elementary operations and a logical gate error bounded 
as stated in Theorem~\ref{thm: uniform bound general multiqubitbipartite}.

\begin{proof}
The proof of statement~\eqref{it:shastateimplementationv} is analogous to the proof of~\eqref{it:glpstateimplementationv}
 and relies on Corollary~\ref{cor:multiqubitrectangular}.  We thus focus on the proof of Claim~\eqref{it:glpstateimplementationv}.

Let us denote the three bosonic modes by $\bosonicmode_1$, $\bosonicmode_2$ and $\bosonicmode_{\mathsf{aux}}$, respectively. 
We note that $\gkpcode{\kappa}{\star}{2^\ell}^{\otimes 2}\cong (\mathbb{C}^2)^{\otimes 2\ell}$, hence the code space~$\cL$ indeed encodes $2\ell$ qubits. 
Let us write 
\begin{align}
Q_{1}\cdots Q_{2\ell}:=(A_0\cdots A_{\ell-1})(B_0\cdots B_{\ell-1})=(\mathbb{C}^2)^{\otimes \ell}\otimes (\mathbb{C}^2)^{\otimes \ell}\ .
\end{align}
where the qubits $A_0\cdots A_{\ell-1}$ are encoded in the code space~$\gkpcode{\kappa}{\star}{2^\ell})$ of the first mode~$\bosonicmode_1$, whereas the qubits $B_0\cdots B_{\ell-1}$ are encoded analogously in the second mode~$\bosonicmode_2$.

Consider a pair $Q_\alpha Q_\beta$ of (logical) qubits, where $\alpha<\beta$. We distinguish between the following three cases:
\begin{enumerate}[(i)]
\item
$Q_\alpha Q_\beta=A_jA_k$ for some $j,k\in \{0,\ldots,\ell-1\}$, $j<k$:
Here the (logical) gate~$U$  acts exclusively on qubits encoded in the first mode. We can thus use Theorem~\ref{lem:uniformboundgeneralmultiqubit}, i.e., the implementation~$W_U$ illustrated  in  Fig.~\ref{fig:W_Uimplement}. 
We apply this implementation to the first mode~$\bosonicmode_1$, the auxiliary bosonic mode~$\bosonicmode_{\mathsf{aux}}$ and the three qubits (while leaving the second mode~$\bosonicmode_2$ untouched). 
By Theorem~\ref{lem:uniformboundgeneralmultiqubit}, this implementation
uses at most $340\ell^2$ elementary operations belonging to~$\Uelem^{3,3}$, and has a gate error upper bounded by~$400\ell\kappa$. The claim follows from this because the logical gate error is composable, i.e., stable under tensoring in an additional system (the bosonic mode~$\bosonicmode_2$). 

\item 
$Q_\alpha Q_\beta=B_jB_k$ for some $j,k\in \{0,\ldots,\ell-1\}$, $j<k$:
Here we can proceed analogously, applying the construction of Theorem~\ref{lem:uniformboundgeneralmultiqubit} 
(see Fig.~\ref{fig:W_Uimplement}) to the second mode~$\bosonicmode_2$ instead, while leaving the mode~$\bosonicmode_1$  untouched. 
\item
$Q_\alpha Q_\beta=A_jB_k$ for some $j,k\in \{0,\ldots,\ell-1\}$, $j<k$: Here we have to implement a two-qubit unitary~$U$ where the first qubit is encoded in the mode~$\bosonicmode_1$, whereas the second qubit is encoded in the mode~$\bosonicmode_2$. We can use the construction 
from Lemma~\ref{lem:multiqubitGKPimplementbipartite}, see Fig.~\ref{fig:constructionUajbk}.  
The proof of the corresponding upper bound on the gate error  proceeds analogously as the proof of  Theorem~\ref{lem:uniformboundgeneralmultiqubit} and is omitted here.
\end{enumerate}
\end{proof}

\bibliographystyle{unsrturl}
\bibliography{q}

\begin{thebibliography}{10}

\bibitem{cliffordslinearoptics2025}
Lukas Brenner, Beatriz Dias, and Robert Koenig.
\newblock Composable logical gate error in approximate quantum error correction: reexamining gate implementations in {G}ottesman-{K}itaev-{P}reskill codes, Sep 2025.
\newblock \href {https://arxiv.org/abs/2509.14658} {\path{arXiv:2509.14658}}.

\bibitem{PhysRevLett.82.2009}
Konrad Banaszek and Krzysztof W\'odkiewicz.
\newblock Testing quantum nonlocality in phase space.
\newblock {\em Phys. Rev. Lett.}, 82:2009--2013, Mar 1999.
\newblock \href {https://doi.org/10.1103/PhysRevLett.82.2009} {\path{doi:10.1103/PhysRevLett.82.2009}}.

\bibitem{PhysRevA.67.012105}
J\'er\^ome Wenger, Mohammad Hafezi, Fr\'ed\'eric Grosshans, Rosa Tualle-Brouri, and Philippe Grangier.
\newblock Maximal violation of {B}ell inequalities using continuous-variable measurements.
\newblock {\em Phys. Rev. A}, 67:012105, Jan 2003.
\newblock \href {https://doi.org/10.1103/PhysRevA.67.012105} {\path{doi:10.1103/PhysRevA.67.012105}}.

\bibitem{Etesse_2014}
Jean Etesse, Rémi Blandino, Bhaskar Kanseri, and Rosa Tualle-Brouri.
\newblock Proposal for a loophole-free violation of {B}ell's inequalities with a set of single photons and homodyne measurements.
\newblock {\em New Journal of Physics}, 16(5):053001, May 2014.
\newblock \href {https://doi.org/10.1088/1367-2630/16/5/053001} {\path{doi:10.1088/1367-2630/16/5/053001}}.

\bibitem{PhysRevLett.97.110501}
Nicolas~C. Menicucci, Peter van Loock, Mile Gu, Christian Weedbrook, Timothy~C. Ralph, and Michael~A. Nielsen.
\newblock Universal quantum computation with continuous-variable cluster states.
\newblock {\em Phys. Rev. Lett.}, 97:110501, Sep 2006.
\newblock \href {https://doi.org/10.1103/PhysRevLett.97.110501} {\path{doi:10.1103/PhysRevLett.97.110501}}.

\bibitem{PhysRevA.79.062318}
Mile Gu, Christian Weedbrook, Nicolas~C. Menicucci, Timothy~C. Ralph, and Peter van Loock.
\newblock Quantum computing with continuous-variable clusters.
\newblock {\em Phys. Rev. A}, 79:062318, Jun 2009.
\newblock \href {https://doi.org/10.1103/PhysRevA.79.062318} {\path{doi:10.1103/PhysRevA.79.062318}}.

\bibitem{RevModPhys.84.621}
Christian Weedbrook, Stefano Pirandola, Ra\'ul Garc\'{\i}a-Patr\'on, Nicolas~J. Cerf, Timothy~C. Ralph, Jeffrey~H. Shapiro, and Seth Lloyd.
\newblock {G}aussian quantum information.
\newblock {\em Rev. Mod. Phys.}, 84:621--669, May 2012.
\newblock \href {https://doi.org/10.1103/RevModPhys.84.621} {\path{doi:10.1103/RevModPhys.84.621}}.

\bibitem{PhysRevLett.102.120501}
Julien Niset, Jarom\'{\i}r Fiur\'a\ifmmode~\check{s}\else \v{s}\fi{}ek, and Nicolas~J. Cerf.
\newblock No-go theorem for {G}aussian quantum error correction.
\newblock {\em Phys. Rev. Lett.}, 102:120501, Mar 2009.
\newblock \href {https://doi.org/10.1103/PhysRevLett.102.120501} {\path{doi:10.1103/PhysRevLett.102.120501}}.

\bibitem{PhysRevA.99.032344}
Christophe Vuillot, Hamed Asasi, Yang Wang, Leonid~P. Pryadko, and Barbara~M. Terhal.
\newblock Quantum error correction with the toric {G}ottesman-{K}itaev-{P}reskill code.
\newblock {\em Phys. Rev. A}, 99:032344, Mar 2019.
\newblock \href {https://doi.org/10.1103/PhysRevA.99.032344} {\path{doi:10.1103/PhysRevA.99.032344}}.

\bibitem{PhysRevLett.89.137903}
Jens Eisert, Stefan Scheel, and Martin~B. Plenio.
\newblock Distilling {G}aussian states with {G}aussian operations is impossible.
\newblock {\em Phys. Rev. Lett.}, 89:137903, Sep 2002.
\newblock \href {https://doi.org/10.1103/PhysRevLett.89.137903} {\path{doi:10.1103/PhysRevLett.89.137903}}.

\bibitem{PhysRevLett.89.137904}
Jarom\'{\i}r Fiur\'a\ifmmode~\check{s}\else \v{s}\fi{}ek.
\newblock {G}aussian transformations and distillation of entangled {G}aussian states.
\newblock {\em Phys. Rev. Lett.}, 89:137904, Sep 2002.
\newblock \href {https://doi.org/10.1103/PhysRevLett.89.137904} {\path{doi:10.1103/PhysRevLett.89.137904}}.

\bibitem{PhysRevA.66.032316}
G\'eza Giedke and J.~Ignacio~Cirac.
\newblock Characterization of {G}aussian operations and distillation of {G}aussian states.
\newblock {\em Phys. Rev. A}, 66:032316, Sep 2002.
\newblock \href {https://doi.org/10.1103/PhysRevA.66.032316} {\path{doi:10.1103/PhysRevA.66.032316}}.

\bibitem{gkp}
Daniel Gottesman, Alexei Kitaev, and John Preskill.
\newblock Encoding a qubit in an oscillator.
\newblock {\em Phys. Rev. A}, 64:012310, Jun 2001.
\newblock \href {https://doi.org/10.1103/PhysRevA.64.012310} {\path{doi:10.1103/PhysRevA.64.012310}}.

\bibitem{brenner2024complexity}
Lukas Brenner, Libor Caha, Xavier Coiteux-Roy, and Robert Koenig.
\newblock The complexity of {G}ottesman-{K}itaev-{P}reskill states, Oct 2024.
\newblock \href {https://arxiv.org/abs/2410.19610} {\path{arXiv:2410.19610}}.

\bibitem{liu2024hybridoscillatorqubitquantumprocessors}
Yuan Liu, Shraddha Singh, Kevin~C. Smith, Eleanor Crane, John~M. Martyn, Alec Eickbusch, Alexander Schuckert, Richard~D. Li, Jasmine Sinanan-Singh, Micheline~B. Soley, Takahiro Tsunoda, Isaac~L. Chuang, Nathan Wiebe, and Steven~M. Girvin.
\newblock Hybrid oscillator-qubit quantum processors: Instruction set architectures, abstract machine models, and applications, Jul 2024.
\newblock \href {https://arxiv.org/abs/2407.10381} {\path{arXiv:2407.10381}}.

\bibitem{brenner2024factoring}
Lukas Brenner, Libor Caha, Xavier Coiteux-Roy, and Robert Koenig.
\newblock Factoring an integer with three oscillators and a qubit, Dec 2024.
\newblock \href {https://arxiv.org/abs/2412.13164} {\path{arXiv:2412.13164}}.

\bibitem{Nielsen_Chuang_2010}
Michael~A. Nielsen and Isaac~L. Chuang.
\newblock {\em Quantum Computation and Quantum Information: 10th Anniversary Edition}.
\newblock Cambridge University Press, 2010.
\newblock \href {https://doi.org/10.1017/CBO9780511976667} {\path{doi:10.1017/CBO9780511976667}}.

\end{thebibliography}
\end{document}